%% file: mainfossacsarxiv.tex
\documentclass[runningheads]{llncs}

\usepackage[T1]{fontenc}

\usepackage{graphicx}

\input xy 
\xyoption{all}
\input{preamble}
\input{macros}

\begin{document}

\title{Tapes as Stochastic Matrices of String Diagrams} 

\author{Filippo Bonchi\inst{} \and Cipriano Junior Cioffo\inst{} \thanks{This research was partly funded by the Advanced Research + Invention Agency (ARIA) Safeguarded AI Programme. Bonchi is supported by the Ministero dell'Università e della Ricerca of Italy grant PRIN 2022 PNRR No. P2022HXNSC - RAP (Resource Awareness in Programming). This study was carried out within the National Centre on HPC, Big Data and Quantum Computing - SPOKE 10 (Quantum Computing) and received funding from the European Union Next-GenerationEU - National Recovery and Resilience Plan (NRRP) – MISSION 4 COMPONENT 2, INVESTMENT N. 1.4 – CUP N. I53C22000690001.}
}



\authorrunning{F. Bonchi, C.J. Cioffo} 

\institute{University of Pisa}

\maketitle

\begin{abstract}
Tape diagrams provide a graphical notation for categories equipped with two monoidal products, $\otimes$ and $\oplus$, where $\oplus$ is a biproduct. Recently, they have been generalised to handle Kleisli categories of arbitrary monoidal monads. In this work, we show that for the subdistribution monad, tapes are isomorphic to stochastic matrices of subdistributions of string diagrams. We then exploit this result to provide a complete axiomatisation of probabilistic Boolean circuits.

\end{abstract}

\input{sections/intro}

\input{sections/pca}
\input{sections/cbproducts}


\input{sections/cmatrix}

\input{sections/syntactic}

\input{sections/probtapes}

\input{sections/probbooltapes}

\input{sections/conclusion}

\bibliographystyle{splncs04}
\bibliography{references}

\appendix

\input{appendices/coherences}

\input{appendices/appendicerefusiprimaversione.tex}

\input{appendices/appconvbiprodcat.tex}

\input{appendices/appcmatrices.tex}
\input{appendices/apptc.tex}

\input{appendices/appprobbooltape.tex}

\end{document}

%% file: preamble.tex
\usepackage{array}
\usepackage{adjustbox}
\usepackage{tikz-cd}
\usepackage{tikz}
\usetikzlibrary{cd,arrows}
\usepackage{mathpartir}
\usepackage{caption}
\usepackage{subcaption}
\usepackage{amsfonts}
\usepackage{amsmath}
\usepackage[renew-matrix,renew-dots]{nicematrix}
\usepackage{float}
\usepackage{enumitem}
\usepackage{xcolor}
\usepackage{booktabs}
\usepackage{rotating}
\usepackage{multirow}
\usepackage{comment}
\usepackage{graphicx}
\usepackage{braket}
\usepackage{enumitem}
\usepackage{xparse}
\usepackage{xstring}
\usepackage{mathtools}
\usepackage{xifthen}
\usepackage{makecell}
\usepackage{txfonts}
\usepackage{hyperref}
\usepackage{cleveref}
\hypersetup{
    colorlinks,
    linkcolor={black!50!black},
    citecolor={blue!50!black},
    urlcolor={blue!80!black}
}
\usepackage{mathtools}
\usepackage{changepage}

\usepackage{quiver}

\usepackage{todonotes}

\usepackage{tabularx}

\usepackage{wrapfig}

%% file: sections/intro.tex
\section{Introduction}\label{sec:intro}

Driven by the growing interest in compositional semantics for probabilistic systems, an increasing body of work~\cite{cho2019disintegration,jacobs2019causal,piedeleu2025boolean,DBLP:journals/corr/abs-2410-10627,DBLP:journals/corr/abs-2502-03477,DBLP:journals/corr/abs-2301-12989,jacobs2021logical,moss2023category,fritz2021finetti,fritz2023dilations,fritz2018bimonoidal,stein2024probabilistic,perrone2023markov,fritz2023d,fritz2023weakly} employs \emph{string diagrams}~\cite{joyal1991geometry,selinger2010survey}, which formally represent morphisms in a strict symmetric monoidal category freely generated by a monoidal signature. These diagrams are typically interpreted in \( \KlD \), the Kleisli category of the subdistribution monad \( \Dis \) (or suitable variants), with the monoidal structure induced by the cartesian product \( \otimes \).
Much research has focused on the role of the \emph{copier} \( \CBcopier \) and the \emph{discharger} \( \CBdischarger \), and their correspondence to marginals and joint distributions.%

In contrast, significantly less attention \cite{fritz2009presentation} has been devoted to the second monoidal tensor \( \oplus \) in $\KlD$, which corresponds to disjoint union of sets, even though its interaction with \( \otimes \) plays a pivotal role in many probabilistic frameworks~\cite{stein2024probabilistic,introductioneffectus,DBLP:journals/pacmpl/LiellCockS25}. One possible reason is that standard string diagrams are inherently limited to representing a single monoidal structure at a time.  To overcome this limitation, we focus on \emph{tape diagrams}~\cite{bonchi2023deconstructing}, which can be informally understood as string diagrams of string diagrams. This more expressive graphical language supports simultaneous reasoning about both \( \otimes \) and \( \oplus \), in particular in categories where \( \oplus \) forms a \emph{biproduct}, i.e., serves as both categorical product and coproduct. Moreover, as recently shown~\cite{bonchi2025tapediagramsmonoidalmonads}, tape diagrams extend naturally to Kleisli categories of arbitrary monoidal monads and thus, in particular, to $\KlD$, which is the main focus of this work.

\medskip

Our starting observation is that the coproduct \( \oplus \) in \( \KlD \) enjoys an additional universal property, akin to that of a product, which makes it resemble a biproduct. We refer to categories where \( \oplus \) satisfies this property as \emph{convex biproduct categories} (Definition~\ref{def:convbicat}), and we study their associated monoidal algebras. While categories with finite biproducts are equipped with natural monoids and \emph{comonoids}, convex biproduct categories carry natural monoids and \emph{co-pointed convex algebras}, which are the dual of pointed convex algebras (pca)--the Eilenberg-Moore algebras for the subdistribution monad~\cite{stone1949postulates,bonchi2017power,DBLP:journals/lmcs/BonchiSV22,DBLP:conf/lics/MioSV21}.

There is a further striking analogy with the theory of finite biproduct categories: given any category \( \Cat{C} \) enriched over monoids, one can construct the category \( \Cat{Mat}(\Cat{C}) \) of matrices over \( \Cat{C} \)~\cite{mac_lane_categories_1978}, which possesses finite biproducts. Analogously, starting from a category \( \Cat{C} \) enriched over pcas, one can define the convex biproduct category \( \stmat{\Cat{C}} \) of \emph{stochastic matrices} over \( \Cat{C} \) (Proposition~\ref{prop:cmatrixcb}). Composition is defined via matrix multiplication: sums are given by the convex sums of the pca-enrichment and products by composition of arrows in $\Cat{C}$.
We show that this construction yields a functor $\stmat{-}$ from the category $\Cat{PCACat}$ of pca-enriched categories to the category $\Cat{CBCat}$ of convex biproduct categories, which is left adjoint to the forgetful functor \( U \) (Theorem~\ref{thm:matfree}):

\[
\xymatrix{
\Cat{Cat} \ar@/^/[rr]^{(-)^+}& \;\;{\tiny{\bot}} & \Cat{PCACat}  \ar@/^/[ll]^{U} \ar@/^/[rr]^{\stmat{-}} &\;{\tiny{\bot}}  & \Cat{CBCat} \ar@/^/[ll]^{U}
}
\]
The leftmost adjunction is standard (see, e.g.,~\cite{borceux2,villoria2024enriching}): the functor \( (-)^+ \) freely enriches a locally small category \( \Cat{C} \) over pcas. The resulting category \( \Cat{C}^+ \) has arrows given by subdistributions of arrows in \( \Cat{C} \).

Our key result is that the composite of the two adjunctions admits a syntactic presentation in terms of generators and equations. Specifically, we define a functor \( \CatT{-} \colon \Cat{Cat} \to \Cat{CBCat} \) that freely adds to each category \( \Cat{C} \) a monoidal structure with natural monoids, and co-pointed convex algebras. We show that \( \CatT{-} \) is left adjoint to the forgetful functor \( U \colon \Cat{CBCat} \to \Cat{Cat} \) (Theorem~\ref{thm:syntacticadjunction}) and we thus conclude that $\CatT{\Cat{C}}$ is isomorphic to $\stmat{\Cat{C}^+}$ (Corollary \ref{cor:isotapematrices}). 


Importantly, when \( \Cat{C} \) is a category of string diagrams \( \DiagS \), \( \CatT{\DiagS} \) coincides with the category of tape diagrams for the monad \( \Dis \) introduced in~\cite{bonchi2025tapediagramsmonoidalmonads}. A result from~\cite{bonchi2025tapediagramsmonoidalmonads} shows that \( \CatT{\DiagS} \) admits two monoidal products, \( \oplus \) and \( \otimes \), forming a \emph{rig category} (also known as a bimonoidal category)~\cite{laplaza_coherence_1972,johnson2021bimonoidal}. The isomorphism \( \CatT{\DiagS} \cong \stmat{\DiagS^+} \) then provides a concrete interpretation of tapes as stochastic matrices whose entries are probability distributions of string diagrams. Interestingly, $\piu$ is direct sum of matrices and $\otimes$ is (an extended) Kronecker product.

We argue that this isomorphism can be fruitfully applied to give more principled axiomatisations of probabilistic languages based on string diagrams, and to streamline their completeness proofs. As a case study, we revisit the complete axiomatisation of \emph{probabilistic Boolean circuits} from~\cite{piedeleu2025boolean}. These can be encoded into tape diagrams of Boolean circuits resulting in a more effective realisation of probabilistic control (see \cite[Example 30]{bonchi2025tapediagramsmonoidalmonads}). We show that (Corollary \ref{cor:final}) adding a single extra axiom--originally appearing in~\cite{Niels}-- to the axioms of Boolean circuits and the tape structure (i.e., natural monoids and co-pcas) suffices to derive an alternative complete axiomatisation. The proof of completeness amounts to several categorical observations, two almost trivial lemmas and the universal properties of convex biproducts.


\emph{Synopsis.} We begin in Section~\ref{sec:pca} by recalling the notions of pointed convex algebras, pca-enriched categories, and the leftmost adjunction in the diagram above. Section~\ref{sec:cbproducts} introduces convex biproduct categories and their associated monoidal algebraic structures. In Section~\ref{sec:cmatrix}, we define stochastic matrices and establish the rightmost adjunction. The syntactic construction \( \CatT{-} \), corresponding to the composite adjunction, is presented in Section~\ref{sec:syntactic}. This construction is then applied to categories of string diagrams in Section~\ref{sec:tapediagrams} to obtain tape diagrams. Finally, Section~\ref{sec:probboolcircuits} revisits probabilistic Boolean circuits from~\cite{piedeleu2025boolean}, explains their encoding as probabilistic tapes from~\cite{bonchi2025tapediagramsmonoidalmonads}, and presents the alternative axiomatisation along with the completeness proof.

%% file: sections/pca.tex
\section{Pointed Convex Algebras}\label{sec:pca} 
We commence our exposition by recalling several well known algebraic structures. 

A \emph{pointed convex algebra} (pca) \cite{stone1949postulates}, consists of a set $X$, a designed element $\star \in X$ and, for all $p$ in the open real interval $(0,1)$, a function $+_p \colon X\times X \to X$ such that, by fixing $\tilde{p}\defeq pq$ and $\tilde{q}\defeq \frac{p(1-q)}{1-pq}$, the following laws hold for all $x_1,x_2,x_3\in X$. 
\begin{equation}\label{eq:pca}(x_1+_q x_2)+_p x_3 =  x_1 +_{\tilde{p}} (x_2+_{\tilde{q}} x_3) \qquad x_1+_px_2=x_2+_{1-p}x_1 \qquad  x_1+_p x_1 = x_1\end{equation}
A morphism of pcas is a function preserving $\star$ and $+_p$. 
We denote by $\Cat{PCA}$ the category of PCAs and their morphisms.  

In any pca, $+_p$ can be defined for all $p\in[0,1]$ by setting $x+_1 y \defeq x$ and $x+_0 y \defeq y$. For all $n\in \mathbb{N}$, $p_1, \dots, p_n \in [0,1]$ such that $\sum_i p_i\leq 1$, one defines $\sum_{i=1}^n p_i\cdot (-)_i \colon X^n \to X$ inductively as
%
\begin{equation}\label{eq:def somma n pca}
    \sum_{i=1}^0 p_i\cdot (-)_i \defeq \star \qquad \sum_{i=1}^{n+1} p_i\cdot (-)_i \defeq (-)_{1} +_{p_{1}} \sum_{j=1}^n q_j \cdot (-)_j\text{.}
\end{equation}
where $(-)_j=(-)_{i+1}$ and $q_j$ is $\frac{p_{i+1}}{1-p_1}$ if $p_1 \neq 1$ and $0$ otherwise.
Note that when $n=1$, $\sum_{i=1}^1p_i \cdot (-)_i$ is, by definition, $(-)_1 +_{p_1} \star$. Since this operation will play a crucial role we name it \emph{multiplication by a scalar} $p\in [0,1]$ and we fix the following notation. \[p\cdot x \defeq x+_p \star\] 

The standard example of pca is provided by the set $\Dis(X)$ of finitely supported probability subdistributions over some set $X$. Recall that a finitely supported subdistribution over  $X$ is a function $d\colon X \to [0,1]$ such that $\sum_{x\in X}d(x)\leq1$ and $d(x)\neq 0$ for finitely many $x$. The designed element $\star$ in $\Dis(X)$ is the null subdistribution, i.e., $\star(x)\defeq 0$ for all $x\in X$; given $d_1,d_2\in \Dis(X)$, for all $p\in(0,1)$, $(d_1+_pd_2)(x) \defeq p\cdot d_1(x)+ (1-p)\cdot d_2(x)$. At this point, it is convenient to fix some extra notation: $\DisF(X)$ is the set of all finitely supported distributions (i.e., $\sum_{x\in X}d(x)=1$) and, for all $x\in X$, $\delta_x$ is the Dirac distribution at $x$ (i.e., $\delta_x(x')=1$ if $x=x'$ and $0$ otherwise).

\medskip

A category $\Cat{C}$ is \emph{$\Cat{PCA}$-enriched} if every homset $\Cat{C}[X,Y]$ carries a pca and composition of arrows is a pcas morphism, i.e.,  for all $p\in(0,1)$ and properly typed arrows $e,f,g,h$, the followings hold.\begin{equation}\label{eq:enr}e; (f+_p g) = (e;f)+_p (e;g) \qquad (f+_pg) ;h= (f;h +_p g;h) \qquad f;\star =\star= \star;f\end{equation}A functor $F$ is $\Cat{PCA}$-enriched if it preserves the structure of pca over each homset. $\Cat{PCA}$-enriched categories and functors form a category, denoted by $\Cat{PCACat}$.
All categories considered hereafter are tacitly assumed to be locally small: $\Cat{Cat}$ stands  for the category of locally small categories. Note that there is a forgetful functor $U\colon \Cat{PCACat} \to \Cat{Cat}$.

Any category $\Cat{C}$, gives rise to the $\Cat{PCA}$-enriched category $\Cat{C}^+$: objects of  $\Cat{C}^+$ are those of $\Cat{C}$; for all objects $X,Y$, the homset is defined as $\Cat{C}^+[X,Y] \defeq \Dis(\Cat{C}[X,Y])$. For $d_1\colon X \to Y$ and $d_2\colon Y \to Z$, their composition $d_1;d_2 \colon X\to Z$ is defined for all $h\in \Cat{C}[X,Z]$ as $d_1;d_2(h) \defeq \sum_{\{(f,g)\mid f;g=h\}}d_1(f) \cdot d_2(g)$; The identity $\id{X}\colon X \to X$ is given by $\delta_{id_X}$. One can easily see that $\Cat{C}^+$ is a $\Cat{PCA}$-enriched category. Moreover, the assignment $\Cat{C} \mapsto \Cat{C}^+$ gives rise to a functor $(-)^+\colon \Cat{Cat} \to \Cat{PCACat}$ which is left adjoint to the forgetful functor $U$ (see e.g. \cite[Prop. 6.4.7]{borceux2} or \cite[Cor. 1]{villoria2024enriching}). 

\begin{theorem}[\cite{borceux2,villoria2024enriching}]\label{thm:freeenriched}
 $(-)^+\colon \Cat{Cat} \to \Cat{PCACat}$ is left adjoint to $U\colon \Cat{PCACat} \to \Cat{Cat}$.
\end{theorem}

\begin{example}\label{ex:PCA}
Let $\Cat{1}$ be a category with a single object and a single arrow. The $\Cat{PCA}$-enriched category $\Cat{1}^+$ has a single object, arrows are $p\in[0,1]$ and composition is given by multiplication.

For a set $A$, the set of words  over $A$ (hereafter denoted by $A^*$) carries the structure of a category with a single object.  The $\Cat{PCA}$-enriched category $(A^*)^+$ has a single object and arrows are subdistributions over $A^*$. Composition of $d_1,d_2\in \Dis(A^*)$ is defined for all words $w\in A^*$ as $d_1;d_2(w)\defeq \sum_{u;v=w}d_1(u)\cdot d_2(v)$. 
\end{example}

Our main example of $\Cat{PCA}$-enriched category is $\KlD$, the Kleisli category of the monad $\subdistr\colon \Sets \to \Sets$ (see e.g.,~\cite{hasuo2007generic}).
Objects are sets; morphisms \(f \colon X \to Y\) are functions \(X \to \subdistr(Y)\). 
  We often write \(f(y \mid x)\) for \(f(x)(y)\)  as this number represents the probability that \(f\) returns \(y\) given the input \(x\).
  Identities \(\id{X} \colon X \to \subdistr(X)\) map each element \(x \in X\) to \(\delta_{x}\).
  For two functions \(f \colon X \to \subdistr(Y)\) and \(g \colon Y \to \subdistr(Z)\), their composition in $\KlD$ is defined as \(f ; g (z \mid x) \defeq \sum_{y \in Y} f(y \mid x) \cdot g(z \mid y)\).

In our work, it is fundamental that $\KlD$ carries two symmetric monoidal structures: $ (\KlD, \per, \uno)$ and $(\KlD, \piu, \zero)$. The monoidal product $\per$ is defined on objects as the cartesian product of sets, often denoted by $\times$, with unit the singleton $\uno \defeq \{\bullet\}$; $\piu$ is the disjoint union of sets with unit the empty set $\zero \defeq \{\}$. Hereafter we denote the disjoint union of two sets $X$ and $Y$ by $X\oplus Y \defeq \{(x,0) \mid x\in X\} \cup \{(y,1) \mid y \in Y\}$ where $0$ and $1$ are tags used to distinguish the set of provenance. For arrows  \(f \colon X \to Y\) and \(g \colon X' \to Y'\), $f\per g\colon X\otimes X' \to Y \otimes Y'$ and $f \piu g \colon X\oplus X' \to Y\oplus Y'$ are defined, for all $x\in X$, $x'\in X'$, $y\in Y$, $y'\in Y'$, $u\in X\oplus X'$ and $v\in Y\oplus Y'$, as follows.
\begin{equation}\label{ex:products}
\begin{aligned}
f \per g(y,y' \mid x,x') &\defeq f(y \mid x) \cdot g(y' \mid x')\\
f\piu g(v \mid u) &\defeq 
\begin{cases} 
f(y \mid x) & \text{if } u=(x,0) \text{ and } v=(y,0)\\ 
g(y' \mid x') & \text{if } u=(x',1) \text{ and } v=(y',1) \\  
0 & \text{otherwise} 
\end{cases}
\end{aligned}
\end{equation}

Symmetries $\symmt{X}{Y}\colon X \otimes Y \to Y \otimes X$ and $\symmp{X}{Y}\colon X \oplus Y \to Y \oplus X$ are defined as in $\Sets$. Actually, $ (\KlD, \piu, \per, \zero, \uno)$ forms a rig (aka bimonoidal) category in the sense of \cite{laplaza_coherence_1972}. We will come back to rig categories in Section \ref{sec:tapediagrams}. Until then, we will focus only on the monoidal structure provided by $\oplus$. 

%% file: sections/cbproducts.tex
\section{Convex Biproduct Categories}\label{sec:cbproducts}  

As is the case for all Kleisli categories of monads on $\Sets$, the category $\KlD$ inherits coproducts from $\Sets$: the operation $\oplus$  in \eqref{ex:products}, serves as a coproduct in $\KlD$. Our initial observation is that $\oplus$ satisfies an additional universal property, which we refer to as the \emph{convex product}, described below.


\begin{definition}
Let $X_1$, $X_2$  be two objects of a $\Cat{PCA}$-enriched category  $\Cat{C}$. The \emph{convex product} of $X_1$ and $X_2$ is an object $Z$ with two arrows $\pi_1\colon Z\to X_1$ and $\pi_2\colon Z\to X_2$ satisfying the following property: for all $p_1,p_2 \in [0,1]$ such that $p_1+p_2\le 1$ and all arrows $f\colon A\to X_1$, $g\colon A \to X_2$, there exists a unique arrow $h\colon A \to Z$ such that $h;\pi_1 = p_1\cdot f$ and $h;\pi_2 = p_2 \cdot g$.
\end{definition}

Similarly, the convex product of $n$ objects $X_1,\ldots,X_n$ is an object $Z$ with arrows $\pi_i\colon Z\to X_i$ for $i=1,\ldots,n$ satisfying the following property: for all $p_1, \dots p_n \in [0,1]$ where $\sum_{i=1}^n p_i \le 1$ and arrows $f_i\colon A\to X_i$, there exists a unique arrow $h\colon A \to Z$ such that $h;\pi_i = p_i\cdot f_i$ for all $i=1,\ldots,n$. Observe that, by definition, the 0-ary convex product is a final object. Hereafter, we will denote the unique arrow $h$ by $\langle f_1, \dots, f_n \rangle_{\vec{p}}$ where $\vec{p}$ is a compact notation for $p_1, \dots , p_n$.


\begin{definition}\label{def:convbicat}
A \emph{convex biproduct category} is a $\Cat{PCA}$-enriched category $\Cat{C}$ with an object $\zero$ which is both initial and final and, for every pair of objects $X_1,X_2$, an object $X_1\oplus X_2$ and morphisms $\pi_i \colon X_1\oplus X_2 \to X_i$ and $\iota_i \colon X_i \to X_1 \oplus X_2$ such that $(X_1\oplus X_2, \iota_1, \iota_2)$ is a coproduct, $(X_1\oplus X_2, \pi_1, \pi_2)$ is a convex product and
\begin{equation}\label{eq:delta}\iota_i; \pi_j = \delta_{i,j}\end{equation}
where $\delta_{i,j} \colon X_i \to X_j$ is defined as $\delta_{i,j}\defeq \id{X_i}$ if $i=j$ and $\delta_{i,j}\defeq\star_{X_i,X_j}$ otherwise. A morphism of convex biproduct categories is a $\Cat{PCA}$-enriched functor $F\colon \Cat{C} \to \Cat{D}$ preserving finite coproducts. We write $\Cat{CBCat}$ for the category of convex biproduct categories and their morphisms.
\end{definition}


\begin{lemma}
    In a convex biproduct category $\Cat{C}$,  the enrichment is compatible with respect to  the (convex-co)products, namely, 
\begin{equation}\label{eq: assioma aggiuntivo}
    f+_p g = \langle f,g \rangle_{p,1-p};[\id{Y},\id{Y}]
\end{equation}
for all $f,g\colon X \to Y$ and $p\in [0,1]$. 
\end{lemma}
\begin{proof}
    Consider the arrows $f;\iota_1$ and $g;\iota_2$ from $X$ to $Y\oplus Y$ and their sum $f;\iota_1 +_p g;\iota_2$. The following computation shows that 
    \begin{align}
        (f;\iota_1 +_p g;\iota_2);\pi_1 & = f,\iota_1;\pi_1 +_p g;\iota_2;\pi_1 \tag{pca-enrichment}\\
        & = f +_p \star_{X,Y} \tag{by \eqref{eq:delta}}\\
        & = p \cdot f \notag
    \end{align} 
    and similarly $(f;\iota_1 +_p g;\iota_2);\pi_2 = (1-p)\cdot g$. By the universal property of convex products, we conclude that $f;\iota_1 +_p g;\iota_2 = \langle f,g \rangle_{p,1-p}$. Composing both sides with $[\id{Y},\id{Y}]$ we obtain 
    \begin{align}
        f +_p g & = (f;\iota_1 +_p g;\iota_2);[\id{Y},\id{Y}] \tag{pca-enrichment}\\
        & = \langle f,g \rangle_{p,1-p};[\id{Y},\id{Y}]. \notag
    \end{align}
\end{proof}


  As for categories with finite biproducts, morphisms of convex biproduct categories preserve (convex) products (see Proposition \ref{prop: functor1} in Appendix~\ref{app:sec:cbproducts}).
\begin{lemma}\label{lemma:naryconvexproduct}
Let $\Cat{C}$ be a convex biproduct category. Then $\Cat{C}$ has $n$-ary convex products.
\end{lemma}

Since any convex biproduct category $\Cat{C}$ has finite coproducts, by Fox theorem \cite{fox1976coalgebras}, it carries a symmetric monoidal category $(\Cat{C}, \oplus, \zero)$ where every object $X$ is equipped with a commutative monoid that is natural and coherent, i.e., morphisms $\codiag{X} \colon X\oplus X \to X$ and $\cobang{X} \colon \zero \to X$ satisfying the laws in Figures~\ref{fig:monoidax} and \ref{fig:fccoherence} where  $\assoc{X}{Y}{Z}$, $\lunit{X}$, $\runit{Y}$ denote associators, left and right unitors. Recall that $\codiag{X}$ is given as the copairing of the identities $[\id{X},\id{X}]$ while $\cobang{X}$ is the unique arrow from the initial object.

Similarly, the convex product structure equips any object $X$ with a natural and coherent \emph{co-pca}: for all $p\in (0,1)$, there exist an arrow $\diagp{X} \colon X \to X \oplus X$, defined as $\langle id_X, id_X \rangle_{p,1-p}$, and an arrow $\bang{X}\colon X \to \zero$, defined as the unique map to the final object $\zero$, satisfying the laws in Figures~\ref{fig:co-pca axioms} and \ref{fig:copcacoherence}.
\begin{proposition}\label{lemma: copca objects in convbicat}
Let $\Cat{C}$ be a convex biproduct category. Then for every object $X$,  the triple $(X,\diagp{X},\bangp{X})$ is a natural and coherent co-pca, i.e., the axioms in Figures~\ref{fig:co-pca axioms} and \ref{fig:copcacoherence} hold. 
\end{proposition}

In Appendix \ref{app:examples}, we illustrate some examples of convex biproduct categories, such as $\KlD$ and its continuous analogue. 
The most relevant for our work is $\KlD$.
For all sets $X$,  co-pcas and monoids in $\KlD$ are illustrated below.
%
%
%
\begin{equation}\label{ex:comonoids}
\arraycolsep=2pt
\begin{array}{cccc}
\begin{array}{rcl}
\diagp{X} \colon  X & \to & X \oplus X \\
x & \mapsto & \delta_{(x,0)}+_p\delta_{(x,1)}
\end{array}
&
\begin{array}{rcl}
\bang{X}  \colon X & \to & \zero\\
x & \mapsto& \star
\end{array}
&
\begin{array}{rcl}
\cobang{X}\colon  \zero &\to& X\\
\text{ }
\end{array}
&
\begin{array}{rcl}
\codiag{X}  \colon  X\oplus X & \to & X\\
(x,i) & \mapsto& \delta_{x}
\end{array}
\end{array}
\end{equation}
Monoids and co-pcas will be useful later to freely generate convex biproduct categories. In particular, morphisms of convex biproduct categories  can be characterised as follows.

\begin{proposition}\label{prop: monoidal functors}
A functor $F\colon\Cat{C}\to \Cat{D}$ is a morphism of convex biproduct categories if and only if it is a strong monoidal functor preserving monoids and co-pcas.
\end{proposition}


The axioms in Figures~\ref{fig:monoidax}, \ref{fig:fccoherence}, \ref{fig:co-pca axioms}, and \ref{fig:copcacoherence} have been placed in the Appendix, since their \emph{strict} counterparts already appear in Tables \ref{fig:freestrictfccat} and \ref{fig:freecopcacat}. From now on, we assume monoidal categories to be strict, meaning that the structural isomorphisms  $\assoc{X}{Y}{Z}$, $\lunit{X}$, $\runit{X}$ are identities. This assumption is made without loss of generality \cite{mac_lane_categories_1978}, it greatly simplifies calculations and allows for diagrammatic notations: see the last three rows of Figure \ref{fig:tapesax} for a diagrammatic representation of the axioms of natural monoids and co-pcas.

\begin{table}[t]
\resizebox{\textwidth}{!}{%
\begin{tabular}{cc cc}
    \toprule
    $(\id{ P}\piu \codiag{ P}) ; \codiag{ P} = (\codiag{ P}\piu \id{ P}) ; \codiag{ P}$ & (\newtag{$\codiag{}$-as}{eq:codiag assoc}) &  $(\cobang{ P}\piu \id{ P}) ; \codiag{ P}  = \id{ P}  $ & (\newtag{$\codiag{}$-un}{eq:codiag unital}) \\[0.3em]
    $\cobang{\zero} = \id{\zero} \qquad\codiag{\zero} = \id{\zero}$ & (\newtag{$\cobang{\zero},\codiag{\zero}$-coh}{eq:codiag zero coherence}) & $\sigma_{ P, P};\codiag{ P}=\codiag{ P}$ & (\newtag{$\codiag{}$-sym}{eq: codiag symmetry}) \\[0.3em]
    $\cobang{P \piu Q} = \cobang{P} \piu \cobang{Q}$ & (\newtag{\,$\cobang{}$-coh}{eq:cobang coherence}) & $\codiag{P \piu Q} = (\id{P} \piu \sigma_{P,Q} \piu \id{Q}) ; (\codiag{P}\piu \codiag{Q}) $ & (\newtag{$\codiag{}$-coh}{eq:codiag coherence}) \\[0.3em]
    $\cobang{P};f =\cobang{Q}$ & (\newtag{\,$\cobang{}$-nat}{eq:cobang nat}) & $\codiag{P};f =(f\piu f); \codiag{Q}$ & (\newtag{$\codiag{}$-nat}{eq:codiag nat}) \\
    \bottomrule
\end{tabular}}
\caption{Axioms for natural and coherent monoids in a strict monoidal category.}\label{fig:freestrictfccat}
\end{table}

\begin{table}[t]
\resizebox{\textwidth}{!}{%
\begin{tabular}{cc cc}
    \toprule
    $\diagp{P};(\diagq{P}\piu \id{P})=\, \diagptilde{P};(\id{P}\piu \,\diagqtilde{})$ & (\newtag{$\diagp{}$-as}{eq:diagp assoc}) &  $\tilde{p}= pq\qquad \tilde{q}= \frac{p(1-q)}{1-pq}$ &  \\[0.3em]
    $\diagp{P};\codiag{P}  = \id{ P}$ & (\newtag{$\diagp{}$-idem}{eq:diagp idempotency}) & $\diagp{P}\sigma_{ P, P}=\,\diagpbar{P}$ & (\newtag{$\diagp{}$-sym}{eq:diagp symmetry}) \\[0.3em]
    $\diagp{\zero} = \id{\zero}$ & (\newtag{$\diagp{\zero}$-coh}{eq:diagp zero coherence}) & $\bangp{\zero} = \id{\zero}$ & (\newtag{\,$\bangp{\zero}$-coh}{eq:bangp0 coherence}) \\[0.3em]
    $\bangp{P \piu Q} = \bangp{P} \piu\, \bangp{Q}$ & (\newtag{\,$\bangp{}$-coh}{eq:bangp coherence}) & $\diagp{P \piu Q} = (\diagp{P}\piu\, \diagp{Q});(\id{P} \piu \sigma_{P,Q} \piu \id{Q}) $ & (\newtag{$\diagp{}$-coh}{eq:diagp coherence}) \\[0.3em]
    $f; \bangp{Q}=\bangp{P}$ & (\newtag{\,$\bangp{}$-nat}{eq:bangp nat}) & $f; \diagp{Q}=\,\diagp{P}; (f \piu f)$ & (\newtag{$\diagp{}$-nat}{eq:diagp nat}) \\
    \bottomrule
\end{tabular}}
\caption{Axioms for natural and coherent co-pcas in a strict monoidal category.}\label{fig:freecopcacat}
\end{table}

%
%
%

%% file: sections/cmatrix.tex
%
%
\section{Stochastic Matrices over PCA-enriched categories}\label{sec:cmatrix}    
It is well known (see, e.g., \cite[Exercises VIII.2.5-6]{mac_lane_categories_1978}) that from a category enriched over commutative monoids, one can freely generate the finite biproduct category of its matrices. In this section, we show that a similar construction holds for $\Cat{PCA}$-enriched categories: given such a category $\Cat{C}$, one can construct the convex biproduct category $\stmat{\Cat{C}}$ of \emph{stochastic matrices} over $\Cat{C}$.

Objects of $\stmat{\Cat{C}}$ are words in $Ob(\Cat{C})^*$. We will write $\bigoplus_{k=1}^m U_k$ for the word $U_1\dots U_n$ and $\zero$ for the empty word. We will denote objects of $\Cat{C}$ by $U,V$ and those of $\stmat{\Cat{C}}$ by $P,Q$. In $\stmat{\Cat{C}}$, an arrow $M\colon\bigoplus_{k=1}^n U_k\to \bigoplus_{k=1}^m V_k$ is an equivalence class of  $m\times n$ matrices with $(j,i)$-entries given by pairs $ (p_{ji}, f_{ji})$ where $f_{ji}\in\Cat{C}[U_i,V_j]$ and $p_{ji}\in[0,1]$ satisfy $\sum_{j=1}^{m}p_{ji}\le 1$. Two matrices $M$ and $M'$ are equivalent, in symbols $M\equiv M'$, if $p_{ji}\cdot f_{ji}= p'_{ji}\cdot f'_{ji}$ in $\Cat{C}[U_i,V_j]$ for all $i,j$. 
\begin{remark}
The use of pairs $ (p_{ji}, f_{ji})$ as entries of matrices is necessary to specify the constraints $\sum_{j=1}^{m}p_{ji}\le 1$. However, $\equiv$ ensures that one can safely write $M_{ji}= p_{ji}\cdot f_{ji}$.
\end{remark}

The composition of two morphisms $M\colon\bigoplus_{k=1}^n U_k\to \bigoplus_{k=1}^m V_k$ and $M'\colon\bigoplus_{k=1}^m V_k\to \bigoplus_{k=1}^l W_k$ 
is obtained by matrix multiplication $M'M$: for all $u\in \{1,\dots,l\}$ and $i\in \{1,\dots,n\}$ the entry at $(u,i)$ is 
\begin{equation}\label{eq:matrixmult}(M'M)_{ui} \defeq  r_{ui}\cdot \left(\sum_{j=1}^{m} \frac{p'_{uj}p_{ji}}{r_{ui}}\cdot (f_{ji};f'_{uj})\right) = \sum_{j=1}^{m} p'_{uj}p_{ji}\cdot (f_{ji};f'_{uj})\end{equation} 
where the convex sums $\sum_k p_k\cdot (-)_k$ are those provided by the $\Cat{PCA}$ enrichment of $\Cat{C}$, the composition $;$ is in $\Cat{C}$ and
$r_{ui}= (\sum_{j=1}^{m} p'_{uj}p_{ji})$. Simple computations confirm that $\sum_{u=1}^{l}r_{ui}\leq 1$, $\frac{p'_{uj}p_{ji}}{r_{ui}}\in[0,1]$ and that the equality in (\ref{eq:matrixmult}) follows from Lemma~\ref{lemma:initialproperties}.\ref{lemma q per somma} in Appendix~\ref{app:sec:pca}.
For all objects $P=\bigoplus_{k=1}^n U_k$, $\id{P}$ is the $n\times n$ matrix with entries $(\id{P})_{jj}=1\cdot \id{U_j}$ and, for $i\neq j$, $(\id{P})_{ji}=0 \cdot \star$.

\begin{example}\label{ex:matrices}
Recall $\Cat{1}^+$ and $(A^*)^+$ from Example \ref{ex:PCA}.
In $\stmat{\Cat{1}^+}$ objects are natural numbers and arrows $n \to m$ are the usual $m\times n$ sub-stochastic matrices (i.e., $\sum_{j=1}^mp_{ji}\leq 1$).
In $\stmat{(A^*)^+}$ objects are natural numbers and arrows $n \to m$ are $m\times n$ matrices with entries in $\Dis(A^*)$. Consider for instance the matrices $N\colon 2\to 3$ and $M\colon 2 \to 2$ on the left below.  The composition $M;N\colon 2 \to 3$ is the matrix on the right.
\begin{equation*}
\begin{pmatrix}
    \frac{1}{2} \cdot a & \frac{1}{2} \cdot c\\
    \frac{1}{3} \cdot ab & 0 \cdot \star \\
    0 \cdot \star & \frac{1}{3} \cdot \id{}
\end{pmatrix} 
\begin{pmatrix}
    \frac{1}{2} \cdot a & 1 \cdot c\\
    \frac{1}{2} \cdot ab & 0 \cdot \star \\
\end{pmatrix} 
=
\begin{pmatrix}
    \frac{1}{2} \cdot (\frac{1}{2} \cdot aa +\frac{1}{2}abc) & \frac{1}{2} \cdot ca\\
    \frac{1}{6} \cdot aab & \frac{1}{3} \cdot cab \\
    \frac{1}{6} \cdot ab  & 0 \cdot \star
\end{pmatrix} 
\end{equation*}
\end{example}

\begin{proposition}\label{prop:cmatrixcb}
     $\stmat{\Cat{C}}$ is a convex biproduct category.
\end{proposition}
%
%
For arbitrary matrices $M$ and $N$, $M\oplus N$ is defined as below on the left
\begin{equation}\label{eq:matsmc}
M \piu N \defeq \begin{pmatrix}
    M & \emptyset\\
    \emptyset & N
\end{pmatrix}
\qquad
\symm{U}{Q}\defeq \begin{pmatrix}
    \emptyset & \id{Q}\\
    \id{P} & \emptyset
\end{pmatrix}
\end{equation}
where $\emptyset$ is the matrix with all entries $ 0\cdot \star$. The symmetry $\symmp{P}{Q}\colon P\oplus Q \to Q \oplus P$ is defined as on the right above. 
%
%
%
%
%
%
%
Every object $P$ is equipped with  co-pca and monoids. These are defined for all objects $U\in \Cat{C}$ as follows, where $!_U$ (respectively $?_U$) is the unique matrix with $0$ rows (columns).
\begin{equation}\label{eq:matmonpca}
\diagp{U}\defeq\begin{pmatrix}
    p\cdot \id{U} \\
    (1-p)\cdot \id{U}
\end{pmatrix}
\qquad
\bang{U}\defeq !_U
\qquad 
\cobang{U}\defeq ?_U
\qquad
\codiag{U}\defeq \begin{pmatrix}
    1\cdot \id{U} & 1\cdot \id{U}
\end{pmatrix} 
\end{equation}
 For arbitrary objects $P$ of $\stmat{\Cat{C}}$, co-pca and monoids are defined inductively as:
\begin{equation}\label{eq:indmoncopca}
\begin{array}{c|c}
\begin{array}{ccc}
 \diagp{\zero}\defeq \id{\zero}\; &\; \bang{\zero}\defeq \id{\zero}\; &\; \bang{U \oplus P}\defeq \bang{U}\oplus \bang{P}
 \end{array}
 &
 \begin{array}{ccc}
\codiag{\zero}\defeq \id{\zero}\; &\; \cobang{\zero}\defeq \id{\zero} \;& \;\cobang{U \oplus P}\defeq \cobang{U}\oplus \cobang{P}
\end{array} \\
\diagp{U \oplus P}\defeq (\diagp{U} \piu\, \diagp{P});(\id{U} \piu \symm{U}{P} \piu \id{P})&
\codiag{U\oplus P}\defeq (\id{U} \piu \symm{P}{U} \piu \id{P});(\codiag{U}\piu \codiag{P})
\end{array}
\end{equation}

%
%
%
%

Like categories of matrices are freely generated finite biproduct categories \cite{mac_lane_categories_1978}, similarly  $\stmat{\Cat{C}}$ is the convex biproduct category freely generated by $\Cat{C}$. We are going to illustrate this below. 

First, for all $\Cat{PCA}$-enriched functors $F\colon \Cat{C}\to \Cat{D}$, one can define $\stmat{F}\colon \stmat{\Cat{C}}\to \stmat{\Cat{D}}$ as the functor mapping an object $\bigoplus_{k=1}^n U_k$ to $\bigoplus_{k=1}^n F(U_k)$ and a matrix
$M\colon\bigoplus_{k=1}^n U_k\to \bigoplus_{k=1}^m V_k$ with entries $M_{ji}= p_{ji}\cdot f_{ji}$  into the matrix with entries $\stmat{F}(M)_{ji}\defeq p_{ji}\cdot F(f_{ji})$.
\begin{proposition}\label{prop:stmatfun}
    Let $F\colon \Cat{C}\to \Cat{D}$ be a PCA-enriched functor. $\stmat{F} \colon\stmat{\Cat{C}}\to \stmat{\Cat{D}}$ is a morphism of convex biproduct categories.
\end{proposition}
Then, it is easy to check that the assignment $\Cat{C} \mapsto \stmat{\Cat{C}}$ and $F \mapsto \stmat{F}$ provides a functor $\stmat{-}:\Cat{PCACat}\to \Cat{CBCat}$ which is left adjoint to the forgetful $U\colon \Cat{CBCat} \to \Cat{PCACat}$.
\begin{theorem}\label{thm:matfree}
    $\stmat{-} \colon \Cat{PCACat}\to \Cat{CBCat}$ is left adjoint to $U\colon \Cat{CBCat} \to \Cat{PCACat}$
\end{theorem}
We conclude this section with a simple observation that will be useful later in Section~\ref{sec:probboolcircuits}.
\begin{proposition}\label{prop:counit fullfaithful}
    For any convex bicategory $\Cat{C}$, the counit of the adjunction above\\ $\epsilon_\Cat{C}\colon \stmat{U(\Cat{C})}\to \Cat{C}$ is a full and faithful morphism of convex bicategories.
\end{proposition}

%% file: sections/syntactic.tex
\section{An equational presentation of stochastic matrices}\label{sec:syntactic}    

\begin{table}[t]
\resizebox{\textwidth}{!}{%
$
\begin{array}{c}
\toprule
\begin{array}{c}
    \begin{array}{ccccc}
        {\diagp{U}\colon U \to U\oplus U} &
        {\bang{U} \colon U \to \zero} &
        {\sigma_{U, V}^{\piu} \colon U \piu V \to V \piu U} &
        {\codiag{U}\colon U \piu U \to U} &
        {\cobang{U} \colon \zero \to U}   
    \end{array}    
\\[0.3em]
    \begin{array}{ccccc}
        {id_\zero \colon \zero \to \zero} &
        {id_U \colon U \to U} &
        \inferrule{c \colon U \to V}{\tapeFunct{c}\colon U \to V} &
        \inferrule{\t \colon P \to Q \and \s \colon Q \to R}{\t ; \s \colon P \to R} &
        \inferrule{\t \colon P_1 \to Q_1 \and \s \colon P_2 \to Q_2}{\t \piu \s \colon P_1 \piu P_2 \to Q_1 \piu Q_2}       
    \end{array}
\end{array}
\\[0.5em]
\begin{array}{cc}
\midrule
\begin{array}{c}
(f;g);h=f;(g;h) \qquad id_P;f=f=f;id_Q\\
(f_1\piu f_2) ; (g_1 \piu g_2) = (f_1;g_1) \piu (f_2;g_2)
\end{array} 
&
\begin{array}{c}
id_{\zero}\piu f = f = f \piu id_{\zero} \qquad (f \piu g)\, \piu h = f \piu \,(g \piu h) \\
\sigma_{P, Q}; \sigma_{Q, P}= id_{P \piu Q} \qquad (\gen \piu id_R) ; \sigma_{Q, R} = \sigma_{P,R} ; (id_R \piu \gen)
\end{array}
\end{array}\\
\bottomrule
\end{array}
$
}
\caption{Typing rules (top) and axioms (bottom) for freely generated strict symmetric monoidal categories.}
\label{fig:freestricmmoncatax}
\end{table}

So far, we have seen that given a category, one can first freely enrich it over $\Cat{PCA}$, and then obtain a convex biproduct category by taking its category of stochastic matrices.
Now, we give a syntactic description, in terms of generators and equations, of such category.

For a category $\Cat{C}$, we consider terms generated by the following grammar
  
\begin{equation}\label{tapesGrammar}
\setlength{\arraycolsep}{3pt}
\begin{array}{rcccccccccccccccccccc}
\t & ::= & \diagp{U} & \mid & \bangp{U} & \mid & \tapeFunct{c} & \mid & \cobang{U} & \mid & 
\codiag{U} & \mid & \id{U} & \mid & \id{\zero} & \mid & \sigma_{U,V}^{\piu} & \mid & \t ; \t & \mid & \t \piu \t
\end{array}
\end{equation}
%
%
where $p\in (0,1)$, $U,V \in \ob{\Cat{C}}$, and $c$ is an arrow in $\Cat{C}$. Terms are typed according to the rules at the top of Table~\ref{fig:freestricmmoncatax}: each type is an arrow $P \to Q$ where $P,Q\in \ob{\Cat{C}}^*$ are regarded as sums of objects of $\Cat{C}$. As expected, we consider only those terms that are typable. For arbitrary $P \in \ob{\Cat{C}}^*$, we define
$\diagp{P},\bangp{P},\cobang{P},\codiag{P}$ as in \eqref{eq:indmoncopca}. Analogous inductive definitions give us $\id{P}$ and $\symmp{P}{Q}$. 

The category $\CatTapeC$ has as set of objects $\ob{\Cat{C}}^*$. Arrows in $\CatTapeC[P,Q]$ are terms modulo the axioms of natural and coherent monoids, co-pcas, strict symmetric monoidal categories (respectively in Tables~\ref{fig:freestrictfccat}, \ref{fig:freecopcacat} and~\ref{fig:freestricmmoncatax} bottom) and the following two axioms.
\begin{equation}\label{eq:TapeFunctAxioms}
\tapeFunct{\id{P}}=\id{P} \qquad \tapeFunct{c;d} = \tapeFunct{c}; \tapeFunct{d} \tag{Tape}
\end{equation}
Identities, composition, symmetries and monoidal product are defined as on terms. It is thus immediate to see that $(\CatTapeC,\oplus, \zero)$ forms a symmetric monoidal category. Moreover such category is enriched over $\Cat{PCA}$:  for all objects $P,Q$ and arrows $f,g\colon P \to Q$
\begin{equation}\label{eq:enrichmentTC}
f +_p g \defeq\   \diagp{P} ; (f\oplus g) ; \codiag{Q} \qquad \star_{P,Q} \defeq \, \bangp{P} ; \cobang{Q}
\end{equation}
By naturality of $\cobang{}$ and $\bang{}$, $\zero$ is both initial and final object. By the Fox theorem \cite{fox1976coalgebras}, $\oplus$ is a coproduct with injections $\id{P_1}\oplus \cobang{P_2} \colon P_1 \to P_1\oplus P_2$ and $\cobang{P_1}\oplus \id{P_2} \colon P_2 \to P_1\oplus P_2$. 

Our main technical effort consists in proving that $\oplus$ is a convex product with projections $\id{P_1}\oplus \bang{P_2} \colon P_1 \oplus P_2 \to P_1$ and $\bang{P_1}\oplus \id{P_2} \colon P_1\oplus P_2 \to  P_2$. By the coproduct property any $\t\colon\bigoplus_{j=1}^mV_j \to \bigoplus_{i=1}^nU_i$ is the copairing $[\t_1, \dots, \t_m]$ for $\t_j=\iota_j; \t$. One can thus restrict to consider the case of arrow of type $ V_j \to \bigoplus_{i=1}^nU_i$.
First, we extend $\diagp{}$ to arbitrary $p\in [0,1]$: $\diagpX{0\;\;}{U}\defeq\cobang{U}\oplus \id{U}$ and $\diagpX{1\;\;}{U}\defeq \id{U} \oplus\cobang{U}$. Then, for all $n\in \mathbb{N}$ and $\vec{p}=p_1,\dots ,p_n$ with $p_i\in[0,1]$ such that $\sum_{i=1}^n p_i\leq 1$, we inductively define $\diagpn{\vec{p}\;}{U}{n} \colon U \to \bigoplus_{i=1}^nU$ as follows
\begin{equation}\label{eq: diagpn cap equational presentation}
\diagpn{\vec{p}\;}{U}{0}\defeq \bang{U} \qquad \diagpn{\vec{p}\;}{U}{n+1}\defeq\, \diagpn{\;p_{1}}{U}{}; (\id{U} \oplus \diagpn{\;\vec{q}\;}{U}{n})
\end{equation}
where $\vec{q}=q_1,\dots q_{n}$ for $q_i=0$ if $p_{1}= 1$ and $q_i=\frac{p_{i+1}}{1-p_{1}}$ otherwise. 
Given $n$ arrows $t_i\colon U \to U_i$, simple computations confirm that $(\diagpn{\;\vec{p}\;}{U}{n}; \bigoplus_{i=1}^n \t_i) ; \pi_i = p_i\cdot \t_i$. Thus, $\diagpn{\;\vec{p}\;}{U}{n}; \bigoplus_{i=1}^n \t_i \colon U \to \bigoplus_{i=1}^nU_i$ is a mediating arrow satisfying the constraints of $n$-ary convex products. The hard part consists in proving its uniqueness. We rely on two key properties of $\CatTapeC$:

\begin{lemma}\label{decomposition}
For all $\t\colon U \to \bigoplus_{i=1}^nU_i$, there exist $\vec{p}=p_1,\dots ,p_n$ and $t_i\colon U \to U_i$ such that $\t=\,\diagpn{\;\vec{p}\;}{U}{n}; \bigoplus_{i=1}^n \t_i$.
\end{lemma}
\begin{lemma}\label{cancellativity}[Cancellativity]
For all $r\in (0,1)$, for all $\s,\t $, 
if $r\cdot \s = r\cdot \t$ then $\s = \t$.
\end{lemma}
\begin{theorem}\label{thm:TCconvexbiproductcategory}
$\CatTapeC$ is a convex biproduct category. In particular, for all $\vec{p}=p_1,\dots ,p_n$ and $\t_i\colon U \to U_i$, $\langle \t_1, \dots \t_n\rangle_{\vec{p}} = \diagpn{\;\;\vec{p}}{U}{n}; \bigoplus_{i=1}^n \t_i$.
\end{theorem}
The assignment $\Cat{C} \mapsto \CatTapeC$ gives rise to a functor $\CatTapeFUN\colon \Cat{Cat} \to \Cat{CBCat}$ which is left adjoint to the forgetfull functor  $U\colon \Cat{CBCat} \to \Cat{Cat} $.

\begin{theorem}\label{thm:syntacticadjunction}
$\CatTapeFUN\colon \Cat{Cat} \to \Cat{CBCat}$ is left adjoint to $U\colon \Cat{CBCat} \to \Cat{Cat} $.
\end{theorem}
\begin{proof}
For a functor $F\colon \Cat{C} \to \Cat{D}$, $\CatTapeF \colon \CatTapeC \to \CatTapeD$ is defined on objects as $\CatTapeF(\bigoplus_{i=1}^nU_i) = \bigoplus_{i=1}^nF(U_i)$.
On arrows, it is defined inductively as follows.
\[
\setlength{\arraycolsep}{2pt} 
\begin{array}{rcl rcl}
  \CatTapeF(\t_1 ; \t_2)      &\defeq& \CatTapeF(\t_1);\CatTapeF(\t_2) &
  \CatTapeF(\t_1 \piu \t_2)   &\defeq& \CatTapeF(\t_1)\oplus\CatTapeF(\t_2) \\
  
  \CatTapeF(\diagp{P})        &\defeq& \diagp{\CatTapeF(P)} &
  \CatTapeF(\codiag{P})       &\defeq& \codiag{\CatTapeF(P)} \\
  
  \CatTapeF(\,\bangp{P})        &\defeq& \bangp{\CatTapeF(P)} &
  \CatTapeF(\cobang{P})       &\defeq& \cobang{\CatTapeF(P)} \\
  
  \CatTapeF(\id{U})           &\defeq& \id{\CatTapeF(U)} &
  \CatTapeF(\id{\zero})       &\defeq& \id{\CatTapeF(\zero)} \\
  
  \CatTapeF(\sigma_{P,Q}^{\piu}) &\defeq& \sigma_{\CatTapeF(P),\CatTapeF(Q)}^{\piu} &
  \CatTapeF(\tapeFunct{c})    &\defeq& \tapeFunct{F(c)}
\end{array}
\]

Since by definition $\CatTapeF$ is a monoidal functor preserving co-pcas and monoids, by Proposition \ref{prop: monoidal functors}, it is a morphism of convex bicategories.
Thus, we have a functor $\CatTapeFUN\colon \Cat{Cat} \to \Cat{CBCat}$. Below, we prove the adjunction.

For all categories $\Cat{C}$, the unit of the adjunction $\eta_{\Cat{C}} \colon \Cat{C} \to \CatTapeC$ is the identity on objects functor mapping each arrow $c\colon U\to V$ in $\Cat{C}$ into $\tapeFunct{c}$. The axioms in \eqref{eq:TapeFunctAxioms} force $\eta_{\Cat{C}}$ to be a functor. Naturality of the unit is straightforward.

Now, take a convex biproduct category $\Cat{D}$ and consider a functor $F\colon \Cat{C} \to U(\Cat{D})$. By Proposition \ref{lemma: copca objects in convbicat}, $\Cat{D}$ is a symmetric monoidal category where every object $X$  carries a  co-pca $(\diagp{X},\bang{X})$ and a natural coherent monoid $(\codiag{X},\cobang{X})$. One can use these structures to define $F^\sharp\colon \CatTapeC \to \Cat{D}$ inductively in the same way as  $\CatTapeF$: e.g., $F^\sharp(\diagp{P})\defeq \,\diagp{F^\sharp(P)}$. The base case $F^\sharp(\tapeFunct{c})=F(c)$ ensures that $\eta_{\Cat{C}}; F^\sharp = F$.
Since by definition $F^\sharp$ is a strict symmetric monoidal functor preserving co-pcas and monoids, by Proposition \ref{prop: monoidal functors}, it is a morphism of convex bicategories. 

To prove uniqueness, take a morphism of convex bicategories $H \colon \CatTapeC \to \Cat{D}$ such that $\eta_{\Cat{C}}; H = F$, i.e.,  $H(\tapeFunct{c})=F(c)$. By Proposition \ref{prop: monoidal functors}, $H$ has to preserve co-pcas and monoids. Thus $H=F^\sharp$.
\end{proof}

\begin{corollary}\label{cor:isotapematrices}
For all categories $\Cat{C}$, $\CatTapeC$ is isomorphic to $\stmat{\Cat{C}^+}$ in $\Cat{CBCat}$.
\end{corollary}
\begin{proof}
By composition of adjoints, their uniqueness and Theorems \ref{thm:freeenriched}, \ref{thm:matfree} and \ref{thm:syntacticadjunction}.
\end{proof}
The isomorphism $\CatTapeC \to \stmat{\Cat{C}^+}$ maps $\tapeFunct{c}$ into the 
$1\times 1$ matrix $\begin{pmatrix}
   1 \cdot c
\end{pmatrix}$ and  $\diagp{U}$, $\bang{U}$, $\cobang{U}$, $\codiag{U}$ into the matrices defined in \eqref{eq:matmonpca}.  Compositions and sums of arrows in $\CatTapeC$ 
are mapped into multiplications and direct sums of matrices as defined in \eqref{eq:matrixmult} and \eqref{eq:matsmc}. Same for identities and symmetries. For instance, $(\diagpX{\frac{1}{2}\;\,}{U} \oplus \id{U}) ; (\tapeFunct{a}\oplus \tapeFunct{ab} \oplus \tapeFunct{c}); (\id{U}\oplus \symmp{U}{U}); (\codiag{U}\oplus \id{U})$ is mapped into the matrix $M$ of Example \ref{ex:matrices}. 

The characterisation of $\stmat{\Cat{C}^+}$ by means of generators and equations provided by the above corollary, paves the way to study further equational theories (such as the one in Section \ref{sec:probboolcircuits}). The following result guarantees that the category obtained by quotienting $\CatTapeC$ by additional axioms is still a convex biproduct category.
%
\begin{proposition}\label{cor: quotient category is convex biproduct }
Let $\Cat{C}$ be a category and $\sim$ be a congruence relation (w.r.t $;$ and $\oplus$) on $\CatTapeC$. Let $\CatTapeC_\sim$ be the category obtained as the quotient of $\CatTapeC$ by $\sim$ and $Q_\sim \colon \CatTapeC\to \CatTapeC_\sim$ be the functor mapping each arrow to its $\sim$-equivalence class.
Then $\CatTapeC_\sim$ is a convex biproduct category and $Q_\sim$ is a morphism of convex biproduct categories.
\end{proposition}

%% file: sections/probtapes.tex
\section{Probabilistic Tape Diagrams}\label{sec:tapediagrams}
Recall that a \emph{monoidal signature} is a tuple $(\sort, \sign, \ari, \coar)$ where $\sort$ is a set of basic sorts, hereafter denoted by $A,B,\dots$, $\sign$ is a set of generators, denoted by $s,t \dots$, and $\ari,\coar \colon \sign \to \sort^*$ assign to each symbol its arity and coarity (words over $\sort$). From a monoidal signature $\sign$, one can freely generate the strict symmetric monoidal $\DiagS$:  objects are words in $\sort^*$; arrows are \emph{string diagrams} \cite{joyal1991geometry,selinger2010survey}. These can be regarded  as the terms generated by the following grammar (where $A,B \in \sort$ snd $s \in \Sigma$)
\begin{equation}\label{stringdiagramGrammar}
\setlength{\arraycolsep}{3pt} 
\renewcommand{\arraystretch}{1.1} 
\begin{array}{rcccccccccccccccc}
c & ::= & \id{A} & \mid & \id{\uno} & \mid & \gen & \mid & \symmt{A}{B} & \mid & c ; c & \mid & c \per c
\end{array}
\end{equation}
modulo the axioms of strict symmetric monoidal categories. A \emph{monoidal theory} $\mathbb{T}=(\sign, E)$ consists of a monoidal signature equipped with a set $E$ of pairs of arrows of $\DiagS$ with same source and target. Let $=_{\mathbb{T}}$ be equivalence relation on arrows  of $\DiagS$ obtained as the congruence closure (w.r.t. $;$ and $\otimes$) of $E$. We write $\DiagT$ for the category obtained as the quotient of $\DiagS$ by $=_{\mathbb{T}}$.

The category $\CatT{\DiagT}$, obtained by applying the construction from Section \ref{sec:syntactic} to $\DiagT$, coincides by definition with the category of \emph{tape diagrams} from \cite{bonchi2025tapediagramsmonoidalmonads}. Objects of $\CatT{\DiagT}$,  are element of $(\sort^*)^*$ which we often write as polynomials $P=\Piu[i=1][n]{\Per[j=1][m_i]{A_{i,j}}}$. 
We will call \emph{monomials} of $P$ the $n$ words $\Per[j=1][m_i]{A_{i,j}}$. For instance the monomials of $(A \per B) \piu \uno$ are $A \per B$ and $1$. We denote monomials by $U,V,\dots$
Arrows of $\CatT{\DiagT}$ enjoy an intuitive diagrammatic representation specified by the following two layered grammar.
\begin{equation*}\label{tapesDiagGrammar} 
    \setlength{\tabcolsep}{2pt}
    \begin{tabular}{rc c@{$\,\mid\,$}c@{$\,\mid\,$}c@{$\,\mid\,$}c@{$\,\mid\,$}c@{$\,\mid\,$}c@{$\,\mid\,$}c@{$\,\mid\,$}c@{$\,\mid\,$}c@{$\,\mid\,$}c}
        $c$  & $\Coloneqq$ &  $\wire{A}$ & $ 
    \InputIfFileExists{empty.tikz}{}{\input{./tikz/empty.tikz}}
 $ & $ \Cgen{\gen}{A}{B}  $ & $ \Csymm{A}{B} $ & $ 
    \InputIfFileExists{seq_compC.tikz}{}{\input{./tikz/seq_compC.tikz}}
   $ & $  
    \InputIfFileExists{par_compC.tikz}{}{\input{./tikz/par_compC.tikz}}
$ \\
        $\t$ & $\Coloneqq$ & $
    \InputIfFileExists{/tapes/cipriano/pcomonoid.tikz}{}{\input{./tikz//tapes/cipriano/pcomonoid.tikz}}
$ & $\Tcounit{U}$  & $ \Tcirc{c}{U}{V}$ & $\Tunit{U}$  & $\Tmonoid{U}$ \\ 
        && $\Twire{U}$ & $ 
    \InputIfFileExists{empty.tikz}{}{\input{./tikz/empty.tikz}}
 $    & $ \Tsymmp{U}{V} $ & $ 
    \InputIfFileExists{tapes/seq_comp.tikz}{}{\input{./tikz/tapes/seq_comp.tikz}}
  $ & $  
    \InputIfFileExists{tapes/par_comp.tikz}{}{\input{./tikz/tapes/par_comp.tikz}}
$  
    \end{tabular}
\end{equation*} 
The first layer corresponds to \eqref{stringdiagramGrammar} while the second to \eqref{tapesGrammar}; diagrams from the first layer are string diagrams, those from the second are called tapes. Note that (a) string diagram can occur inside tapes, (b) string diagrams have type $U\to V$, for $U,V\in \sort^*$ and vertical composition corresponds to $\per$, (c) tapes have type $P\to Q$ for $P,Q\in (\sort^*)^*$ and vertical composition corresponds to $\piu$. 

The identity $\id\zero$ is rendered as the empty tape $
    \InputIfFileExists{empty.tikz}{}{\input{./tikz/empty.tikz}}
$, while $\id\uno$ is $
    \InputIfFileExists{tapes/empty.tikz}{}{\input{./tikz/tapes/empty.tikz}}
$: a tape filled with the empty string diagram. 
For a monomial $U \!=\! A_1\dots A_n$, $\id U$ is depicted as a tape containing  $n$ wires labelled by $A_i$. For instance, $\id{AB}$ is rendered as $\TRwire{A}{B}$. When clear from the context, we will simply represent it as a single wire  $\Twire{U}$ with the appropriate label.
Similarly, for a polynomial $P = \Piu[i=1][n]{U_i}$, $\id{P}$ is obtained as a vertical composition of tapes, as illustrated below on the left of (\ref{ex:tape}). 
\begin{equation}\label{ex:tape}
\scalebox{0.82}{$ 
    \id{AB \piu \uno \piu C} = \!\!\!\begin{aligned}\begin{gathered} \TRwire{A}{B} \\[-1.8mm] \Twire{\uno} \\[-1.8mm] \Twire{C} \end{gathered}\end{aligned}
    \quad
    \codiag{A\piu B \piu C} = \!\!
    \InputIfFileExists{tapes/examples/codiagApBpC.tikz}{}{\input{./tikz/tapes/examples/codiagApBpC.tikz}}

    \quad
    \cobang{AB \piu B \piu C} = \!
    \InputIfFileExists{tapes/examples/cobangABpBpC.tikz}{}{\input{./tikz/tapes/examples/cobangABpBpC.tikz}}

    \quad
    
    \InputIfFileExists{tapes/cipriano/algt.tikz}{}{\input{./tikz/tapes/cipriano/algt.tikz}}

$} \end{equation}
\noindent The codiagonal $\codiag{U} \colon  U \piu U \!\to\! U$ is represented as a merging of tapes, $\diagp{U}\colon U\oplus U \to U$ as a splitting of tapes labeled with $p\in(0,1)$, the cobang $\cobang{U} \colon \zero \!\to\! U$ is a tape closed on its left boundary, while $\bang{U} \colon U \to \zero $ is closed on the right.
Exploiting the definitions in \eqref{eq:indmoncopca}, we can construct $\codiag{P},\diagp{P},\cobang{P},\bang{P}$ for arbitrary polynomials.
For example, $\codiag{A\piu B \piu C}$ and $\cobang{AB \piu B \piu C}$ are depicted as the second and third diagrams above. The last diagram is the tape for 
$\diagp{A};(\diagq{A}\oplus \id{A}) ; (\bang{A} \oplus \codiag{A})$. 
For an arbitrary $\t \colon P \to Q$ we write $\Tbox{\t}{P}{Q}$. 

The graphical representation embodies several axioms such as \eqref{eq:TapeFunctAxioms} and those of monoidal categories. 
The remaining axioms are shown in Figure~\ref{fig:tapesax}. 
%
%
%

\medskip

By Corollary \ref{cor:isotapematrices}, we know that $\CatT{\DiagT}$ is isomorphic to $\stmat{\DiagT^+}$: tapes represents exactly stochastic matrices of subprobability distributions of string diagrams in $\DiagT$. For instance the rightmost diagram in \eqref{ex:tape} corresponds to the $1\times 1$ matrix $\begin{pmatrix} (1-p)(1-q) \cdot \wire{A}\end{pmatrix}$ while the following tape diagram on the left represents the $2\times 2$ stochastic matrix on the right.
\[
\resizebox{0.95\linewidth}{!}{$

    \InputIfFileExists{tapes/cipriano/esempiomatrice.tikz}{}{\input{./tikz/tapes/cipriano/esempiomatrice.tikz}}

\quad
\text{\raisebox{1.7ex}{$
\begin{pNiceMatrix}[first-col,first-row]
	\rotatebox{90}{$\Lsh$} & AB & C \\
	DE & p\cdot 
    \InputIfFileExists{par_compCMAT.tikz}{}{\input{./tikz/par_compCMAT.tikz}}
 & q\cdot \Cgen{h}{}{} \\
	F & 0\cdot \star & (1-q)\cdot(
    \InputIfFileExists{seq_compCMAT1.tikz}{}{\input{./tikz/seq_compCMAT1.tikz}}
 +_r 
    \InputIfFileExists{seq_compCMAT2.tikz}{}{\input{./tikz/seq_compCMAT2.tikz}}
)
\end{pNiceMatrix}
$}}
$}
\]

\begin{figure}[t]
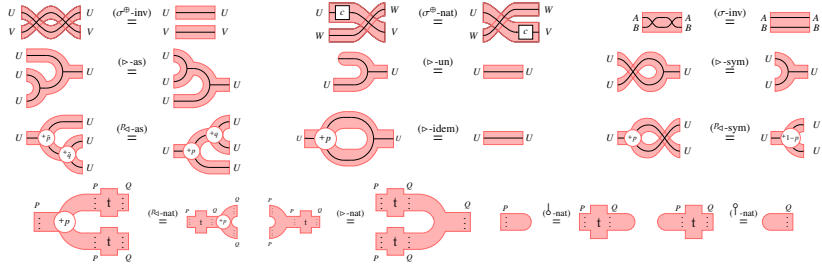

\centering
\setlength{\tabcolsep}{2pt} 
\renewcommand{\arraystretch}{0.9} 
\resizebox{0.92\linewidth}{!}{
\begin{tabular}{r@{\;}c@{\;}l @{\qquad} r@{\;}c@{\;}l @{\qquad} r@{\;}c@{\;}l}
    
    \InputIfFileExists{tapes/ax/symminv_left.tikz}{}{\input{./tikz/tapes/ax/symminv_left.tikz}}
 &$\stackrel{(\symmp\text{-inv})}{=}$& 
    \InputIfFileExists{tapes/ax/symminv_right.tikz}{}{\input{./tikz/tapes/ax/symminv_right.tikz}}
  
    & 
    
    \InputIfFileExists{tapes/ax/symmnat_left.tikz}{}{\input{./tikz/tapes/ax/symmnat_left.tikz}}
  &$\stackrel{(\symmp\text{-nat})}{=}$&  
    \InputIfFileExists{tapes/ax/symmnat_right.tikz}{}{\input{./tikz/tapes/ax/symmnat_right.tikz}}

    & 
    
    \InputIfFileExists{cb/symm_inv_left.tikz}{}{\input{./tikz/cb/symm_inv_left.tikz}}
 & $\axeq{\sigma\text{-inv}}$ & 
    \InputIfFileExists{cb/symm_inv_right.tikz}{}{\input{./tikz/cb/symm_inv_right.tikz}}
 
    \\
    
    \InputIfFileExists{tapes/whiskered_ax/monoid_assoc_left.tikz}{}{\input{./tikz/tapes/whiskered_ax/monoid_assoc_left.tikz}}
 &$\stackrel{(\codiag{}\text{-as})}{=}$& 
    \InputIfFileExists{tapes/whiskered_ax/monoid_assoc_right.tikz}{}{\input{./tikz/tapes/whiskered_ax/monoid_assoc_right.tikz}}
  
    &
    
    \InputIfFileExists{tapes/whiskered_ax/monoid_unit_left.tikz}{}{\input{./tikz/tapes/whiskered_ax/monoid_unit_left.tikz}}
 &$\stackrel{(\codiag{}\text{-un})}{=}$& \Twire{U}
    &
    
    \InputIfFileExists{tapes/whiskered_ax/monoid_comm_left.tikz}{}{\input{./tikz/tapes/whiskered_ax/monoid_comm_left.tikz}}
 &$\stackrel{(\codiag{}\text{-sym})}{=}$& \Tmonoid{U}
    \\
    
    \InputIfFileExists{tapes/cipriano/monoid_assoc_left.tikz}{}{\input{./tikz/tapes/cipriano/monoid_assoc_left.tikz}}
 & $\stackrel{(\diagp{}{}\text{-as})}{=}$ & 
    \InputIfFileExists{tapes/cipriano/monoid_assoc_right.tikz}{}{\input{./tikz/tapes/cipriano/monoid_assoc_right.tikz}}
 
    &
    \scalebox{0.8}{
    \InputIfFileExists{tapes/cipriano/idempotency.tikz}{}{\input{./tikz/tapes/cipriano/idempotency.tikz}}
}&$\stackrel{(\codiag{}\text{-idem})}{=}$&   \Twire{U}
    &
    
    \InputIfFileExists{tapes/cipriano/monoid_comm_left.tikz}{}{\input{./tikz/tapes/cipriano/monoid_comm_left.tikz}}
 & $\stackrel{(\diagp{}{}\text{-sym})}{=}$ & 
    \InputIfFileExists{tapes/cipriano/pcomonoidbar.tikz}{}{\input{./tikz/tapes/cipriano/pcomonoidbar.tikz}}

    \\
    \multicolumn{9}{c}{%
        \scalebox{0.85}{%
        \begin{tabular}{r@{\;}c@{\;}l @{\quad} r@{\;}c@{\;}l @{\quad} r@{\;}c@{\;}l @{\quad} r@{\;}c@{\;}l}
            
    \InputIfFileExists{tapes/cipriano/pcanatleft.tikz}{}{\input{./tikz/tapes/cipriano/pcanatleft.tikz}}
 & $\stackrel{(\diagp{}\text{-nat})}{=}$ & 
    \InputIfFileExists{tapes/cipriano/pcanatright.tikz}{}{\input{./tikz/tapes/cipriano/pcanatright.tikz}}
 
            &
            
    \InputIfFileExists{tapes/cipriano/monoidnatnewright.tikz}{}{\input{./tikz/tapes/cipriano/monoidnatnewright.tikz}}
 & $\stackrel{(\codiag{}\text{-nat})}{=}$ & 
    \InputIfFileExists{tapes/cipriano/monoidnatnewleft.tikz}{}{\input{./tikz/tapes/cipriano/monoidnatnewleft.tikz}}
 
            &
            
    \InputIfFileExists{tapes/cipriano/bangnatleft.tikz}{}{\input{./tikz/tapes/cipriano/bangnatleft.tikz}}
 & $\stackrel{(\,\bangp{}\text{-nat})}{=}$ & 
    \InputIfFileExists{tapes/cipriano/bangnatright.tikz}{}{\input{./tikz/tapes/cipriano/bangnatright.tikz}}
 
            &
            
    \InputIfFileExists{tapes/cipriano/cobangnatleft.tikz}{}{\input{./tikz/tapes/cipriano/cobangnatleft.tikz}}
 & $\stackrel{(\cobang{}\text{-nat})}{=}$ & 
    \InputIfFileExists{tapes/cipriano/cobangnatright.tikz}{}{\input{./tikz/tapes/cipriano/cobangnatright.tikz}}

        \end{tabular}
        }
    }
\end{tabular}
}
\caption{Axioms for tape diagrams.}
\label{fig:tapesax}
\end{figure}

\medskip 
It is now crucial to recall from \cite{bonchi2025tapediagramsmonoidalmonads} that $\per$ can be defined not only on string diagrams but also on tapes: for $\t_1 \colon P \to Q$, $\t_2 \colon R \to S$,  $\t_1 \per \t_2 \defeq \LW{P}{\t_2} ; \RW{S}{\t_1} $ where $\LW{P}{\cdot}$, $\RW{S}{\cdot}$ are the left and right whiskerings, defined as in Table 4 of \cite{bonchi2025tapediagramsmonoidalmonads} (reported in Table~\ref{tab:producttape} in Appendix \ref{app:coherence axioms}). For objects
$P = \Piu[i]{U_i}$ and $Q = \Piu[j]{V_j}$, 
$P \per Q \defeq \Piu[i]{\Piu[j]{U_iV_j}}$.

Theorem 27 in \cite{bonchi2025tapediagramsmonoidalmonads} guarantees that $(\CatT{\DiagT} ,\otimes, \uno)$ is a symmetric monoidal category. Such category is monoidally enriched over $\Cat{PCA}$: 
\[ \t \otimes (\s_1+_p \s_2)= (\t \otimes \s_1)+_p(\t \otimes \s_2) \qquad  (\s_1+_p \s_2)  \otimes \t= ( \s_1 \otimes \t )+_p(\s_2 \otimes  \t) \qquad \star \otimes \t= \t = \t\otimes \star \]
 Most importantly, the two monoidal categories $(\CatT{\DiagT}, \piu, \zero)$ and $(\CatT{\DiagT}, \otimes, \uno)$ interact through the laws of (right strict) rig category (see \cite{laplaza_coherence_1972,johnson2021bimonoidal}).

\begin{theorem}[From \cite{bonchi2025tapediagramsmonoidalmonads}]
$(\,\CatT{\DiagT}, \oplus, \otimes, \zero, \uno \,)$ is a right strict rig category.
\end{theorem}

By means of the isomorphism in Corollary \ref{cor:isotapematrices}, we can equip $\stmat{\DiagT^+}$ with an additional product $\otimes$. This is described as follows. On objects $P\otimes Q$ is defined as above for tapes. On arrows, first define, for all $d_1\in \DiagT^+[U_1,V_1]$ and $d_2\in \DiagT^+[U_2,V_2]$,  $d_1\otimes^+d_2 \in \DiagT^+[U_1U_2,V_1V_2]$ as
\[d_1\otimes^+ d_2 (c) \defeq \sum_{\{(c_1,c_2) \mid c_1\otimes c_2 =c\}} d_1(c_1)\cdot d_2(c_2) \qquad \text{ for }c\in \DiagT[U_1U_2,V_1V_2]\text{.}\]
(recall from Section \ref{sec:pca} that $ \DiagT^+[U,V]=\Dis(\DiagT[U,V])$). Then for matrices, it is defined as the Kronecker product using $\otimes^+$ as multiplication. 
More precisely, let $M\colon\bigoplus_{i=1}^nU_i \to \bigoplus_{i'=1}^{n'}U_{i'}$ and $N\colon\bigoplus_{j=1}^m V_j \to \bigoplus_{j'=1}^{m'} V_{j'}$, $M\otimes N$ is a matrix of size $n'm'\times nm$ whose rows are indexed by pairs $(i',j')\in \{1,\dots, n'\}\times \{1,\dots m'\}$ and columns by pairs $(i,j)\in \{1,\dots, n\}\times \{1,\dots m\}$. Assuming that $M_{i',i}=p_{i',i} \cdot d_{i',i}$
and $N_{j',j}=q_{j',j} \cdot e_{j',j}$, $M\otimes N$ is defined for all indexes $(i',j'),(i,j)$ as
$(M\otimes N)_{(i',j'),(i,j)} \defeq (p_{i',i}\cdot q_{j',j}) \cdot (d_{i',i}\otimes^+ e_{j',j}) $.
\begin{corollary}
$\CatT{\DiagT}$ and $\stmat{\DiagT^+}$ as isomorphic as rig categories.
\end{corollary}

%% file: sections/probbooltapes.tex
\begin{table}[t]
\centering
\resizebox{0.96\linewidth}{!}{%
\begin{tabular}{@{}l@{}}
$\Orgate \defeq (\Notgate \otimes \Notgate); \Andgate; \Notgate 
\qquad 
\Flip{0} \defeq \Flip{1}; \Notgate 
\qquad 

    \InputIfFileExists{tapes/cipriano/multi.tikz}{}{\input{./tikz/tapes/cipriano/multi.tikz}}
 \defeq 
    \InputIfFileExists{tapes/cipriano/multiexpl.tikz}{}{\input{./tikz/tapes/cipriano/multiexpl.tikz}}
$ \\[4pt]
$\raisebox{-1.5em}{
    \InputIfFileExists{tapes/cipriano/mmultiplexerleftzero.tikz}{}{\input{./tikz/tapes/cipriano/mmultiplexerleftzero.tikz}}
} 
\defeq 

    \InputIfFileExists{tapes/cipriano/mmultiplexerrightzero.tikz}{}{\input{./tikz/tapes/cipriano/mmultiplexerrightzero.tikz}}

\qquad 
\raisebox{-1.5em}{
    \InputIfFileExists{tapes/cipriano/mmultiplexerleft.tikz}{}{\input{./tikz/tapes/cipriano/mmultiplexerleft.tikz}}
} 
\defeq 

    \InputIfFileExists{tapes/cipriano/mmultiplexerright.tikz}{}{\input{./tikz/tapes/cipriano/mmultiplexerright.tikz}}

\qquad
\raisebox{-0.5em}{
    \InputIfFileExists{tapes/cipriano/ncopierzero.tikz}{}{\input{./tikz/tapes/cipriano/ncopierzero.tikz}}
} 
\defeq 

    \InputIfFileExists{empty.tikz}{}{\input{./tikz/empty.tikz}}

\qquad 
\raisebox{-0.5em}{
    \InputIfFileExists{tapes/cipriano/ncopier.tikz}{}{\input{./tikz/tapes/cipriano/ncopier.tikz}}
}$
\end{tabular}
}
\caption{Definitions of $\Orgate$, $\Flip{0}$ and multiplexer (top); inductive definitions of $m$-ary multiplexer and $n$-ary copier (bottom).}
\label{table:multiplexer}
\end{table}

\section{Probabilistic Boolean Circuits}\label{sec:probboolcircuits} 

Now, we  quickly recall causal probabilistic boolean circuits from \cite{piedeleu2025boolean}.
Consider the monoidal signature with a single sort, $\sort = \{ A\}$, and  generators \[ PB \defeq \{\Andgate  , \Notgate , \Flip{1},   \CBcopier, \CBdischarger \} \cup\{ \Flip{p} \mid p\in (0,1)\}\text{.}\]  
Arities and coarities are determined by the number of ports on the left and on the right: for instance $\Andgate$ has arity $A^2=AA$ and coarity $A$, while $\CBdischarger $ has arity $A$ and coarity $A^0 =1$. 
The first three generators represent operations and constants of boolean algebras ($\wedge$, $\neg$, $1$), $\CBcopier$ receives a boolean signal on the left and emits two copies on the right, $\CBdischarger$ receives one signal on the left and discards it, while $\Flip{p}$ denotes the distributions on the set $2\defeq \{0,1\}$ mapping $1\mapsto p$ and $0\mapsto 1-p$. 

Formally, the semantics of diagrams is the monoidal functor  $\osem{-}\colon \Diag{ PB} \to \KlD$ defined on objects as $A^n \mapsto 2^n$; for arrows it is defined by first interpreting the generators as arrows in $\KlD$
  \[
    \begin{array}{rclrclrcl}
      \osem{\Andgate}\colon 2 \times 2& \to & 2  & \osem{\Notgate}\colon  2& \to & 2 & \osem{\Flip{1}} \colon 1 & \to & 2 \\
      (x,y)&\mapsto&\delta_{x\wedge y} &  x & \mapsto & \delta_{\neg x} & \bullet &\mapsto&\delta_{1}\\[4pt]
       \osem{\CBcopier} \colon 2  & \to & 2\times 2 & \osem{\CBdischarger} \colon 2 & \to & 1 & \osem{\Flip{p}} \colon 1 & \to & 2  \\
      x &\mapsto&\delta_{(x,x)} & x &\mapsto&\delta_{\bullet} & \bullet &\mapsto& \delta_1 +_p \delta_0
    \end{array}
  \]
and, then inductively by means of the structure of monoidal category $(\KlD, \otimes ,\uno)$: for instance $\osem{c_1\otimes c_2}=\osem{c_1}\otimes \osem{c_2}$ where the second $\otimes$ is the one in \eqref{ex:products}.

Figure 4 in \cite{piedeleu2025boolean} illustrates a sound and complete equational axiomatisation for probabilistic boolean circuits. In a nutshell, these axioms are those of boolean algebras  (reported in Table~\ref{tab:booleanalgebra}) together with 4 extra sophisticated axioms to deal with $\Flip{p}$. It is worth to mention that such axioms crucially rely on \emph{multiplexer}: the  $A^3 \to A$ diagram in Table \ref{table:multiplexer} denoting the boolean function mapping $(x,y,z)$ into $y$ if $x=1$ and into $z$ if $x=0$. Table \ref{table:multiplexer} also illustrates the definition of a $m$-ary multiplexer $A^{2m+1} \to A^m$. This have been used in \cite{Niels} to provide an alternative axiomatisation which is based on the following law \footnote{More precisely, in \cite{Niels}, the law \eqref{eq:natmultiplexer} appears in the equivalent form of naturality of multiplexer.}. 
\begin{equation}\label{eq:natmultiplexer}
\raisebox{-1.5em}{\scalebox{1}{
    \InputIfFileExists{tapes/cipriano/axiompluscorretto.tikz}{}{\input{./tikz/tapes/cipriano/axiompluscorretto.tikz}}
}}\text{.}
\end{equation}

\begin{table}
  \scalebox{0.8}{\begin{tabular}{ccccc}
  \scalebox{0.8}{
    \InputIfFileExists{tapes/cipriano/axA1.tikz}{}{\input{./tikz/tapes/cipriano/axA1.tikz}}
} & &	\scalebox{0.8}{
    \InputIfFileExists{tapes/cipriano/axA2.tikz}{}{\input{./tikz/tapes/cipriano/axA2.tikz}}
} && \scalebox{0.8}{
    \InputIfFileExists{tapes/cipriano/axA3.tikz}{}{\input{./tikz/tapes/cipriano/axA3.tikz}}
}\\
  \midrule
  
    \InputIfFileExists{tapes/cipriano/axB1.tikz}{}{\input{./tikz/tapes/cipriano/axB1.tikz}}
 & &	
    \InputIfFileExists{tapes/cipriano/axB2.tikz}{}{\input{./tikz/tapes/cipriano/axB2.tikz}}
 && 
    \InputIfFileExists{tapes/cipriano/axB3.tikz}{}{\input{./tikz/tapes/cipriano/axB3.tikz}}
\\
  
    \InputIfFileExists{tapes/cipriano/axB4.tikz}{}{\input{./tikz/tapes/cipriano/axB4.tikz}}
 & &	
    \InputIfFileExists{tapes/cipriano/axB5.tikz}{}{\input{./tikz/tapes/cipriano/axB5.tikz}}
 && 
    \InputIfFileExists{tapes/cipriano/axB6.tikz}{}{\input{./tikz/tapes/cipriano/axB6.tikz}}
\\
    &&
    \InputIfFileExists{tapes/cipriano/axB7.tikz}{}{\input{./tikz/tapes/cipriano/axB7.tikz}}
&& \\ \midrule
    
    \InputIfFileExists{tapes/cipriano/axC1.tikz}{}{\input{./tikz/tapes/cipriano/axC1.tikz}}
 && 
    \InputIfFileExists{tapes/cipriano/axC2.tikz}{}{\input{./tikz/tapes/cipriano/axC2.tikz}}
 &&
    \InputIfFileExists{tapes/cipriano/axC3.tikz}{}{\input{./tikz/tapes/cipriano/axC3.tikz}}
\\\midrule
    
    \InputIfFileExists{tapes/cipriano/axD1.tikz}{}{\input{./tikz/tapes/cipriano/axD1.tikz}}
  && 
    \InputIfFileExists{tapes/cipriano/axD2.tikz}{}{\input{./tikz/tapes/cipriano/axD2.tikz}}
 &&
    \InputIfFileExists{tapes/cipriano/axD3.tikz}{}{\input{./tikz/tapes/cipriano/axD3.tikz}}
 \raisebox{-0.5em}{
    \InputIfFileExists{empty.tikz}{}{\input{./tikz/empty.tikz}}
}\\ 
\end{tabular}}
\caption{The monoidal theory of boolean algebras.}\label{tab:booleanalgebra}
\end{table}

In \cite[Example 30]{bonchi2025tapediagramsmonoidalmonads}, an encoding of probabilistic Boolean circuits into tape diagrams is presented.
Let $ B$ be the signature obtained from $ PB$ by dropping $\Flip{p}$ and let $\mathbb{B}$ be the monoidal theory in Table \ref{tab:booleanalgebra} consisting of the usual axioms of boolean algebras. For later use, it is worth to recall (see e.g. \cite[Theorem 3.2]{piedeleu2025boolean})  that $\Diag{\mathbb{B}}$ is isomorphic to $\Sets_2$, the full subcategory of $\Sets$ having as objects the sets $2^n$ for all $n\in \mathbb{N}$. We call boolean circuits the arrows of $\Diag{ B}$ and probabilistic tapes of boolean circuits those of $\CatT{\Diag{ B}}$. The semantics of such tapes is given in \cite[Example 30]{bonchi2025tapediagramsmonoidalmonads} by the morphism of rig categories $\dsem{-}\colon \CatT{\Diag{ B}} \to \KlD$ reported in Figure \ref{eq:SEMANTICA}. The encoding of probabilistic boolean circuits into tapes is the unique monoidal functor $\encoding{-} \colon (\Diag{ PB},\otimes, \uno) \to (\CatT{\Diag{ B}},\otimes, \uno)$ mapping
\[
\resizebox{0.95\linewidth}{!}{$
\Andgate \mapsto \Andgate[t] \qquad
\Notgate \mapsto \Notgate[t] \qquad
\Flip{1} \mapsto \Flip{1}[t] \qquad
\CBcopier \mapsto  \CBcopier[t] \qquad
\CBdischarger \mapsto \CBdischarger[t] \qquad
\Flip{p} \mapsto \Flip{p}[t]
$}
\]
where
\[ \Flip{p}[t] \defeq 
    \InputIfFileExists{tapes/examples/1p0.tikz}{}{\input{./tikz/tapes/examples/1p0.tikz}}
\]
As expected, the encoding preserves the semantics, in the sense that $\dsem{\encoding{c}}=\osem{c}$ for all $c$ in $\Diag{ PB}$.

\begin{figure}[t]
    \renewcommand{\arraystretch}{1.5}
\!\!\!\begin{tabular}{l@{\;\;\;\;}l@{\;\;\;\;}l@{\;\;\;\;}l@{\;\;\;\;}l}
$\CBdsem{\id{A}} \defeq \id{2} $&$ \CBdsem{\id{1}} \defeq \id{1} $&$ \CBdsem{\symmt{A}{A}}  \defeq \symmt{2}{2}  $&$ \CBdsem{c;d}  \defeq \CBdsem{c} ; \CBdsem{d}   $ & $ \CBdsem{c\per d}  \defeq \CBdsem{c}  \per \CBdsem{d} $\\
$\CBdsem{s}  \defeq \osem{s} $& $\CBdsem{\, \tapeFunct{c} \,} \defeq \CBdsem{c}  $ & $\CBdsem{\diagp{A^n}}  \defeq \; \diagp{2^n}$& $\CBdsem{\bang{A^n}} \defeq \bang{2^n}$  \; $\CBdsem{\cobang{A^n}}  \defeq \cobang{2^n}$    & $\CBdsem{\codiag{A^n}} \defeq \codiag{2^n} $ \\
$\CBdsem{\id{A^n}} \defeq \id{2^n} $&$ \CBdsem{\id{\zero}}  \defeq \id{\zero}  $&$ \CBdsem{\symmp{A^n}{A^m}} \defeq \symmp{2^n}{2^m} $&$ \CBdsem{\s;\t}  \defeq \CBdsem{\s}  ; \CBdsem{\t}  $&$ \CBdsem{\s \piu \t}  \defeq \CBdsem{\s}  \piu \CBdsem{\t}  $ 
\end{tabular}
\caption{The semantics $\dsem{-}\colon \CatT{\Diag{ B}} \to \KlD$. Here $s$ is a generator in $ B$.}\label{eq:SEMANTICA}
\end{figure}

\medskip
In this section, we show that the axioms of Boolean algebras and \eqref{eq:natmultiplexer} are enough to characterise the equivalence induced by $\dsem{-}$. 
More precisely, we take the encoding of \eqref{eq:natmultiplexer} in tapes:  for all $\t\colon A^{n+1} \to A^m$, $\t = (\, \id{A} \otimes (\ncopier[]; \t_1 \otimes \t_0)\,); \Ifgatem[]$ where $\t_0=(\Flip{0}[t]\otimes \id{A^n});\t$ and $\t_1=(\Flip{1}[t]\otimes \id{A^n});\t$. Simple computations confirm that the axiom is sound: $\dsem{-}$ maps the left and the right hand side of the equation in the same arrow in $\KlD$.

We write $\eqsyn$ for the congruence (w.r.t. $\oplus$, $\otimes$ and $;$) on $\CatT{\Diag{\mathbb{B}}}$ generated by the above axiom and $\eqsynq{\CatT{\Diag{\mathbb{B}}}}$ for the category obtained as the quotient of $\CatT{\Diag{\mathbb{B}}}$ by $\eqsyn$. We denote by $Q_\eqsyn\colon \CatT{\Diag{\mathbb{B}}} \to \eqsynq{\CatT{\Diag{\mathbb{B}}}}$ the functor mapping each tape into its $\eqsyn$ equivalence class and $Q_\mathbb{B}\colon \Diag{ B} \to \Diag{\mathbb{B}}$ the functor mapping each diagram into its $=_{\mathbb{B}}$ equivalence class. 

Since the axioms are sound $\dsem{-}\colon \CatT{\Diag{ B}} \to \KlD$ can be factored as
\[\xymatrix{\CatT{\Diag{B}} \ar[r]^{\CatT{Q_\mathbb{B}}} & \CatT{\Diag{\mathbb{B}}} \ar[r]^{Q_\eqsyn} & \eqsynq{\CatT{\Diag{\mathbb{B}}}} \ar@{.>}[r]^I& \KlD}\]
for some functor $I$. Proving completeness amounts to prove that  $I$  is faithful. 

By Corollary \ref{cor:isotapematrices} we know that $\CatT{\Diag{\mathbb{B}}} \cong \stmat{\Diag{\mathbb{B}}^+} \cong  \stmat{\Sets_2^+}$. 
Now, let $\KlD_2$ be the full subcategory of $\KlD$ having as objects the sets $2^n$ for all $n\in \mathbb{N}$, and $J\colon \Sets_2 \to \KlD_2$ be the obvious embedding. Since $\KlD_2$ is $\Cat{PCA}$-enriched, by Theorem \ref{thm:freeenriched}, there exists a $\Cat{PCA}$-enriched functor $J^\sharp\colon \Sets_2^+ \to \KlD_2$. This enjoys the following key property.

\begin{lemma}\label{lemma:Jfull}
For all $d_1,d_2\in \Sets_2^+[1, 2^m]$, $J^\sharp(d_1)=J^\sharp(d_2)$ iff $d_1=d_2$.
\end{lemma}
\begin{proof}
It is convenient to give the explicit definition of $J^\sharp$. Recall that an arrow $d\colon 2^n \to 2^m$ in $\Sets_2^+$ is a subprobability distribution on $\Sets_2[2^n,2^m]$, i.e., an element of $\Dis(\Sets_2[2^n,2^m])$. $J^\sharp$ is the identity on objects functor mapping $d\colon 2^n \to 2^m$, into the Kleisli arrow defined for all $b\in 2^n, b'\in 2^m$ as $J^\sharp(d)(b'|b)=\sum_{\{f\mid f(b)=b'\}}d(f)$. It is trivial to see that for $n=0$, $J^\sharp$ provides the obvious bijection $$\Dis(\Sets_2[1,2^m])\cong \Dis(2^m) \cong   \KlD_2[1,2^m]\text{.}$$
%
\end{proof}

We call $F\colon \CatT{\Diag{\mathbb{B}}} \to \stmat{\KlD_2}$, the morphism of convex bicategories obtained by composing the isomorphism $\CatT{\Diag{\mathbb{B}}} \to  \stmat{\Sets_2^+}$ with $\stmat{J^\sharp}$. 

We call $G\colon \stmat{\KlD_2} \to \KlD$ the morphism of convex bicategories obtained by composing the embedding $\stmat{\KlD_2} \to \stmat{\KlD}$ with the counit of the adjunction in Theorem \ref{thm:matfree}, $\epsilon_{\KlD} \colon \stmat{\KlD}\to \KlD$. The latter is, by Proposition \ref{prop:counit fullfaithful}, full and faithful and, therefore, so is $G$. 
\[\xymatrix{
\CatT{\Diag{\mathbb{B}}}  \ar[r]^{\eqsynq{Q}} \ar[d]_{F} & \eqsynq{\CatT{\Diag{\mathbb{B}}}} \ar[d]^{I} \ar@{.>}[dl]|H \\
\stmat{\KlD_2} \ar[r]_G & \KlD
}\]
Since $\eqsynq{Q}$ is identity on objects and $G$ is full and faithful, then by the property of orthogonal factorisation systems (see e.g. \cite{nlab:bofactorizationsystem}), there exists a unique $H\colon  \eqsynq{\CatT{\Diag{\mathbb{B}}}} \to \stmat{\KlD_2}$ making the above diagram commute. Below we prove that $H$ is faithful to conclude that also $I$ is faithful.

Hereafter, we write $\Cat{B}[1,A^n]$ for the set $n$-ary boolean vectors, i.e., all those tapes in $\CatT{\Diag{\mathbb{B}}}[1,A^n]$ of the form $\bigotimes_{i=1}^n b_i$ for $b_i \in \{\Flip{0}[t], \Flip{1}[t]\}$. The following is proved by a simple inductive argument.

\begin{lemma}\label{lemma:booleans}
Let $\s,\t\in \CatT{\Diag{\mathbb{B}}}[A^n,A^m]$. If, for all $\bool\in \Cat{B}[1,A^n]$, $\bool;\s \eqsyn\bool; \t$, then $\s \eqsyn \t$.
\end{lemma}

\begin{theorem}
$H \colon \eqsynq{\CatT{\Diag{\mathbb{B}}}} \to \stmat{\KlD_2}$ is faithful.
\end{theorem}
\begin{proof}
    Let $\s,\t \colon A^n \to A^m$ be arrows in $\CatT{\Diag{\mathbb{B}}}$. Denote by $[\s]_\eqsyn$ and $[\t]_\eqsyn$ their $\eqsyn$ equivalence classes.
\begin{align*}
H([\s]_\eqsyn) = H([\t]_\eqsyn) & \Rightarrow F(\s)=F(\t) \tag{$F=\eqsynq{Q};H$}\\
& \Rightarrow \text{for all }\bool\in \Cat{B}[1,A^n]\text{, } F(\bool);F(\s)=F(\bool);F(\t) \\
& \Rightarrow \text{for all }\bool\in \Cat{B}[1,A^n]\text{, } F(\bool; \s)=F(\bool;\t) \tag{Functoriality} \\
& \Rightarrow  \text{for all }\bool\in \Cat{B}[1,A^n]\text{, }  \bool;\s=\bool; \t \tag{Lemma \ref{lemma:Jfull}}\\
& \Rightarrow \s \sim \t \tag{Lemma~\ref{lemma:booleans}}
\end{align*}
The above derivation proves that $H$ is faithful for arrows of type $A^n \to A^m$. This fact easily entails that $H$ is faithful for arrows $\s,\t \colon A^k \to \bigoplus_{i=1}^nA^{m_i}$. 
Indeed, by Lemma~\ref{decomposition} and Theorem~\ref{thm:TCconvexbiproductcategory}, 
\begin{equation}\label{eq:boh}\s=\langle \s_1, \dots, \s_n \rangle_{\vec{q}} \quad \t=\langle \t_1, \dots, \t_n \rangle_{\vec{p}}\end{equation} 
for some $\vec{q}=q_1,\dots, q_n$, $\vec{p}=p_1,\dots, p_n$, $\s_i \colon A^k \to A^{m_i}$ and $\t_i \colon A^k \to A^{m_i}$. Thus
\begin{align*}
H([\s]_\eqsyn) = H([\t]_\eqsyn) & \Rightarrow F(\s)=F(\t) \tag{$F=\eqsynq{Q};H$}\\
& \Rightarrow \text{for all }i \text{, } F(\s);F(\pi_i)=F(\t);F(\pi_i) \\
& \Rightarrow \text{for all }i\text{, } F(\s;\pi_i)=F(\t;\pi_i) \tag{Functoriality} \\
& \Rightarrow  \text{for all }i\text{, }  \s;\pi_i\sim\t;\pi_i \tag{Previous implication}\\
& \Rightarrow  \text{for all }i\text{, }  q_i \cdot \s_i\sim p_i \cdot t_i \tag{\ref{eq:boh}}\\
& \Rightarrow  \text{for all }i\text{, }  [q_i \cdot \s_i]_\sim = [p_i \cdot t_i]_i 
\end{align*}
By Proposition \ref{cor: quotient category is convex biproduct }, $[\s]_\sim=\langle [\s_1]_\sim, \dots, [\s_n]_\sim \rangle_{\vec{q}}$ and $[\t]_\sim=\langle [\t_1]_\sim, \dots, [\t_n]_\sim \rangle_{\vec{p}}$. Since for all $i$, $[\t];[\pi_i]_\sim$ and  $[\s];[\pi_i]_\sim$ are both equal to  $p_i\cdot[\t_i]_\sim=[p_i \cdot t_i]_\sim=[q_i \cdot \s_i]_\sim$, the universal  property of  n-ary convex products in $\CatT{\Diag{\mathbb{B}}}_\sim$ implies that $[\s]_\sim=[\t]_\sim $.

For arrows of arbitrary type $\bigoplus_{j=1}^o A^{k_j} \to \bigoplus_{i=1}^nA^{m_i}$, one can easily rely on the universal property of coproducts and the case that we just proved.
        \end{proof}

\begin{corollary}\label{cor:final}
$I \colon \eqsynq{\CatT{\Diag{\mathbb{B}}}} \to \KlD$ is faithful.
\end{corollary}


\begin{remark}
Our proof of completeness for probabilistic Boolean circuits via tape diagrams differs substantially from the approach in \cite{piedeleu2025boolean}, which is based on a normal-form argument. Instead, we rely on the isomorphism $\CatT{\Diag{\mathbb{B}}} \cong \stmat{\Diag{\mathbb{B}}^+}$ established in Corollary~\ref{cor:isotapematrices}.
The axioms for tape diagrams involve only monoid and co-pca structures, the laws of Boolean algebra, and the additional multiplexer axiom~\eqref{eq:natmultiplexer}. By contrast, the axioms for $\Flip{p}$ in \cite{piedeleu2025boolean} are more elaborate, as they act directly on the probabilistic structure. In particular, axioms E1, E3, and E4 of \cite{piedeleu2025boolean} follow readily from the monoid, co-pca, and Boolean axioms, while axiom E2 (see Figure~\ref{fig:E2}) requires the multiplexer axiom~\eqref{eq:natmultiplexer} from \cite{Niels}.
\end{remark}

\begin{figure}[t]
    \centering
    \scalebox{0.6}{\includegraphics{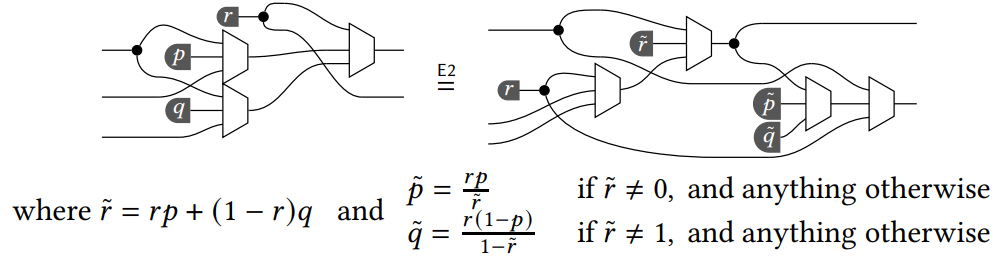} }
    \caption{Axiom E2 from \cite{piedeleu2025boolean}.}
    \label{fig:E2}
\end{figure}

%% file: sections/conclusion.tex
\section{Conclusion}\label{sec:conclusion}  

We have introduced the notion of convex biproduct categories (Def.~\ref{def:convbicat}) and the construction \( \stmat{\Cat{C}} \) of stochastic matrices over a \( \Cat{PCA} \)-enriched category \( \Cat{C} \), showing that \( \stmat{\Cat{C}} \) is the free convex biproduct category generated by \( \Cat{C} \) (Thm.~\ref{thm:matfree}).

We  then presented a syntactic construction: for any category \( \Cat{C} \), the category \( \CatT{\Cat{C}} \) is obtained by freely adding a monoidal structure with natural monoids and co-pointed convex algebras. We showed that \( \CatT{\Cat{C}} \) is the free convex biproduct category over \( \Cat{C} \) (Thm.~\ref{thm:syntacticadjunction}). Combining this with a known result (Thm.~\ref{thm:freeenriched}), we derived the isomorphism \( \CatT{\Cat{C}} \cong \stmat{\Cat{C}^+} \) (Cor.~\ref{cor:isotapematrices}), where arrows in \( \Cat{C}^+ \) are subdistributions of arrows in \( \Cat{C} \).


This construction is particularly relevant when \( \Cat{C} = \DiagS \), a category of string diagrams: in this case, \( \CatT{\DiagS} \) yields the category of probabilistic tape diagrams introduced in~\cite{bonchi2025tapediagramsmonoidalmonads}. The isomorphism \( \CatT{\DiagS} \cong \stmat{\DiagS^+} \) thus provides a concrete interpretation of tape diagrams as stochastic matrices whose entries are probability subdistributions over string diagrams.

Relying on~\cite[Ex. 30]{bonchi2025tapediagramsmonoidalmonads}, we applied this framework to the complete axiomatisation of probabilistic Boolean circuits in~\cite{piedeleu2025boolean}. We proposed an alternative axiomatisation, combining the axioms for Boolean circuits, those for probabilistic tapes, and one additional law from~\cite{Niels}. Thanks to the isomorphism \( \CatT{\DiagS} \cong \stmat{\DiagS^+} \), the resulting completeness proof becomes conceptually simpler and structurally more transparent.

\emph{Related and future work.} The construction of stochastic matrices over pca-enriched categories is conceptually akin to the classical construction of matrices over categories enriched in commutative monoids~\cite{mac_lane_categories_1978,coecke2017two}. However, while the latter is straightforward enough to be left as an exercise in~\cite{mac_lane_categories_1978}, the former involves significantly more work. In particular, finite biproduct categories enjoy a bijective correspondence between arrows \( f \colon A \oplus B \to C \oplus D \) and quadruples of morphisms \( (f_{A,C}, f_{B,C}, f_{A,D}, f_{B,D}) \), where each \( f_{X,Y} \colon X \to Y \). This correspondence breaks down in convex biproduct categories. 
Our probabilistic tape diagrams appear to be closely related to effectuses \cite{introductioneffectus,chophd}. 
Effectuses are based on finitely \emph{partially additive categories} \cite{manes}, that is, categories enriched over partial commutative monoids, rather than pointed convex algebras.  A precise understanding of the relationship between these two frameworks constitutes an interesting direction for future research.

In~\cite{villoria2024enriching}, a general framework is introduced to enrich string diagrams over arbitrary algebraic theories. This approach fundamentally differs from the tape-based construction in~\cite{bonchi2025tapediagramsmonoidalmonads}, in that it does not account for the structure induced by the coproduct \( \oplus \). As a consequence, the correspondence with matrix-based semantics 
does not arise in~\cite{villoria2024enriching}.

We plan to extend probabilistic tape diagrams with traces for \( \oplus \), following the approach developed in~\cite{bonchi2024diagrammatic}. That work showed that a specific property of monoidal traces--called \emph{uniformity}--corresponds to reasoning by invariants in Hoare-style program logics~\cite{hoare1969axiomatic}. Interestingly, as demonstrated in~\cite{jacobs2010coalgebraic}, uniformity holds for traces over \( \oplus \) in \( \KlD \). Preliminary results suggest that, in the probabilistic setting, invariants correspond to sub-martingales, thus offering a promising link between diagrammatic semantics and probabilistic program verification.

Moreover, probabilistic tape diagrams extended with traces are expressive enough to encode probabilistic regular expressions from~\cite{DBLP:conf/lics/RozowskiS24}, which provide a complete axiomatisation for the regular behaviours of generative probabilistic systems--that is, coalgebras in \( \KlD \)~\cite{hasuo2007generic}. We expect that tapes may again offer a more principled and modular axiomatisation, akin to Kleene algebra in~\cite{Kozen94acompleteness}. Interestingly, the completeness proof in~\cite{Kozen94acompleteness} relies heavily on matrix representations over languages; as we hint in Example~\ref{ex:matrices}, our tape formalism naturally supports stochastic matrices over probabilistic languages.

To conclude, we mention that extending our axiomatisation to probabilistic circuits with $\CBcocopier$ (used in \cite{piedeleu2025boolean} for explicit conditioning) is a challenging future work. Indeed, our proof crucially relies on the well-known isomorphism $\Diag{\mathbb{B}} \cong \Sets_2$, while we are not aware of a similar result for partial functions, necessary to accommodate $\CBcocopier$.

%% file: appendices/coherences.tex
\section{Additional Figures and Tables}\label{app:coherence axioms}

\begin{figure}[H]
    \begin{equation}\label{ax:Mon1}\tag{Mon1}
        \input{tikz-cd/monoid_assoc.tikz}
    \end{equation}
    \begin{minipage}[t]{0.50\textwidth}
        \begin{equation}\label{ax:Mon2}\tag{Mon2}
            \input{tikz-cd/monoid_unit.tikz}
        \end{equation}
    \end{minipage}
    \hfill
    \begin{minipage}[t]{0.46\textwidth}
        \begin{equation}\tag{Mon3}
            \input{tikz-cd/monoid_comm.tikz}
        \end{equation}
    \end{minipage}
    \caption{Commutative monoid axioms}
    \label{fig:monoidax}
\end{figure}

 \begin{figure}[H]		
    \begin{equation}\label{eq:coherence codiag}\tag{FC1}
        \begin{tikzcd}[column sep=4.5em,baseline=(current  bounding  box.center)]
        (X \piu Y) \piu (X \piu Y) \ar[dd,"\assoc X Y {X \piu Y}"'] \ar[r,"\codiag{X \piu Y}"] & X \piu Y \\
        & (X \piu X) \piu (Y \piu Y) \ar[u,"\codiag X \piu \codiag Y"']\\
            X \piu (Y \piu (X \piu Y)) \ar[d,"\id X \piu \Iassoc Y X Y"'] & X \piu (X \piu (Y \piu Y)) \ar[u,"\Iassoc X X {Y \piu Y}"']  \\
        X \piu ((Y \piu X) \piu Y) \ar[r,"\id X \piu ( \symm{Y}{X}^{\piu} \piu \id Y)"] & X \piu (( X \piu Y) \piu Y) \ar[u,"\id X \piu \assoc X Y Y"']
        \end{tikzcd}
    \end{equation}
    \\
\begin{minipage}[b]{0.33\textwidth}
	\begin{equation}\label{eq:coherence cobang}\tag{FC2}
        \begin{tikzcd}[baseline=(current  bounding  box.center)]
        \zero \ar[r,"\cobang{X \piu Y}"] \ar[d,"\Ilunit \zero"']  & X \piu Y \\
        \zero \piu \zero \ar[ur,"\cobang X \piu \cobang Y"']  
        \end{tikzcd}
	\end{equation}
\end{minipage}
\hfill
\begin{minipage}[b]{0.26\textwidth}
	\begin{equation}\tag{FC3}
        \begin{tikzcd}
        \zero \piu \zero \ar[r,shift left=2,"\codiag \zero"] \ar[r,shift right=2,"\lunit \zero"'] &  \zero
        \end{tikzcd}
	\end{equation}
\end{minipage}
\hfill
\begin{minipage}[b]{0.26\textwidth}
	\begin{equation}\label{eq:cobang I = id I}\tag{FC4}
        \begin{tikzcd}
        \zero \ar[r,shift left=2,"\cobang \zero"] \ar[r,shift right=2,"\id \zero"'] & \zero
        \end{tikzcd}
	\end{equation}
\end{minipage}
\begin{minipage}[b]{0.45\textwidth}
		\begin{equation}\label{eq:nat monoid1}\tag{Nat. mon1}
\begin{tikzcd}
	{X\piu X} & X \\
	{Y\piu Y} & Y
	\arrow["{\codiag{X}}", from=1-1, to=1-2]
	\arrow["{f\piu f}"', from=1-1, to=2-1]
	\arrow["f", from=1-2, to=2-2]
	\arrow["{\codiag{Y}}"', from=2-1, to=2-2]
\end{tikzcd}
		\end{equation}
	\end{minipage}
	\hfill
	\begin{minipage}[b]{0.45\textwidth}
		\begin{equation}\label{eq:nat monoid2}\tag{Nat. mon2}
\begin{tikzcd}
	\zero & Y & X
	\arrow["{\cobang{Y}}"', from=1-1, to=1-2]
	\arrow["{\cobang{X}}", bend left =30pt, from=1-1, to=1-3]
	\arrow["f"', from=1-2, to=1-3]
\end{tikzcd}
		\end{equation}
	\end{minipage}
\caption{Coherence and naturality axioms for commutative monoids}
\label{fig:fccoherence}
\end{figure}

\begin{figure}[H]
    \begin{equation}\label{ax:PCA1}\tag{PCA1}
\begin{tikzcd}[column sep=large]
	X && {X\piu X} \\
	{X\piu X} & {X\piu (X\piu X)} & {(X\piu X)\piu X}
	\arrow["{\diagp{}}", from=1-1, to=1-3]
	\arrow["{\diagptilde{}}"', from=1-1, to=2-1]
	\arrow["{\diagq{}\piu \id{X}}", from=1-3, to=2-3]
	\arrow["{\id{X}\piu \diagqtilde{}}"', from=2-1, to=2-2]
	\arrow["{\Iassoc{X}{X}{X}}"', from=2-2, to=2-3]
\end{tikzcd} \qquad \tilde{p} = pq, \quad \tilde{q} = \frac{p(1-q)}{1-pq}
    \end{equation}
    \begin{minipage}[t]{0.50\textwidth}
        \begin{equation}\label{ax:PCA2}\tag{PCA2}
\begin{tikzcd}
	X & {X\piu X} \\
	& X
	\arrow["{\diagp{}}", from=1-1, to=1-2]
	\arrow["{\id{X}}"', from=1-1, to=2-2]
	\arrow["{\codiag{X}}", from=1-2, to=2-2]
\end{tikzcd}
        \end{equation}
    \end{minipage}
    \hfill
    \begin{minipage}[t]{0.46\textwidth}
        \begin{equation}\label{ax:PCA3}\tag{PCA3}
\begin{tikzcd}
	& X \\
	{X\piu X} && {X\piu X}
	\arrow["{\diagp{}}"', from=1-2, to=2-1]
	\arrow["{\diagpbar{}}", from=1-2, to=2-3]
	\arrow["{\symm{X}{X}^{\piu}}"', from=2-1, to=2-3]
\end{tikzcd}
        \end{equation}
    \end{minipage}
    \caption{co-pca axioms}
    \label{fig:co-pca axioms}
\end{figure}

\begin{figure}[H]		
	\begin{equation}\label{eq:coherence conv sum}\tag{Coh1}
		\begin{tikzcd}[column sep=4.5em,baseline=(current  bounding  box.center)]
				{X\piu Y} && {(X\piu Y)\piu (X\piu Y)} \\
				{(X\piu X)\piu (Y\piu Y)} && {X\piu (Y\piu (X\piu Y))} \\
				{X \piu (X \piu (Y\piu Y))} \\
				{X \piu ((X \piu Y)\piu Y)} && {X \piu ((Y \piu X)\piu Y)}
				\arrow["{\diagp{}}", from=1-1, to=1-3]
				\arrow["{\diagp{}\piu \diagp{}}"', from=1-1, to=2-1]
				\arrow["{\assoc{X}{X}{Y\piu Y}}"', from=2-1, to=3-1]
				\arrow["{\Iassoc{X}{Y}{X\piu Y}}"', from=2-3, to=1-3]
				\arrow["{\id{X}\piu \Iassoc{X}{Y}{Y}}"', from=3-1, to=4-1]
				\arrow["{\id X \piu ( \symm{X}{Y}^{\piu} \piu \id Y)}"', from=4-1, to=4-3]
				\arrow["{\id{X}\piu \assoc{Y}{X}{Y}}"', from=4-3, to=2-3]
		\end{tikzcd}
	\end{equation}
	\\
	\begin{minipage}[b]{0.33\textwidth}
		\begin{equation}\label{eq:coherence cobang}\tag{Coh2}
			\begin{tikzcd}[baseline=(current  bounding  box.center)]
					{X\piu Y } & \zero \\
					& {\zero \piu \zero}
					\arrow["{\bangp{X\piu Y}}", from=1-1, to=1-2]
					\arrow["{\bangp{X}\piu \bangp{Y}}"', from=1-1, to=2-2]
					\arrow["{\lunit{0}}"', from=2-2, to=1-2]
			\end{tikzcd}
		\end{equation}
	\end{minipage}
	\hfill
	\begin{minipage}[b]{0.26\textwidth}
		\begin{equation}\label{ax:coh3}\tag{Coh3}
		\begin{tikzcd}
			\zero & \zero
			\arrow["{\bangp{\zero}}", shift left, from=1-1, to=1-2]
			\arrow["{\id{\zero}}"', shift right, from=1-1, to=1-2]
		\end{tikzcd}
		\end{equation}
	\end{minipage}
	\hfill
	\begin{minipage}[b]{0.26\textwidth}
		\begin{equation}\label{eq:star I = id I}\tag{Coh4}
			\begin{tikzcd}
				\zero & {\zero\piu \zero}
				\arrow["{\diagp{}}", shift left, from=1-1, to=1-2]
				\arrow["{\Ilunit{\zero}}"', shift right, from=1-1, to=1-2]
			\end{tikzcd}
		\end{equation}
	\end{minipage}
	\begin{minipage}[b]{0.45\textwidth}
		\begin{equation}\label{eq:nat copca1}\tag{Nat. co-pca1}
\begin{tikzcd}
	X & {X\piu X} \\
	Y & {Y\piu Y}
	\arrow["{\diagp{X}}", from=1-1, to=1-2]
	\arrow["f"', from=1-1, to=2-1]
	\arrow["{f\piu f}", from=1-2, to=2-2]
	\arrow["{\diagp{Y}}"', from=2-1, to=2-2]
\end{tikzcd}
		\end{equation}
	\end{minipage}
	\hfill
	\begin{minipage}[b]{0.45\textwidth}
		\begin{equation}\label{eq:nat copca2}\tag{Nat. co-pca2}
\begin{tikzcd}
	X & Y & \zero
	\arrow["f", from=1-1, to=1-2]
	\arrow["{\bangp{X}}"', curve={height=30pt}, from=1-1, to=1-3]
	\arrow["{\bangp{Y}}", from=1-2, to=1-3]
\end{tikzcd}
		\end{equation}
	\end{minipage}
	\caption{Coherence and naturality axioms for co-pca objects}
	\label{fig:copcacoherence}
\end{figure}

\begin{table}[t]
    \begin{center}
    {
        \hfill {\tiny
  \[\begin{array}{c}
  \toprule
        \def\arraystretch{1.2}
        \begin{array}{cc}
            \begin{array}{@{}l}
                \dl{P}{Q}{R} \colon P \per (Q\piu R)  \to (P \per Q) \piu (P\per R) \vphantom{\symmt{P}{Q}} \\
                \midrule
                \dl{\zero}{Q}{R} \defeq \id{\zero} \vphantom{\symmt{P}{\zero} \defeq \id{\zero}} \qquad
                \dl{U \piu P'}{Q}{R} \defeq (\id{U\per (Q \piu R)} \piu \dl{P'}{Q}{R});(\id{U\per Q} \piu \symmp{U\per R}{P'\per Q} \piu \id{P'\per R}) \vphantom{\symmt{P}{V \piu Q'} \defeq \dl{P}{V}{Q'} ; (\Piu[i]{\tapesymm{U_i}{V}} \piu \symmt{P}{Q'})} \\
            \end{array}
            &
            \begin{array}{@{}l}
                \symmt{P}{Q} \colon P\per Q \to Q \per P, \text{ with } P = \Piu[i]{U_i} \\
                \midrule
                \symmt{P}{\zero} \defeq \id{\zero} \qquad
                \symmt{P}{V \piu Q'} \defeq \dl{P}{V}{Q'} ; (\Piu[i]{\tapesymm{U_i}{V}} \piu \symmt{P}{Q'})
                \phantom{\quad}
            \end{array}
        \end{array}\\
        \midrule
        \begin{array}{rclrcl|rclrcl}
            \midrule
            \LW U {\id\zero} &\defeq& \id\zero& \LW U {\t_1 \piu \t_2} &\defeq& \LW U {\t_1} \piu \LW U {\t_2} &  \RW U {\id\zero} &\defeq& \id\zero & \RW U {\t_1 \piu \t_2} &\defeq& \RW U {\t_1} \piu \RW U {\t_2} \\
            \LW U {\tape{c}} &\defeq& \tape{\id U \per c} &    \LW U {\t_1 ; \t_2} &\defeq& \LW U {\t_1} ; \LW U {\t_2}&  \RW U {\tape{c}} &\defeq& \tape{c \per \id U} & \RW U {\t_1 ; \t_2} &\defeq& \RW U {\t_1} ; \RW U {\t_2} \\
            \LW U {\symmp{V}{W}} &\defeq& \symmp{UV}{UW}       &       && & \RW U {\symmp{V}{W}} &\defeq& \symmp{VU}{WU}&  &&  \\
                        \LW U {\diagp V} &\defeq& \diagp{UV} &  \LW U {\bang V} &\defeq& \bang{UV}& \RW U {\diagp V} &\defeq& \diagp{VU} & \RW U {\bang V} &\defeq& \bang{VU} \\
            \LW U {\codiag V} &\defeq& \codiag{UV} &  \LW U {\cobang V} &\defeq& \cobang{UV}& \RW U {\codiag V} &\defeq& \codiag{VU} & \RW U {\cobang V} &\defeq& \cobang{VU} \\
            \hline \hline
            \LW{\zero}{\t} &\defeq& \id{\zero} &  \LW{W\piu S'}{\t} &\defeq& \LW{W}{\t} \piu \LW{S'}{\t} & \RW{\zero}{\t} &\defeq& \id{\zero} &
            \RW{W \piu S'}{\t} &\defeq& \dl{P}{W}{S'} ; (\RW{W}{\t} \piu \RW{S'}{\t}) ; \Idl{Q}{W}{S'} \\
            \hline \hline
            \multicolumn{12}{c}{
                \t_1 \per \t_2 \defeq \LW{P}{\t_2} ; \RW{S}{\t_1}   \quad \text{ ( for }\t_1 \colon P \to Q, \t_2 \colon R \to S   \text{ )}
            }
            \\
            \bottomrule
        \end{array}
    \end{array}        
      \]
      }
      \hfill
      \caption{Inductive definitions of  left distributor $\dl{P}{Q}{R}$ and $\otimes$-symmetry $\symmt{P}{Q}$ (top); monomial and polynomial whiskerings (center); definition of $\per$ (bottom).}\label{tab:producttape}
       }
    \end{center}
\end{table}

%% file: appendices/appendicerefusiprimaversione.tex
\section{Some examples of convex biproduct categories}\label{app:examples}

In this appendix we illustrare several examples of convex biproduct categories:  the Kleisli category for the subdistribution monad,  its continuous version and, finally, a result about finitely partially additive categories.

\subsection{The Kleisli category of the subdistribution monad}

\begin{proposition}\label{prop:Klconvexbiproduct} 
$\KlD$ is a convex biproduct category
\end{proposition}
\begin{proof}
 Coproducts are induced by the coproducts in $\Sets$, and the empty set $\zero$ is both initial and final. Indeed, for every object $X$, the unique arrow $X\to \zero$ sends an element $x\in X$ into the null subdistribution (which is the unique element of $\Dis(\zero)$). 
    
$\KlD$ is $\Cat{PCA}$-enriched since, for every $X,Y\in \KlD$, the hom-set $\KlD[X,Y]$ is a pointed convex algebra with the operations defined as follows: $+_p(f,g)(y|x)\defeq p\cdot f(y|x) + (1-p)\cdot g(y|x)$ for all $f,g\colon X \to Y$ and $p\in (0,1)$, and $\star_{X,Y}(y|x) \defeq 0$ (i.e., it is the arrow which sends every $x$ into the null subdistribution over $Y$). An easy calculation shows that arrow composition is a morphism of pointed convex algebras, i.e., it preserves the operations $+_p$ and $\star_{X,Y}$. For instance $(f+_pg) ;h= (f;h +_p g;h)$ for all $f,g\colon X \to Y$ and $h\colon Y \to Z$ is obtained as follows:
\begin{align}
    ((f+_pg) ;h)(z|x) &= \sum_{y\in Y} (f+_pg)(y|x) \cdot h(z|y) \tag*{} \\
    &= \sum_{y\in Y} (p\cdot f(y|x) + (1-p)\cdot g(y|x)) \cdot h(z|y) \tag*{}\\
    &= p\cdot \sum_{y\in Y} f(y|x) \cdot h(z|y) + (1-p)\cdot \sum_{y\in Y} g(y|x) \cdot h(z|y) \tag*{}\\
    &= p\cdot (f;h)(z|x) + (1-p)\cdot (g;h)(z|x) \tag*{}\\
    &= (f;h +_p g;h)(z|x)\text{.}\tag*{}
    \end{align}
Now observe that $X_1\oplus X_2$ with the projections $\pi_i\colon X_1 \oplus X_2 \to X_i$ for $i=1,2$ given by 
\[ \pi_1(x_1|z)\defeq \begin{cases}
    \delta_{x_1} & z=\iota_1(x_1)\\
    0 & otherwise
\end{cases} \qquad \pi_2(x_2|z)\defeq \begin{cases}
     \delta_{x_2} & z=\iota_2(x_2)\\
        0 & otherwise
\end{cases}\]
 is a convex product of $X_1$ and $X_2$. Indeed, given arrows $f_1\colon A\to X_1$ and $f_2\colon A\to X_2$ and $p_1,p_2 \in (0,1)$ such that $p_1+p_2\le 1$, there exists an arrow $\langle f_1,f_2\rangle_{p_1,p_2}\colon A \to X_1\oplus X_2$ defined as follows:
\[ \langle f_1,f_2\rangle_{p_1,p_2}(z|a) \defeq \begin{cases}
    p_1\cdot f_1(x_1|a) & z=\iota_1(x_1)\\
    p_2\cdot f_2(x_2|a) & z=\iota_2(x_2) \text{.}
\end{cases}\]
Observe that, for $i\in \{1,2\}$,
\begin{align*}
\langle f_1,f_2\rangle_{p_1,p_2} ; \pi_i (x_i | a) = & \sum_{z\in X_1+X_2} \langle f_1,f_2\rangle_{p_1,p_2}(z|a) \cdot \pi_i (x_i | z) \\
= & \langle f_1,f_2\rangle_{p_1,p_2} (\iota_i(x_i)|a)\\
= & p_i\cdot f_i (x_i|a)
\end{align*}
namely, it holds that $(\langle f_1,f_2\rangle_{p_1,p_2} );\pi_i = p_i\cdot f_i$. This is the unique arrow with such property since, for any other arrow $h'\colon A \to X_1\oplus X_2$ such that $h';\pi_i = p_i\cdot f_i$ for all $i=1,2$ it follows that 
\begin{align*}
  h'(\iota_1(x_1)|a) &=    h'(\iota_1(x_1)|a) \cdot 1\\
    & =\sum_{z\in X_1\oplus X_2} h'(z|a) \cdot \pi_1(x_1|z) \tag*{}\\
   &= h' ; \pi_1(x_1|a) \\
    &= p_1\cdot f_1(\iota_1(x_1)|a) \tag*{}
\end{align*}
and similarly for $\pi_2$ one obtains that $h'(\iota_2(x_2)|a) = p_2\cdot f_2(\iota_2(x_2)|a)$, i.e., $h' = \langle f_1,f_2\rangle_{p_1,p_2}$. 
Simple computations confirms that \eqref{eq:delta} and \eqref{eq: assioma aggiuntivo} hold.
\end{proof}



\subsection{Substochastic Markov kernels on standard Borel spaces}

Consider the category of \textit{standard Borel spaces} and measurable functions denoted with $\mathbf{BorelMeas}$.
Given a {standard Borel space} $(X,\Sigma_X)$, denote with $\mathcal{G}_{\le}(X)$ the set of subprobability measures on $X$, i.e. measurable functions $\mu:(X,\Sigma_X)\to ([0,1],\mathcal{B}([0,1]))$ such that $\mu(X)\le 1$, where $\mathcal{B}([0,1])$ is the Borel $\sigma$-algebra on $[0,1]$.  This assignment extends to a functor $\mathcal{G}_{\le}:\mathbf{BorelMeas}\to \mathbf{BorelMeas}$ which corresponds to the subdistribution version of the \textit{Giry monad}~\cite{giry2006categorical}. The functor $\mathcal{G}_{\le}$ is a symmetric monoidal monad and its Kleisli category, $\Kl({\mathcal{G}_{\le}})$, is the category of standard Borel spaces and substochastic Markov kernels. A substochastic Markov kernel from $(X,\Sigma_X)$ to $(Y,\Sigma_Y)$ is a function $f\colon \Sigma_Y\times X\to [0,1]$ such that for all $x\in X$, $f(-|x)\colon \Sigma_Y\to [0,1]$ is a subprobability measure on $Y$ and for all $U\in \Sigma_Y$, $f(U|-)\colon X\to [0,1]$ is a measurable function.
The identity arrow on $(X,\Sigma_X)$ is the Markov kernel $\delta_X$ such that for all $x\in X$ and $U\in \Sigma_X$, $\delta_X(U|x)=1$ if $x\in U$ and $0$ otherwise. The composition of arrows in $\Kl({\mathcal{G}_{\le}})$ is given by the Chapman-Kolmogorov equation: for all  substochastic Markov kernels $f\colon (X,\Sigma_X)\to (Y,\Sigma_Y)$ and $g\colon (Y,\Sigma_Y)\to (Z,\Sigma_Z)$, $f;g\colon (X,\Sigma_X)\to (Z,\Sigma_Z)$ is defined as the Markov kernel such that for all $x\in X$ and $U\in \Sigma_Z$,
\begin{equation*}
(f;g)(U|x)\coloneqq \int_Y g(U|y)\, f(dy|x)
\end{equation*}
   
\begin{proposition}
    \label{prop: SubStoch is convex biproduct}
    $\Kl({\mathcal{G}_{\le}})$ is a convex biproduct category.
\end{proposition}
\begin{proof}
Coproducts in $\Kl({\mathcal{G}_{\le}})$ are given by the coproducts in $\mathbf{BorelMeas}$, i.e. the disjoint union of standard Borel spaces. The initial and terminal object is the empty standard Borel space. The $\mathbf{PCA}$-enrichment is given by the pointwise convex combination of Markov kernels: for all Markov kernels $f,g\colon (X,\Sigma_X)\to (Y,\Sigma_Y)$ and $p\in [0,1]$, $f +_p g\colon (X,\Sigma_X)\to (Y,\Sigma_Y)$ is defined  for all $x\in X$ and $U\in \Sigma_Y$ as  $f+_p g(U|x)\coloneqq p\cdot f(U|x)+(1-p)\cdot g(U|x)$. The zero arrow $\star_{(X,\Sigma_X),(Y,\Sigma_Y)}\colon (X,\Sigma_X)\to (Y,\Sigma_Y)$ is the Markov kernel such that for all $x\in X$ and $U\in \Sigma_Y$, $\star_{(X,\Sigma_X),(Y,\Sigma_Y)}(U|x)=0$. The axioms of $\mathbf{PCA}$ are satisfied since they are so in $[0,1]$. A simple computation shows that $+_p$ and $\star$ are compatible with composition. For instance, as done in Proposition \ref{prop:Klconvexbiproduct}, we obtain that $(f+_p g);h=f;h +_p g;h$ for all Markov kernels $f,g\colon (X,\Sigma_X)\to (Y,\Sigma_Y)$, $h\colon (Y,\Sigma_Y)\to (Z,\Sigma_Z)$ and $p\in [0,1]$ in the following way
\begin{align}
\big((f+_{p}g);h\big)(C\mid x)
&= \int_{Y} h(C\mid y)\,(f+_{p}g)(dy\mid x)
\tag*{} \\
&= \int_{Y} h(C\mid y)\,\big(p\,f(dy\mid x)+(1-p)\,g(dy\mid x)\big)
\tag*{} \\
&= p\int_{Y} h(C\mid y)\,f(dy\mid x) \;+\; (1-p)\int_{Y} h(C\mid y)\,g(dy\mid x)
\tag{linearity of the integral}\\
&= p\,(f;h)(C\mid x) \;+\; (1-p)\,(g;h)(C\mid x)
\tag*{} \\
&= \big((f;h)+_{p}(g;h)\big)(C\mid x).
\tag*{}
\end{align}

Now we prove that $(X_1,\Sigma_{X_1})\overset{\pi_1}{\leftarrow}(X_1,\Sigma_{X_1})\oplus (X_2,\Sigma_{X_2})\overset{\pi_2}{\rightarrow} (X_2,\Sigma_{X_2})$ is a binary convex product in $\Kl({\mathcal{G}_{\le}})$. The Markov kernels $\pi_1$ and $\pi_2$ are defined as follows: for all $z\in X_1\oplus X_2$ and $U\in \Sigma_{X_1}$, $\pi_1(U|z)=1$ if $z\in \iota_1(U)$ and $0$ otherwise; for all $z\in X\oplus Y$ and $V\in \Sigma_{X_2}$, $\pi_2(V|z)=1$ if $z\in \iota_2(V)$ and $0$ otherwise. Given $f\colon (A,\Sigma_A)\to (X_1,\Sigma_{X_1})$, $g\colon (A,\Sigma_A)\to (X_2,\Sigma_{X_2})$ and $p_1,p_2\in [0,1]$ such that $p_1+p_2\leq 1$, we prove that there is a unique Markov kernel $h\colon (A,\Sigma_A)\to (X_1,\Sigma_{X_1})\oplus (X_2,\Sigma_{X_2})$ such that $h;\pi_1 = p_1\cdot f$ and $h;\pi_2 = p_2\cdot g$. Existence is provided by the Markov kernel $h$ defined as follows: for all $a\in A$, $U\in \Sigma_{X_1}$ and $V\in \Sigma_{X_2}$,
\[ h(\iota_1(U)|a)\coloneqq p_1\cdot f(U|a) \quad h(\iota_2(V)|a)\coloneqq p_2\cdot g(V|a) \]
where for a generic element of $\Sigma_{X_1\oplus X_2}$ $h$ is extended by countable additivity. Uniqueness is proved as follows: suppose that there is also $h'$ with the same property. Then, for all $a\in A$, $U\in \Sigma_{X_1}$ and $V\in \Sigma_{X_2}$, then
\begin{align}
    h'(\iota_1(U)\mid a)
    &= \int_{X_1 \oplus X_2} \mathbf{1}_{\iota_1(U)}(z)\, h'(dz\mid a)
    \tag{def. of integration w.r.t.\ $h'(\cdot\mid a)$}\\
    &= \int_{X_1 \oplus X_2} \pi_1(U\mid z)\, h'(dz\mid a)
    \tag{def. of indicator function}\\
    &= (h';\pi_1)(U\mid a)
    \tag{kernel composition}\\
    &= (p_1\cdot k_1)(U\mid a)
    \tag{hp. $h';\pi_1 = p_1k_1$}.
\end{align}

similarly $h'(\iota_2(V)|a)=(p_2\cdot g)(V|a)$. Since every element of $\Sigma_{X_1\oplus X_2}$ is obtained by countable unions of elements of the form $\iota_1(U)$ and $\iota_2(V)$, $h'$ is completely determined by the above equalities, hence $h'=h$.
\end{proof}

%% file: appendices/appconvbiprodcat.tex
\section{Appendix to Section \ref{sec:pca}}\label{app:sec:pca}

\begin{proof}[Proof of Theorem~\ref{thm:freeenriched}]
The result follows from more general results about enriching on arbitrary algebraic theories: see for instance \cite[Prop. 6.4.7]{borceux2} or \cite[Cor. 1]{villoria2024enriching}. It is anyway convenient to show the involved structures.

We first illustrate the definition of $F^+\colon \Cat{C}^+ \to \Cat{D}^+$ for a functor $F\colon \Cat{C} \to \Cat{D}$. For all objects $X$, $F^+$ is defined as $F^+(X)\defeq F(X)$; For all $d\in \Cat{C}^+[X,Y]$, $F^+(d)\in \Cat{D}^+[FX, FY]$ is defined for all $h\colon FX \to FY$ as $F^+(d)(g)\defeq \sum_{\{f\in \Cat{C}[X,Y]| F(f)=g\}}d(f)$. 
One can easily check that the functor is PCA-enriched.

For all categories $\Cat{C}$, the unit $\eta_{\Cat{C}}\colon \Cat{C} \to \Cat{C}^+$ is the functor acting as identity on objects and mapping any arrow $f\in \Cat{C}[X,Y]$ into $\delta_f \in \Cat{C^+}[X,Y]$.

Now, given a PCA-enriched category $\Cat{D}$ and a functor $F\colon \Cat{C} \to U(\Cat{D})$, one can define a functor $F^\sharp \colon \Cat{C}^+ \to \Cat{D}$ as follows: for all objects $X$, $F^\sharp(X)\defeq F(X)$, for all arrow $d\in \Cat{C}^+[X,Y]$, $F^\sharp(d)\defeq \sum_{f\in \Cat{C}[X,Y]}d(f)\cdot \delta_{F(f)}$.   
\end{proof}

\section{Appendix to Section \ref{sec:cbproducts}}\label{app:sec:cbproducts}

\begin{proof}[Proof of Lemma~\ref{lemma:naryconvexproduct}]
By induction on $n$, we show that $\bigoplus_{i=1}^n X_i$ is a convex product of $X_1, \dots, X_n$. For $n=0$, observe that the object $\zero$ is final by Definition~\ref{def:convbicat}.
For $n+1$ assume by induction hypothesis that $\bigoplus_{i=1}^n X_i$ is a convex product of $X_1, \dots, X_n$ with projection $\pi_i\colon \bigoplus_{i=1}^n X_i \to X_i$. Since $\Cat{C}$ has binary convex product we can take the convex product $(\bigoplus_{i=1}^n X_i) \oplus X_{n+1}$ with projections $\pi'_1\colon (\bigoplus_{i=1}^n X_i) \oplus X_{n+1} \to \bigoplus_{i=1}^n X_i$ and $\pi'_2\colon (\bigoplus_{i=1}^n X_i) \oplus X_{n+1} \to X_{n+1}$. We claim that $(\bigoplus_{i=1}^n X_i) \oplus X_{n+1}$ with projections $\pi''_i\colon (\bigoplus_{i=1}^n X_i) \oplus X_{n+1} \to X_i$ defined for all $i\in 1\dots n+1$ as 
\[\pi''_i = \begin{cases} \pi_1'; \pi_i & i\in 1,\dots n \\ \pi_2' & i=n+1 \end{cases}\]
is a convex product of $X_1, \dots, X_{n+1}$. 
To check the universal property, let $\vec{p}=p_1,\dots, p_{n+1}$ such that $\sum_{i=1}^{n+1}p_i\leq 1$ and $f_i\colon A \to X_i$. Take $\vec{q}=q_1,\dots, q_n$ with $q_i=\frac{p_i}{1-p_{n+1}}$, since 
$\bigoplus_{i=1}^n X_i$ is a $n$-ary convex product by induction hypothesis, there exists an arrow $\langle f_1, \dots ,f_n\rangle_{\vec{q}} \colon A \to \bigoplus_{i=1}^n X_i$ such that $p_i\cdot f_i = \langle f_1, \dots ,f_n\rangle_{\vec{q}} ; \pi_i$. For the $n+1$-ary, we construct the mediating morphism $h\colon A \to (\bigoplus_{i=1}^n X_i) \oplus X_{n+1}$ as $\langle \, \langle f_1, \dots ,f_n\rangle_{\vec{q}} \,,\, f_{n+1} \, \rangle_{1-p_{n+1},p_{n+1}}$.
 \end{proof}

\begin{lemma}\label{lemma:initialproperties}
In a convex biproduct category the following hold:
\begin{enumerate}
\item $\star_{X,Y} = \bang{X};\cobang{Y}$\label{lemma:starxy}
\item  $\pi_1=(\id{X_1}\oplus \bang{X_2});\runit{X_1}$ for $\pi_1\colon X_1\oplus X_2 \to X_1$;\label{lemma:pi1}
\item $\pi_2 =(\bang{X_1}\oplus \id{X_2});\lunit{X_2}$  for $\pi_2\colon X_1\oplus X_2 \to X_2$;\label{lemma:pi2}
\item $p \cdot f =\, \diagp{X}; (f \oplus \bang{X} ) ; \runit{Y} $ for all $f\colon X \to Y$;\label{lemma:pf}
\end{enumerate}
Moreover, in a PCA-enriched category the following hold:
\begin{enumerate}[start=5]
\item $(1-p)\cdot  f= \diagp{X} ;  \, (\bang{X}  \oplus f \,)  ; \lunit{y}$\  for all $f\colon X \to Y$; \label{lemma:1-pf}
\item $(p\cdot f) ; g =p\cdot (f;g)$ and $f;(p\cdot g) = p\cdot ( f;g)$ ; \label{lemma:p;}
\item $p\cdot (q \cdot f) =pq \cdot f$; \label{lemma:pq}
\item $q\cdot\sum_{i=1}^{n}p_i\cdot f_i= \sum_{i=1}^{n} (qp_i)\cdot f_i$;\label{lemma q per somma}
\end{enumerate}
\end{lemma}
\begin{proof}
We prove below item by item.
\begin{enumerate}
\item Follow easily by enrichment and finality of $\zero$: 
\begin{align*}
\star_{X,Y} & = \star_{X,\zero} ; \cobang{Y} \tag{\ref{eq:enr}}\\
&= \bang{X};\cobang{Y} \tag{$\zero$ is final}
\end{align*}
\item Observe that, for $\iota_i\colon X_i \to X_1\oplus X_2$, 
\[\iota_1 ; (\id{X_1}\oplus \bang{X_2});\runit{X_1}=\id{X_1} \qquad \iota_2 ;  (\id{X_1}\oplus \bang{X_2});\runit{X_1} = \bang{X_2}; \cobang{X_1}= \star_{X_2,X_1}\]
namely $(\id{X_1}\oplus \bang{X_2});\runit{X_1} = [id_{X_1}, \star_{X_2,X_1}]$. Now, by \eqref{eq:delta}, one has that $[id_{X_1}, \star_{X_2,X_1}]=\pi_1$.
\item Analogous to the point above.
\item By the following derivation.
\begin{align*}
p \cdot f &= f+_p\star_{X,Y} \tag{def. of $p\cdot(-)$}\\
&=\, \diagp{X} ; (f \oplus \star_{X,Y}) ; \codiag{Y} \tag{\ref{eq: assioma aggiuntivo}} \\ 
&=\, \diagp{X} ; (f \oplus (\bang{X};\cobang{Y})) ; \codiag{Y} \tag{Lemma \ref{lemma:initialproperties}.1} \\ 
&=\, \diagp{X} ; (f \oplus \bang{X}) ;(\id{Y} \oplus \cobang{Y}) ; \codiag{Y} \tag{Symmetric Monoidal Category} \\ 
&=\, \diagp{X} ; (f \oplus \bang{X}) ;\runit{Y} \tag{\ref{ax:Mon2}} \\ 
 \end{align*}
 \item By the following derivation
\begin{align*}
(1-p) \cdot f &= f+_{1-p}\star_{X,Y} \tag{def. of $p\cdot(-)$}\\
&= \star_{X,y} +_{p} f \tag{\ref{eq:pca}}\\
&=\, \diagp{X} ; (\star_{X,Y} \oplus f) ; \codiag{Y} \tag{\ref{eq: assioma aggiuntivo}} \\ 
&=\, \diagp{X} ; ((\bang{X};\cobang{Y}) \oplus f) ; \codiag{Y} \tag{Lemma \ref{lemma:initialproperties}.1} \\ 
&=\, \diagp{X} ; (\bang{X} \oplus f) ;(\id{Y} \oplus \cobang{Y}) ; \codiag{Y} \tag{Symmetric Monoidal Category} \\ 
&=\, \diagp{X} ; ( \bang{X} \oplus f) ;\lunit{Y} \tag{\ref{ax:Mon2}} \\ 
 \end{align*} 
 \item  Follow easily by the enrichment. For instance, the first equality is proved as follows.
\begin{align*} 
(p\cdot f) ; g & = (f+_p\star_{X,Y});g \tag{def of $p\cdot(-)$} \\
&= (f;g) +_p(\star_{X,Y};g) \tag{\ref{eq:enr}}\\
&= (f;g) +_p(\star_{X,Z}) \tag{\ref{eq:enr}}\\
&= p\cdot (f ; g) \tag{def of $p\cdot(-)$} \\
\end{align*}
\item Follow easily by the laws of PCAs.
\begin{align*}
p\cdot (q \cdot f) &= (f+_q \star )+_p \star \tag{def. of $p\cdot(-)$}\\
&= f+_{pq} (\star +_{\frac{p\cdot(1-q)}{1-pq}} \star) \tag{\ref{eq:pca}} \\
&= f+_{pq} \star \tag{\ref{eq:pca}} \\
&= pq \cdot f \tag{def. of $p\cdot(-)$}
\end{align*}
\item Follows by induction. The case $n=0$ follows by laws of PCA: $q\cdot \star= \star +_q \star= \star$. For $n+1$ consider 
\begin{align}
    q\cdot\sum_{i=1}^{n+1}p_i\cdot f_i &= q\cdot ( f_{1} +_{p_{1}}\sum_{j=1}^{n}q_j\cdot f_j) \tag{Def. of $\sum_{i=1}^{n+1}{p_i}\cdot f_i$ in (\ref{eq:def somma n pca})}\\
    &= q\cdot f_{1} +_{p_{1}} \sum_{j=1}^{n}qq_j\cdot f_j \tag{PCA laws + Ind. Hp.}\\
    &= \sum_{i=1}^{n+1}qp_i\cdot f_i \tag{Def. of $\sum_{i=1}^{n+1}{p_i}\cdot f_i$ in (\ref{eq:def somma n pca})}
\end{align}
\end{enumerate}
\end{proof}

\begin{lemma}\label{lemma:naturality}
In a convex biproduct category, $\bang{}$ and $\diagp{}$ form natural transformations, namely 
$f;\bang{Y}=\bang{X}$  and $f;\diagp{Y} =\, \diagp{X};(f\oplus f)$ for all $f\colon X\to Y$.
\end{lemma}
\begin{proof}
The equation $f;\bang{Y}=\bang{X}$ for $f\colon X \to Y$ follows from the fact that $\zero$ is final. The equation
\[f;\diagp{Y} =\, \diagp{X};(f\oplus f)\]
follows from the universal property of convex products. Indeed, observe that
\begin{align*}
    \diagp{X};(f\oplus f);\pi_1 &=\, \diagp{X};(f\oplus f);(\id{Y}\piu \bang{Y});\runit{X} \tag{Lemma \ref{lemma:initialproperties}.\ref{lemma:pi1}}\\
    &=\, \diagp{X};(f\oplus (f;\bang{Y}));\runit{X}  \tag{Symmetric Monoidal Category}\\
    &=\, \diagp{X};(f\oplus \bang{X});\runit{X} \tag{Naturality of $\bang{X}$}\\
    &= p\cdot f \tag{Lemma \ref{lemma:initialproperties}.\ref{lemma:pf}}
\end{align*}
\begin{align*}
    \diagp{X};(f\oplus f);\pi_2 &=\, \diagp{X};(f\oplus f);(\bang{Y} \piu \id{Y} );\lunit{X} \tag{Lemma \ref{lemma:initialproperties}.\ref{lemma:pi2}}\\
    &=\, \diagp{X};( (f;\bang{Y})\oplus f);\lunit{X}  \tag{Symmetric Monoidal Category}\\
    &=\, \diagp{X};( \bang{X} \oplus f);\lunit{X} \tag{Naturality of $\bang{X}$}\\
    &= (1-p)\cdot f \tag{Lemma \ref{lemma:initialproperties}.\ref{lemma:1-pf}}
\end{align*}
and that
\begin{align*}
f;\diagp{Y};\pi_1 &= f ; \langle \id{X}, \id{X} \rangle_{p,(1-p)};\pi_1 \tag{def. of $\diagp{}$}\\
&= f; (p\cdot \id{X}) \tag{convex product} \\
&= p\cdot(f;\id{X}) \tag{Lemma \ref{lemma:initialproperties}.\ref{lemma:p;}}\\
&=p\cdot f \tag{Category}
\end{align*}
\begin{align*}
f;\diagp{Y};\pi_2 &= f ; \langle \id{X}, \id{X} \rangle_{p,(1-p)};\pi_2 \tag{def. of $\diagp{}$}\\
&= f; (\,(1-p)\cdot \id{X}\,) \tag{convex product} \\
&= (1-p)\cdot(f;\id{X}) \tag{Lemma \ref{lemma:initialproperties}.\ref{lemma:p;}}\\
&=(1-p)\cdot f \tag{Category}
\end{align*}
\end{proof}

\begin{lemma}\label{lemma:coherence}
In a convex biproduct category, the coherence axioms in Figure~\ref{fig:copcacoherence} hold.
\end{lemma}
\begin{proof}
(Coh2), (Coh3) and (Coh4) follow from the fact that $\zero$ is both initial and final. (Coh1) follows from the naturality and the universal property of coproducts. 

Let $s\defeq \assoc{X}{X}{Y\piu Y}; ( {\id{X}\piu \Iassoc{X}{Y}{Y}} ); ({\id X \piu ( \symm{X}{Y}^{\piu} \piu \id Y)} ); (\id{X}\piu \assoc{Y}{X}{Y}); \Iassoc{X}{Y}{X\oplus Y}$. Let
\[\iota_1^{(X\oplus X)\oplus(Y\oplus Y)}\colon X\piu X \to (X\piu X)\piu (Y\piu Y) \text{ and }\iota_2^{(X\oplus X)\oplus(Y\oplus Y)}\colon Y\piu Y \to (X\piu X)\piu (Y\piu Y)\] be the injections. By using the definitions of the structural isomorphisms induced by the coproducts, one can easily check that
\begin{equation}\label{eq:local}
\iota_1^{(X\oplus X)\oplus(Y\oplus Y)}; s= \iota_1 \oplus \iota_1 \qquad \text{and} \qquad \iota_2^{(X\oplus X)\oplus(Y\oplus Y)}; s= \iota_2 \oplus \iota_2\text{.}
\end{equation}
Then
\begin{align*}
\iota_1; (\diagp{X}\piu\, \diagp{Y}) ;s &= \iota_1; [\diagp{X}; \iota_1^{(X\oplus X)\oplus(Y\oplus Y)} \, , \, \diagp{Y}; \iota_2^{(X\oplus X)\oplus(Y\oplus Y)}] ;s \tag{def. of $\oplus$}\\
& =\, \diagp{X};\iota_1^{(X\oplus X)\oplus(Y\oplus Y)};s  \tag{coproduct}\\
&=\, \diagp{X}; (\iota_1 \oplus \iota_1) \tag{\ref{eq:local}} \\
&= \iota_1;\,\diagp{X\piu Y} \tag{Lemma \ref{lemma:naturality}}
\end{align*}
and 
\begin{align*}
\iota_2; (\diagp{X}\piu\, \diagp{Y}) ;s &= \iota_2; [\diagp{X}; \iota_1^{(X\oplus X)\oplus(Y\oplus Y)} \, , \, \diagp{Y}; \iota_2^{(X\oplus X)\oplus(Y\oplus Y)}] ;s \tag{def. of $\oplus$}\\
& =\, \diagp{Y};\iota_2^{(X\oplus X)\oplus(Y\oplus Y)};s  \tag{coproduct}\\
&=\, \diagp{X}; (\iota_2 \oplus \iota_2) \tag{\ref{eq:local}} \\
&= \iota_2;\,\diagp{X\piu Y} \tag{Lemma \ref{lemma:naturality}}
\end{align*}
Hence, the universal property of coproducts implies that $\diagp{X\piu Y} = \diagp{X}\piu \diagp{Y};s$.
\end{proof}

\begin{lemma}\label{lemma:pdiag} \label{lemma:diagp}
In a convex biproduct category the following hold:
\begin{enumerate}
\item $p \cdot \diagq{X} = \langle \id{X},\id{X}\rangle_{pq, p(1-q)} $.
\item $\diagp{X}; (\id{X} \oplus q\cdot \id{X})=\langle \id{X},\id{X}\rangle_{p,(1-p)q}$ 
\end{enumerate}
\end{lemma}
\begin{proof}
We prove below item by item.
\begin{enumerate}
\item By the uniqueness of $\langle \id{X},\id{X}\rangle_{pq, p(1-q)}$ and the following two derivations.
\begin{align*}
(p \cdot \diagq{X})  ; \pi_1 &= p \cdot (\diagq{X};\pi_1) \tag{Lemma \ref{lemma:initialproperties}.4}\\
&=p\cdot (\langle \id{X},\id{X}\rangle_{q,1-q} ; \pi_1) \tag{def. of $\diagq{}$}\\
&= p\cdot (q \cdot \id{X}) \tag{convex product}\\
&= (p q) \cdot \id{X} \tag{\ref{lemma:pq}}
\end{align*}
\begin{align*}
(p \cdot \diagq{X})  ; \pi_2 &= p \cdot (\diagq{X};\pi_2) \tag{Lemma \ref{lemma:initialproperties}.4}\\
&=p\cdot (\langle \id{X},\id{X}\rangle_{q,1-q} ; \pi_2) \tag{def. of $\diagq{}$}\\
&= p\cdot ((1-q) \cdot \id{X}) \tag{convex product}\\
&= (p (1-q)) \cdot \id{X} \tag{\ref{lemma:pq}}
\end{align*}
\item By uniqueness of $\langle \id{X},\id{X}\rangle_{p,(1-p)q}$ and the following two derivations. 
\begin{align*}
\diagp{X}; (\id{X} \oplus q\cdot \id{X}) ; \pi_1 &=\, \diagp{X}; (\id{X} \oplus q\cdot \id{X}) ; (\id{X}\oplus \bang{X});\runit{X} \tag{Lemma \ref{lemma:initialproperties}.\ref{lemma:pi1}}\\
&=\, \diagp{X}; (\id{X} \oplus (q\cdot \id{X} ; \bang{X})); \runit{X} \tag{Symmetric Monoidal Category}\\
&=\, \diagp{X}; (\id{X} \oplus  \bang{X}); \runit{X} \tag{Naturality of $\bang{}$} \\
&= p\cdot \id{X} \tag{Lemma \ref{lemma:initialproperties}.\ref{lemma:pf}}
\end{align*}
\begin{align*}
\diagp{X}; (\id{X} \oplus q\cdot \id{X}) ; \pi_2 &=\, \diagp{X}; (\id{X} \oplus q\cdot \id{X}) ; (\bang{X}\oplus \id{X});\lunit{X} \tag{Lemma \ref{lemma:initialproperties}.\ref{lemma:pi2}}\\
&=\, \diagp{X}; ( \bang{X}\oplus( q\cdot \id{X}) ) ; \lunit{X} \tag{Symmetric Monoidal Category}\\
&= (1-p)\cdot (q \cdot \id{X}) \tag{Lemma \ref{lemma:initialproperties}.\ref{lemma:1-pf}}\\
&= ((1-p)q) \cdot \id{X} \tag{Lemma \ref{lemma:initialproperties}.\ref{lemma:pq}}
\end{align*}
\end{enumerate}
\end{proof}

\begin{proof}[Proof of Proposition \ref{lemma: copca objects in convbicat}]
We proved naturality and coherence in Lemmas~\ref{lemma:naturality} and \ref{lemma:coherence}. Below we prove the axioms in Figure~\ref{fig:co-pca axioms}.

    (PCA1) in Figure~\ref{fig:co-pca axioms} follows from the universal property of convex products. Indeed, observe that 
    \begin{align*}
        \diagp{X} ; (\diagq{X} \oplus \id{X}) ; \pi_1 &=\, \diagp{X} ; (\diagq{X} \oplus \id{X}) ;( \id{X \oplus X} \oplus \bang{X}) ; \runit{X\oplus X} \tag{Lemma \ref{lemma:initialproperties}.\ref{lemma:pi1}}\\
        &=\, \diagp{X} ; (\diagq{X} \oplus \bang{X}) ; \runit{X\oplus X} \tag{Symmetric Monoidal Category}\\
        &= p \cdot \diagq{} \tag{Lemma \ref{lemma:initialproperties}.\ref{lemma:pf}}\\
        &= \langle \id{X}, \id{X} \rangle_{pq,\,p(1-q)}\tag{Lemma \ref{lemma:pdiag}}
    \end{align*}
    
        \begin{align*}
        \diagp{X} ; (\diagq{X} \oplus \id{X}) ; \pi_2 &=\, \diagp{X} ; (\diagq{X} \oplus \id{X}) ;( \bang{X \oplus X} \oplus \id{X}) ; \lunit{X} \tag{Lemma \ref{lemma:initialproperties}.\ref{lemma:pi2}}\\
        & =\, \diagp{X} ; ( \, (\diagq{X} ; (\bang{X} \oplus \bang{X}) ) \oplus \id{X} \,)  ; \lunit{X} \tag{Symmetric Monoidal Category}\\
        &=\, \diagp{X} ;  (\bang{X}  \oplus \id{X} )  ; \lunit{X} \tag{Lemma \ref{lemma:naturality}} \\
        &=  (1-p)\cdot \id{X} \tag{Lemma \ref{lemma:initialproperties}.\ref{lemma:1-pf}}
    \end{align*}
    
Similarly, by fixing  $\tilde{p}\defeq  pq$ and $\tilde{q} \defeq \frac{p(1-q)}{1-pq}$, it holds  that
\begin{align}
    \diagptilde{}; (\id{X}\oplus\, \diagqtilde{}) ; \Iassoc{X}{X}{X}; \pi_1 &=\, \diagptilde{}; (\id{X}\oplus\,  \diagqtilde{}) ; \Iassoc{X}{X}{X}; ( \id{X \oplus X} \oplus \bang{X}) ; \runit{X\oplus X} \tag{Lemma \ref{lemma:initialproperties}.\ref{lemma:pi1}}\\
     &=\, \diagptilde{}; (\id{X}\oplus\, \diagqtilde{}) ; ( \id{X} \oplus (\id{X} \oplus \bang{X})) ; \runit{X} \tag{Symmetric Monoidal Category}\\
 &=\, \diagptilde{}; (\,\id{X}\oplus ( \diagqtilde{} ;(\id{X} \oplus \bang{X})) \,); \runit{X} \tag{Symmetric Monoidal Category}\\
&=\, \diagptilde{}; (\,\id{X}\oplus \tilde{q}\cdot \id{X} \,) \tag{Lemma \ref{lemma:initialproperties}.\ref{lemma:pf}}\\
    &=\langle \id{X}, \id{X} \rangle_{\tilde{p},(1-\tilde{p})\tilde{q}} \tag{Lemma \ref{lemma:diagp}}\\
    &= \langle \id{X}, \id{X} \rangle_{pq,p(1-q)} \tag{def of $\tilde{p}$ and $\tilde{q}$}
\end{align} 
    \begin{align}
    \diagptilde{}; (\id{X}\oplus\, \diagqtilde{}); \Iassoc{X}{X}{X};\pi_2&=\, \diagptilde{};( \id{X}\oplus\, \diagqtilde{}); \Iassoc{X}{X}{X}; (\bang{X\piu X}\piu \id{X});\lunit{X}  \tag{Lemma \ref{lemma:initialproperties}.\ref{lemma:pi2}}\\
    &=\, \diagptilde{};( \id{X}\oplus\, \diagqtilde{}); \Iassoc{X}{X}{X}; ((\bang{X}\piu \bang{X})\piu \id{X});\lunit{X}  \tag{Lemma \ref{lemma:coherence}}\\
    &=\, \diagptilde{}; (\bang{X} \oplus \id{X}); \lunit{X} ; \diagqtilde{X}; (\bang{X}\piu \id{X}) ;\lunit{X}  \tag{Symmetric Monoidal Category}\\
    &=\, \diagptilde{}; (\bang{X} \oplus \id{X}); \lunit{X} ; (1-\tilde{q})\cdot \id{X}  \tag{Lemma \ref{lemma:initialproperties}.\ref{lemma:1-pf}}\\
    &=\, \diagptilde{}; (\bang{X} \oplus ( (1-\tilde{q})\cdot \id{X} )) ;  \lunit{X} \tag{Symmetric Monoidal Category}\\
    &= (1-\tilde{p}) \cdot ((1-\tilde{q})\cdot \id{X} )  \tag{Lemma \ref{lemma:initialproperties}.\ref{lemma:1-pf}}\\
    &=  (1-\tilde{p}) (1-\tilde{q}) \cdot \id{X} \tag{Lemma \ref{lemma:initialproperties}.\ref{lemma:pq}}\\
    &= (1-p) \cdot \id{X} \tag{def. $\tilde{p}$ and $\tilde{q}$}
\end{align}
    Hence, $\diagp{} ; (\diagq{} \oplus \id{X}) =\, \diagptilde{} ; \id{X}\oplus\, \diagqtilde{}; \Iassoc{X}{X}{X}$.\\

    (PCA2) in Figure~\ref{fig:co-pca axioms} follows from the coherence of the PCA-enrichment (i.e., property \eqref{eq: assioma aggiuntivo}):
  \begin{align*}
    \diagp{X}; \codiag{X} &=\,   \diagp{X}; (\id{X} \oplus \id{X}) \codiag{X}\\
    &= \id{X}+_p \id{X} \tag{\ref{eq: assioma aggiuntivo}}\\
    &=\id{X} \tag{\ref{eq:pca}}
    \end{align*}
    \\

    (PCA3) in Figure~\ref{fig:co-pca axioms} follows from the universal property of convex products. Indeed, observe that
    \begin{align*}
        \diagp{X};\symm{X}{X};\pi_1 &=\, \diagp{X};\symm{X}{X};(\id{X}\piu \bang{X});\runit{X} \tag{Lemma \ref{lemma:initialproperties}.\ref{lemma:pi1}}\\
        &=\, \diagp{X};(\bang{X} \piu \id{X});\lunit{X} \tag{Symmetric Monoidal Categories}\\
        &= (1-p)\cdot \id{X} \tag{Lemma \ref{lemma:initialproperties}.\ref{lemma:1-pf}}
    \end{align*}
    And similarly for $\pi_2$ we have that $\diagp{X};\symm{X}{X};\pi_2 = p\cdot \id{X}$. Hence, the universal property of convex products implies that $\diagp{X};\symm{X}{X} =\, \diagpbar{X}$.
 \end{proof}

\begin{proposition}\label{prop: functor1}
    Let $F\colon \Cat{C}\to \Cat{D}$ a functor between convex biproduct categories. If $F$ preserves finite coproducts then
     $F$ preserves convex binary products.
   \end{proposition}
   \begin{proof}
    Preservation of finite coproducts is equivalent to preservation of binary coproducts and initial object. Hence, assume that $F(X)\overset{F(\iota_1)}{\to}F(X\piu Y)\overset{F(\iota_2)}{\leftarrow}$ is a coproduct of $F(X)$ and $F(Y)$, which is equivalent to requiring that the arrow $[F(\iota_1),F(\iota_2)]\colon F(X)\piu F(Y)\to F(X\piu Y)$ is an isomorphism (call $k$ the inverse). And assume also that $F(\zero_\Cat{C})$ is initial in $\Cat{D}$ (hence $\cobang{F(\zero_{\Cat{C}})}$ is an isomorphism). Now observe that
    \begin{align}
        F(\pi_1^{X\piu Y})&= F(\id{X}\piu\, \bangp{Y};\rho_X) \tag{Lemma \ref{lemma:initialproperties}.\ref{lemma:pi1}}\\
        &=F(\id{X}\piu\, \bangp{Y});F([\id{X}, \cobang{X}]) \tag{def. $\rho$}\\
        &= F([\id{X}, \bangp{Y};\cobang{X}]) \tag{Coproducts}\\
        &=k;[\id{F(X)}, F(\,\bangp{Y};\cobang{X})] \tag{coproducts and $k$ iso}\\
        &= k;(\id{F(X)}\oplus\, \bangp{F(Y)});[\id{F(X)},\cobang{F(X)}] \tag{$F(\zero_\Cat{C})\cong\zero_\Cat{D}$}\\
        &= k;\pi_1^{F(X)\piu F(Y)} \tag{def $\rho$}
    \end{align}
    and similarly $F(\pi_2^{X\piu Y})=k;\pi_2^{F(X)\piu F(Y)}$. We can now prove that $F(X\piu Y)$ with projections $F(\pi_1^{X\piu Y})$ and $F(\pi_2^{X\piu Y})$ is a convex binary product of $F(X)$ and $F(Y)$. Let $f\colon A\to F(X)$ and $g\colon A\to F(Y)$ and $p\in (0,1)$ and take the arrow $\langle f,g \rangle_{p,1-p};[F(\iota_1), F(\iota_2)]\colon A \to F(X \piu Y)$, where $\langle f,g \rangle_{p,1-p}\colon A \to F(X)\piu F(Y)$ is the unique arrow induced bu convex binary products in $\Cat{D}$. Now, since $F(\pi_1^{X\piu Y})= k; \pi_1^{F(X\piu F(Y))}$ it follows that 
    \begin{align}
        \langle f,g \rangle_{p,1-p};[F(\iota_1), F(\iota_2)];F(\pi_1^{X\piu Y})&= \langle f,g \rangle_{p,1-p};[F(\iota_1), F(\iota_2)];k;\pi_1^{F(X\piu F(Y))}  \tag*{}\\
        &= \langle f,g \rangle_{p,1-p};\pi_1^{F(X)\piu F(Y)} \tag{$k$ inverse}\\
        &=p\cdot f \tag{convex products}
    \end{align}
    and similarly $\langle f,g \rangle_{p,1-p};[F(\iota_1), F(\iota_2)];F(\pi_2^{X\piu Y})=(1-p)\cdot g$. For the uniqueness, if $h:A \to F(X\piu Y)$ is such that $h;F(\pi_1^{X\piu Y})=p\cdot f$ and $h;F(\pi_2^{X\piu Y})=(1-p)\cdot g$ then $h;k= \langle f,g \rangle_{p,1-p}$. Post-composition with $[F(\iota_1),F(\iota_2)]$ gives $h= \langle f,g \rangle_{p,1-p};[F(\iota_1), F(\iota_2)]$.
   \end{proof}
   
   \begin{proof}[Proof of Proposition \ref{prop: monoidal functors}]\label{proof: prop monoidal functors}
By Fox's theorem we know that every finite coproduct category $\Cat{C}$ (hence any convex biproduct category) correspond to  monoidal category $(\Cat{C},\oplus_\Cat{C},\zero_\Cat{C})$ in which every object $X$ is equipped with a coherent and natural monoid structure $(X,\codiag{X},\cobang{X})$. Moreover, a functor $F\colon \Cat{C}\to \Cat{D}$ between finite coproducts  categories preserves finite coproducts if and only if seen as a monoidal funtor $F\colon(\Cat{C},\oplus_\Cat{C},\zero_\Cat{C})\to (\Cat{D},\oplus_\Cat{D},\zero_\Cat{D})$ it is strong monoidal  and preserves monoids
    \begin{equation*}
        \psi_{X,X};F(\codiag{X})= \codiag{F(X)} \qquad \psi_\zero;F(\cobang{X})=\cobang{F(X)}
    \end{equation*}
    where  $\psi\colon  \piu_\Cat{C}\circ (F\times F) \overset{\cong}{\Rightarrow} F\circ \piu_{\Cat{D}}$ and $\psi_\zero\colon \zero_\Cat{D}\overset{\cong}{\to} F(\zero_\Cat{C})$ are given by the strong monoidal structure. For any convex biproduct category,  Proposition~\ref{lemma: copca objects in convbicat} implies that every object $X$ is equipped with a natural and coherent co-pca structure $(X,\diagp{X}, \bangp{X})$. In this case, a strong monoidal functor that preserves monoids between monoidal categories, in which every object has the structure of a coherent and natural monoid and co-pca, also preserves $\,\bangp{X}$. Indeed,
     \begin{align}
        \psi_0^{-1}&= \psi_0^{-1};\,\bangp{0} \tag{\ref{ax:coh3} in Figure~\ref{fig:copcacoherence} }\\
        &=\bangp{F(0)} \tag{\ref{eq:nat copca1}}
    \end{align}
    Hence, $\,\bangp{X}$ is preserved: $F(\,\bangp{X});\psi_\zero^{-1}=F(\,\bangp{X});\bangp{F(0)}=\bangp{F(X)}$. 
    
    Moreover, observe that a convex biproduct category $\Cat{C}$ corresponds to a monoidal category 
    $(\Cat{C},\oplus_\Cat{C},\zero_\Cat{C})$ in which every object is equipped with coherent and natural monoid and co-pca structures and moreover it holds that given $f:X\to A$ and $g:X\to B$ then for every $h\colon X\to A\piu B$
    \begin{equation}\label{eq:ax aggiuntivo per fox convex biprod}
         \begin{cases}
    h;(\id{A}\piu\, \bangp{B})=\, \diagp{X};(f\piu \,\bangp{B})\\
    h;(\,\bangp{A}\piu \id{B})=\,\diagp{X};(\,\bangp{A}\piu g)
\end{cases}
\qquad
\Rightarrow 
\qquad 
h=\,\diagp{X};(f\piu g)
    \end{equation}
    since $\Cat{C}$ has convex binary products.
    Now if $F$ is PCA-enriched, then 
    \begin{align}
        F(\diagp{X});\psi_{X,X}^{-1};(\id{F(X)}\piu\, \bangp{F(X)})&=F(\diagp{X});\psi_{X,X}^{-1};(\id{F(X)}\piu\, \bangp{F(X)}); \rho_{F(X)};\rho^{-1}_{F(X)} \tag*{}\\
        &=F(\diagp{X});\psi_{X,X}^{-1};(\id{F(X)}\piu\, \bangp{F(X)}); (\id{F(X)}\piu \cobang{F(X)});\codiag{F(X)};\rho_{F(X)}^{-1} \tag{Def. $\rho$}\\
        &=F(\diagp{X});F(\id{X}\piu (\,\bangp{X};\cobang{X})); \psi_{X,X}^{-1};\codiag{F(X)};\rho^{-1}_{F(X)}\tag{Nat. $\psi$}\\
        &=F(\diagp{X});F(\id{X}\piu (\,\bangp{X};\cobang{X}));F(\codiag{X});\rho^{-1}_{F(X)} \tag{$\psi_{X,X};F(\codiag{X})= \codiag{F(X)} $}\\
        &= F(p\cdot \id{X});\rho^{-1}_{F(X)}  \tag{Eq.\ \ref{eq: assioma aggiuntivo}}\\
        &= p\cdot \id{F(X)};\rho^{-1}_{F(X)}\tag{PCA-enrichment}\\
        &=\, \diagp{F(X)};(\id{F(X)}\piu\,\bangp{F(X)});\rho_{F(X)};\rho^{-1}_{F(X)}\tag{Eq.\ \ref{eq: assioma aggiuntivo} + Def. $\rho$ + Lemma \ref{lemma:initialproperties}.1}\\
        &=\, \diagp{F(X)};(\id{F(X)}\piu\,\bangp{F(X)}) \tag*{}
    \end{align}
and similarly one can prove that $F(\diagp{X});\psi_{X,X}^{-1};(\,\bangp{F(X)}\piu\id{F(X)})= \, \diagp{F(X)};(\,\bangp{F(X)}\piu \id{F(X)})$. Hence, by property (\ref{eq:ax aggiuntivo per fox convex biprod}) it follows that $F(\diagp{X});\psi_{X,X}^{-1}=\, \diagp{F(X)}$.
Vice versa, if $F\colon(\Cat{C},\oplus_\Cat{C},\zero_\Cat{C})\to (\Cat{D},\oplus_\Cat{D},\zero_\Cat{D})$ preserves also diagonals $\diagp{}$ then it is PCA-enriched since
\begin{align}
    F(f+_pg)&= F(\diagp{X};(f\piu g);\codiag{X}) \tag{Eq.\ \ref{eq: assioma aggiuntivo}}\\
    &=\, \diagp{F(X)};\psi_{X,X};F(f\piu g; \codiag{X}) \tag*{}\\
&=\, \diagp{F(X)};F(f)\piu F(g);\codiag{F(X)}\tag*{}\\
&= F(f)+_p F(g) \tag{Eq.\ \ref{eq: assioma aggiuntivo}}
\end{align}
 Similarly, the preservation of $\,\bangp{X}$ and $\cobang{Y}$ implies that $ F(\star_{X,Y})= \star_{F(X),F(Y)}$.

\end{proof}

 \begin{corollary}\label{prop:functor 1 e mezzo}
     Let $F\colon \Cat{C}\to \Cat{D}$ a functor between convex biproduct categories. If $F$ preserves finite coproducts then
     $F$ preserves n-ary convex products.
\end{corollary}
\begin{proof}
    It follows from Proposition~\ref{prop: functor1} and Lemma~\ref{lemma:naryconvexproduct} which shows how to produce n-ary convex products from binary ones.
\end{proof}

   \begin{proposition}\label{prop: functor2}
     Let $F\colon \Cat{C}\to \Cat{D}$ a functor between convex biproduct categories that  preserves finite coproducts, then $F$ is PCA-enriched if and only if $F$ preserves $\diagp{X}$ for every $X\in\Cat{C}$.
   \end{proposition}
   \begin{proof}
    Using the notation of proposition~\ref{prop: functor1}, if $F$ is PCA-enriched, then 
    \begin{align}
        F(\langle \id{X},\id{X}\rangle_{p,1-p});k;\pi_1^{F(X)\piu F(Y)}&= F(\langle \id{X},\id{X}\rangle_{p,1-p});F(\pi_1^{X\piu Y}) \tag{Prop. \ref{prop: functor1}}\\
        &= F(p\cdot \id{X})= p\cdot \id{F(X)}\tag{PCA-enrichment}
    \end{align}
and similarly $F(\langle id,id\rangle_{p,1-p});k;\pi_1^{F(X)\piu F(Y)}= (1-p)\cdot \id{F(X)}$. Hence, the universal property of convex products in $\Cat{D}$ implies that $F(\langle \id{X},\id{X}\rangle_{p,1-p});k= \langle \id{F(X)}, \id{F(X)} \rangle_{p,1-p}$.

Vice versa, if $F(\langle \id{X},\id{X}\rangle_{p,1-p});k= \langle \id{F(X)}, \id{F(X)} \rangle_{p,1-p}$ then for every pair of arrows $f,g\in \Cat{C}([X,Y])$
\begin{align}
    F( f+_p g)=& F(\langle \id{X},\id{X} \rangle_{p,1-p};(f\piu g);\codiag{Y}) \tag{Eq.\ \ref{eq: assioma aggiuntivo}}\\
    &= F(\langle \id{X},\id{X} \rangle_{p,1-p});k; (F(f)\piu F(g)); \codiag{F(Y)} \tag{coproducts preserv.}\\
    &=\langle \id{F(X)}, \id{F(X)} \rangle_{p,1-p};(F(f)\piu F(g)); \codiag{F(Y)} \tag*{}\\
    &= F(f)+_pF(g). \tag{Eq.\ \ref{eq: assioma aggiuntivo}}
\end{align}
Moreover,
\begin{align}
    F(\star_{X,Y})&= F(\,\bangp{X};\cobang{Y}) \tag{Lemma \ref{lemma:initialproperties}.\ref{lemma:starxy}}\\
    &= F(\,\bangp{X});F(\,\cobang{Y}) \tag*{}\\
    &= \bangp{F(X)};\cobang{F(Y)} \tag{$\zero_\Cat{D}$ terminal + $F(\zero_\Cat{C})\cong \zero_\Cat{D}$   }\\
    &= \star_{F(X),F(Y)}. \tag{Lemma \ref{lemma:initialproperties}.\ref{lemma:starxy}}
\end{align}
   \end{proof}

%% file: appendices/appcmatrices.tex
\section{Appendix to Section \ref{sec:cmatrix}}\label{app:sec:cmatrix}

We break the proof of Proposition~\ref{prop:cmatrixcb} in several intermediate results.

\begin{lemma}\label{lemma:stmat coproduct}
Let $\Cat{C}$ be a PCA-enriched category. $\stmat{\Cat{C}}$ is a finite coproduct category.
\end{lemma}
\begin{proof}
We rely on Fox's theorem~\cite{fox1976coalgebras} to prove this result. Recall the definitions of $\oplus$ and $\symm$ in \eqref{eq:matsmc}. Checking that $(\stmat{\Cat{C}}, \oplus, \zero)$ is a strict symmetric monoidal category amounts to several tedious matrix multiplications. 
We proceed by illustrating that, for every object $P$, $(\codiag{P},\cobang{P})$ defined as in \eqref{eq:matmonpca}, is a natural and coherent comonoid.


Observe that, by the definitions in  \eqref{eq:matmonpca}, $\codiag{P}$ where $P=\bigoplus_{k=1}^{n} U_k$ is the matrix
\[\begin{pNiceMatrix}
1\cdot{\id{U_1}}  	& \emptyset & \Cdots & \emptyset & 1\cdot{\id{U_1}}  	& \emptyset & \Cdots & \emptyset \\
\emptyset  &   & \Ddots & \Vdots & \emptyset  &   & \Ddots & \Vdots \\	
\Vdots & \Ddots &   & \emptyset & \Vdots & \Ddots &   & \emptyset \\
\emptyset  & \Cdots & \emptyset  & 1\cdot{\id{U_{n}}} & \emptyset  & \Cdots & \emptyset  & 1\cdot{\id{U_{n}}} 
\CodeAfter
\line{1-1}{4-4}
\line{1-5}{4-8}
\end{pNiceMatrix}\]
while $\cobang{U}$ is the unique $(n+1)\times 0$ matrix.

Axioms in Table~\ref{fig:freestrictfccat} follow by matrix multiplication. For instance, the associativity axiom \ref{eq:codiag assoc} is obtained by the following computation:
\[ (\id{ P}\piu \codiag{ P}) ; \codiag{ P}=  \begin{pmatrix}
 	 \id{P} &  \id{P} &   \id{P}
 \end{pmatrix}= (\codiag{ P}\piu \id{P});\codiag{P}
 \]
and the naturality  \ref{eq:codiag nat} by:
\[ (M\piu M);\codiag{Q} 
=\begin{pmatrix}
    M & M\end{pmatrix}=
\codiag{P};M\]
where $M\colon P\to Q$. Finally, 
\ref{eq:cobang nat} is obvious since the multiplication of an $m\times n$ matrix with the $n\times 0$ empty matrix is the empty $m\times 0$ matrix. 

We can collect the above observations and conclude that $\stmat{\Cat{C}}$ is a strict symmetric monoidal category in which every object $U$ is equipped with a coherent and natural monoid structure $(U,\codiag{U},\cobang{U})$, hence, by Fox's theorem $\stmat{\Cat{C}}$ is a finite coproduct category.
\end{proof}


\begin{lemma}\label{lemma:stmat pca enriched}
    Let $\Cat{C}$ be a PCA-enriched category. $\stmat{\Cat{C}}$ is PCA-enriched.
\end{lemma}
\begin{proof}
Before proving the enrichment it is convenient to recall $(\diagp{P},\bang{P})$  from \eqref{eq:matmonpca} structure. 
For $P=\bigoplus_{k=1}^{n} U_k$, $\diagp{P}$ is the $(2n\times n)$ matrix illustrated below
\[\begin{pNiceMatrix}
p\cdot{\id{U_1}}  	& \emptyset & \Cdots & \emptyset   \\
\emptyset  &   & \Ddots & \Vdots   \\	
\Vdots & \Ddots &   & \emptyset   \\
\emptyset  & \Cdots & \emptyset  & p\cdot{\id{U_{n}}}   \\
(1-p)\cdot{\id{U_1}}  	& \emptyset & \Cdots & \emptyset \\
\emptyset  &   & \Ddots & \Vdots \\
\Vdots & \Ddots &   & \emptyset\\
\emptyset  & \Cdots & \emptyset  & (1-p)\cdot{\id{U_{n}}}
\CodeAfter
\line{1-1}{4-4}
\line{5-1}{8-4}
\end{pNiceMatrix}\]
while $\,\bangp{U}$ is the empty $0\times (n+1)$ matrix.
Axioms in Table~\ref{fig:freecopcacat} follow by matrix multiplication. For instance, in the case $U\in\Cat{C}$, the associativity axiom \ref{eq:diagp assoc} which states that  $\diagp{U};(\diagq{U}\piu \id{U})=\, \diagptilde{U};(\id{U}\piu \diagqtilde{U})$ where $\tilde{p}= pq$ and $ \tilde{q}= \frac{p(1-q)}{1-pq}$
follows, for $U\in\Cat{C}$, by the following computation

\[ \diagp{U};(\diagq{U}\piu \id{U}) = \begin{pmatrix}
    q\cdot \id{U} &  \emptyset \\
    (1-q)\cdot \id{U} &  \emptyset\\
    \emptyset & 1\cdot \id{U}
\end{pmatrix}
\begin{pmatrix}
    p\cdot \id{U} \\
    (1-p)\cdot \id{U}
\end{pmatrix}= \begin{pmatrix}
    pq\cdot \id{U} \\
    p(1-q)\cdot \id{U} \\
    (1-p)\cdot \id{U}
\end{pmatrix} = \begin{pmatrix}
    \tilde{p}\cdot \id{U} \\
    \tilde{q}(1-\tilde{p})\cdot \id{U} \\
    (1-\tilde{p})(1-\tilde{p})\cdot \id{U}
    \end{pmatrix} \]
which is equal to
\[ \diagptilde{U};(\id{U}\piu\, \diagqtilde{U}) = \begin{pmatrix}
    1\cdot \id{U} &  \emptyset\\
    \emptyset & \tilde{q}\cdot \id{U} \\
    \emptyset & (1-\tilde{q})\cdot \id{U}
\end{pmatrix}
\begin{pmatrix}
    \tilde{p}\cdot \id{U} \\
    (1-\tilde{p})\cdot \id{U}
\end{pmatrix}= \begin{pmatrix}
    \tilde{p}\cdot \id{U} \\
    \tilde{q}(1-\tilde{p})\cdot \id{U} \\
    (1-\tilde{p})(1-\tilde{p})\cdot \id{U}
\end{pmatrix} \]
The general case $P=\bigoplus_{k=1}^{n} U_k$ is obtained similarly. Axiom~\ref{eq:diagp idempotency} is obvious since $p\cdot \id{U_i}+ (1-p)\cdot \id{U_i}= \id{U_i}$ by PCA axiom~\ref{eq:pca}. 
Axiom~\ref{eq:diagp nat} is obtained as follows: for $M\colon U\to V$

\[M;\diagp{V} = \begin{pmatrix}
    p\cdot M\\
    (1-p)\cdot M
\end{pmatrix} 
M=\, \diagp{U};M\piu M\]
where $p\cdot M$ denotes the matrix whose $(j,i)$-entries are given by $p\cdot M_{ji}$. 
Finally, \ref{eq:bangp nat} is obvious since the multiplication of an $m\times n$ matrix with the $n\times 0$ empty matrix is the empty $m\times 0$ matrix. Hence, we can collect the above observations and conclude that every object $U\in\stmat{\Cat{C}}$ is equipped with a co-pca structure $(U,\diagp{U},\bangp{U})$ which satisfies the axioms in Table~\ref{fig:freecopcacat}.

Enrichment: the monoid and co-pca structure on every object of $\stmat{\Cat{C}}$ induces a PCA-enrichment as follows. For every pair of objects $P$ and $Q$ and arrows $M,M'\colon P\to Q$ define 
\[M+_p M' =\, \diagp{P};(M\piu M');\codiag{Q}\] 
matrix multiplication implies that
\[M+_p M'_{ji} = (p\cdot M_{ji} + (1-p)\cdot M'_{ji})= M_{ji}+_p M'_{ji}\]
While $\star_{P,Q}$ is given by $\,\bangp{P};\cobang{Q}$ which is the $(m\times n)$ matrix whose $(j,i)$-entry is given by $\star_{A_i,B_j}$.

PCA-enrichment equations in (\ref{eq:enr}) are obtained through axioms of coherent and natural monoid and co-pca structures, for instance, the first equation in (\ref{eq:enr}) $N;(M+_pM')= (N;M +_p N;M')$, for $N\colon P'\to P$, is obtained as follows:
\begin{align}
    N;(M+_pM') &= N;(\diagp{P};(M\piu M');\codiag{Q}) \tag*{}\\
    &=\, \diagp{P'};(N\piu N);(M\piu M');\codiag{Q} \tag{axiom \ref{eq:diagp nat}}\\
    &=\, \diagp{P'};(N;M\piu N;M');\codiag{Q} \tag{monoidal structure}\\
    &= (N;M) +_p (N;M') \tag*{}
\end{align}
\end{proof}

\begin{lemma}\label{lemma:stmat convex prod}
    Let $\Cat{C}$ be a PCA-enriched category. $\stmat{\Cat{C}}$ has convex binary products.
\end{lemma}
\begin{proof}
    Let $M_1\colon P\to Q_1$ and $M_2\colon P\to Q_2$ be two morphisms in $\stmat{\Cat{C}}$ and $p,q\in[0,1]$ such that $p+q\le 1$, where $P=\bigoplus_{k=1}^n U_k$, $Q_1=\bigoplus_{k=1}^{m_1} V_k$ and $Q_2=\bigoplus_{k=1}^{m_2} V'_k$. Now consider $Q_1\piu Q_2$ with 
    $\pi_1\colon Q_1\piu Q_2 \to Q_1$ and $\pi_2\colon Q_1\piu Q_2 \to Q_2$ given by the $(m_1\times (m_1+m_2))$ and $(m_2\times (m_1+m_2))$ matrices
    \[\pi_1=\begin{pNiceMatrix}
1\cdot{\id{V_1}}  	& \emptyset & \Cdots & \emptyset & \emptyset  	& \emptyset & \Cdots & \emptyset \\
\emptyset  &   & \Ddots & \Vdots & \emptyset  &   &  & \Vdots \\	
\Vdots & \Ddots &   & \emptyset & \Vdots &  &   & \emptyset \\
\emptyset  & \Cdots & \emptyset  & 1\cdot{\id{V_{m_1}}}  & \emptyset  & \Cdots & \emptyset  & \emptyset
\CodeAfter
\line{1-1}{4-4}
\line{1-5}{4-8}
\tikz \draw[thin,dotted] (2-5.east) -- (4-7.north west);
\tikz \draw[thin,dotted] (1-6.east) -- (3-8.north west);
\end{pNiceMatrix}\]
\[ \pi_2=\begin{pNiceMatrix}
\emptyset  	& \emptyset & \Cdots & \emptyset & 1\cdot{\id{V'_1}}  	& \emptyset & \Cdots & \emptyset \\
\emptyset  &   & \Ddots & \Vdots & \emptyset  &   &   & \Vdots \\	
\Vdots & \Ddots &   & \emptyset & \Vdots &  &   & \emptyset \\
\emptyset  & \Cdots & \emptyset  & \emptyset  & \emptyset  & \Cdots & \emptyset  & 1\cdot{\id{V'_{m_2}}}
\CodeAfter
\tikz \draw[thin,dotted] (2-5.south east) -- (4-7.north west);
\tikz \draw[thin,dotted] (1-6.south east) -- (3-8.north west);
\line{1-1}{4-4}
\line{1-5}{4-8}
\end{pNiceMatrix}\]
which correspond to $(\id{Q_1}\piu\, \bangp{Q_2})$ and $(\,\bangp{Q_1}\piu \id{Q_2})$ respectively. Now, the arrow $\diagvecp{P}$ with $\vec{p}=p,q$ given by the $(2n\times n)$ matrix
\[\begin{pNiceMatrix}
p\cdot{\id{U_1}}  	& \emptyset & \Cdots & \emptyset   \\
\emptyset  &   & \Ddots & \Vdots   \\	
\Vdots & \Ddots &   & \emptyset   \\
\emptyset  & \Cdots & \emptyset  & p\cdot{\id{U_{n}}}   \\
q\cdot{\id{U_1}}  	& \emptyset & \Cdots & \emptyset \\
\emptyset  &   & \Ddots & \Vdots \\
\Vdots & \Ddots &   & \emptyset\\
\emptyset  & \Cdots & \emptyset  & q\cdot{\id{U_{n}}}
\CodeAfter
\line{1-1}{4-4}
\line{5-1}{8-4}
\end{pNiceMatrix}\]
matrix multiplication shows that 
\[\diagvecp{P};(M_1\piu M_2);\pi_1= p\cdot M_1 \qquad\diagvecp{P};(M_1\piu M_2);\pi_2 = q\cdot M_2\]
 Now, if $H\colon P\to Q_1\piu Q_2$ is a morphism in $\stmat{\Cat{C}}$ such that $H;\pi_1 = p\cdot M_1$ and $H;\pi_2 = (1-p)\cdot M_2$, then
\[H_{ji}= p\cdot M_{ji}\]
for every $j=1,\dots,m_1$ and $i=1,\dots,n$ and
\[H_{ji}= q\cdot M'_{(j-m_1)i}\]
for every $j=m_1+1,\dots,m_1+m_2$ and $i=1,\dots,n$. Hence, $H\equiv\, \diagvecp{P};(M_1\piu M_2)$.
Observe that if $p=1$ (and hence $q=0$) then $\diagvecp{P}$ is $\iota_1\defeq (\id{Q_1}\piu \,\cobang{Q_2})$ and if $p=0$ (and hence $q=1$) then $\diagvecp{P}$ is $\iota_2\defeq (\,\cobang{Q_1}\piu \id{Q_2})$. 
\end{proof}

\begin{proof}[Proof of Proposition \ref{prop:stmatfun}]
    The functor $\stmat{F}$ is strict monoidal and preserves monoids and co-pca objects, hence by Proposition~\ref{prop: monoidal functors} it is a morphism of convex biproduct categories.
\end{proof}

\begin{proof}[Proof of Theorem \ref{thm:matfree}]
    Let $\Cat{C}$ be a PCA-enriched category and $\Cat{D}$ a convex biproduct category, and consider a PCA-enriched functor $F: \Cat{C}\to U(\Cat{D})$. The unit $\eta:\Cat{C}\to \stmat{\Cat{C}}$ is the identity on objects functor mapping an arrow $f\in\Cat{C}[U,V]$ into the $(1\times 1)$ matrix with entry $1\cdot f$. 
    
   The functor $F^\sharp:\stmat{\Cat{C}}\to \Cat{D}$ sends an object $\bigoplus_{k=1}^n U_k$ to $\bigoplus_{k=1}^n F(U_k)$ and an arrow $M\colon\bigoplus_{k=1}^n U_k\to \bigoplus_{k=1}^m V_k$, which corresponds to a matrix with entries $M_{ji}= p_{ji}\cdot f_{ji}$, into the arrow of $\Cat{D}$ obtained as follows:
     we first define the image of a column of the matrix, then we observe that $M$ is the copairing of its columns, denoted with $\column{M}{i}$ for $i=1,\dots,n$, and define $F^\sharp(M)$ as the copairing of the images of the columns. The interpretation of the $i$-column $\column{M}{i}$ is defined as the arrow induced by the $m$-ary convex product through the vector $\vec{p}=(p_{1i},\dots,p_{mi})$ and the arrows $F(f_{ji})$ in $\Cat{D}$, which exists by Lemma~\ref{lemma:naryconvexproduct}. $F^\sharp$ is well defined since, for $M\equiv M'$, it holds that $M_{ji}=p_{ji}\cdot f_{ji}= p'_{ji}\cdot f'_{ji}= M'_{ji}$. Now, for $k=1,\dots m$
    \begin{align} 
        F^\sharp(\column{M'}{i});\pi_k&=p'_{ki}\cdot F(f'_{ki})\tag{Def. $F^\sharp(\column{M'}{i})$ }\\
        &=F(p'_{ki}\cdot f'_{ki}) \tag{$F$ PCA-enriched}\\
        &= F(p_{ki}\cdot f_{ki}) \tag{$M\cong M'$}\\
        &=p_{ki}\cdot F(f_{ki}) \tag{$F$ PCA-enriched}
        \end{align}   
        and the universal property of convex products in $\Cat{D}$ implies that $F^\sharp(\column{M'}{i})= F^\sharp(\column{M}{i})$ and, hence, $F^\sharp(M')=F^\sharp(M)$ by the universal property of coproducts in $\Cat{D}$. Moreover, a simple computation shows that for every object $U\in\stmat{\Cat{C}}$, $F^\sharp$ preserves monoids and co-pca structures, i.e.\ it holds: 
    
    \begin{itemize}
    \item $F^\sharp(\,\diagp{U})=\, \diagp{F(U)}$ and $F^\sharp(\,\bangp{U})=\, \bangp{F(U)}$;
    \item $F^\sharp(\,\codiag{U})= \codiag{F(U)}$ and $F^\sharp(\,\cobang{U})= \cobang{F(U)}$;
    \item $F^\sharp(1\cdot f)= F(f)$, for every $f\colon U\to V$ in $\Cat{C}$;
    \end{itemize}
     Hence, by Proposition~\ref{prop: monoidal functors} it follows that $F^\sharp$ is a morphism of convex biproduct categories. It is also the unique morphism such that $\eta;F^\sharp = F$ since any other morphism $G:\stmat{\Cat{C}}\to \Cat{D}$ such that $\eta;G=F$ must preserve coproducts and convex products by Proposition~\ref{prop: functor1} and \ref{prop:functor 1 e mezzo}. Hence, $F^\sharp= G$ since every column of $M$ is the arrow induced by the m-ary convex product through $\vec{p}=(p_{1i},\dots,p_{mi})$ and the arrows $f_{ji}$ for a fixed $i$. 
    \end{proof}

\begin{proof}[Proof of Proposition \ref{prop:counit fullfaithful}]\label{proof:counit fullfaithful}
    Recall that the counit $\epsilon$ is obtained as $\id{\Cat{C}}^\sharp$, following the notation in the proof of Theorem \ref{thm:matfree}. Fullness follows from the fact that every arrow $f:A\to B$ in $\Cat{C}$ if equal to $\epsilon(\eta(f))$ by construction. Faithfulness is obtained by observing that if $\epsilon(M)=\epsilon(M')$ for two matrices in $\stmat{U(\Cat{C})}$, $M,M'\colon\bigoplus_{k=1}^n U_k\to \bigoplus_{k=1}^m V_k$ with entries $M_{ji}= p_{ji}\cdot f_{ji}$ and $M'_{ji}= p'_{ji}\cdot f'_{ji}$ , then the copairing of the image through $\epsilon$ of the columns of $M$ is equal to the copairing of the image through $\epsilon$ of the columns of $M'$. The universal property of the coproducts in $\Cat{C}$ then ensures that for every $i$-column of $M$ and $M'$, denoted with $\column{M}{i}$ and $\column{M'}{i}$, for $i=1,\dots,n$, it holds $\epsilon(\column{M}{i})=\epsilon(\column{M'}{i})$. Now, since $\epsilon(\column{M}{i})$ is given by the arrow induced by the m-ary convex product through $f_{1i},\dots,f_{mi}$, with $\vec{p}=(p_{1i},\dots,p_{mi})$, and $\epsilon(\column{M'}{i})$ is given by the arrow induced by the m-ary convex product through $f'_{1i},\dots,f'_{mi}$, with $\vec{p}'=(p'_{1i},\dots,p'_{mi})$, it follows that $\epsilon(\column{M}{i});\pi_j=p_{ji}\cdot f_{ji}=p'_{ji}\cdot f'_{ji} =\epsilon(\column{M'}{i})\pi_{j}$, for $j=1,\dots,m$. Hence, $M_{ji}\equiv M'_{ji}$ and therefore they are equal in $\stmat{U(\Cat{C})}$.  
    \end{proof} 

%% file: appendices/apptc.tex
\section{Appendix to Section \ref{sec:syntactic}}\label{app:sec:syntactic}

\begin{proof}[Proof of Theorem \ref{thm:TCconvexbiproductcategory}]
The fact that the definitions of $+_p$ and $\star$ in \eqref{eq:enrichmentTC} provides a  $\Cat{PCA}$ follows immediately by the axioms \ref{eq:diagp assoc}, \ref{eq:diagp idempotency} and \ref{eq:diagp symmetry}. The fact that this provides an enrichment of $\CatTapeC$ over $\Cat{PCA}$ is easily proved by using naturality of $\diagp{}$, $\codiag{}$, $\bangp{}$ and $\cobang{}$ see, e.g., \cite{bonchi2025tapediagramsmonoidalmonads}.

The fact that $\zero$ is both initial and final follows immediately by naturality of  $\cobang{}$ and $\bang{}$: axioms \eqref{eq:bangp nat} and \eqref{eq:cobang nat}.

Since, by the axioms in Figure \ref{fig:freestrictfccat}, every object $P$ of $\CatTapeC$ is equipped with a natural and coherent monoid, by \cite{fox1976coalgebras}, it holds that $\oplus$ is a coproducts with injections
\begin{equation*}
\iota_1 \defeq \id{P_1}\oplus \cobang{P_2} \colon P_1 \to P_1\oplus P_2\quad\text{and}\quad\iota_2 \defeq \cobang{P_1}\oplus \id{P_2} \colon P_2 \to P_1\oplus P_2\text{.}
\end{equation*}
Recall that (see for instance \cite[Ch. 6.4]{mellies2009categorical} ) the copairing of $\t\colon P_1 \to Q$ and $\s\colon P_2 \to Q$ is defined as
\begin{equation*}
[\t,\s] \defeq (f\oplus g); \codiag{Q} 
\end{equation*}
and thus $[\id{Q},\id{Q}] = \codiag{Q}$.

Our main technical effort, illustrated below in Proposition \ref{prop:TChasconvexcoprduct}, consists in proving that $\oplus$ is also a convex product with projections
\begin{equation*}
\pi_1 \defeq \id{P_1}\oplus \bang{P_2} \colon P_1 \oplus P_2 \to P_1\quad\text{and}\quad\pi_2 \defeq \bang{P_1}\oplus \id{P_2} \colon P_1\oplus P_2 \to  P_2\text{.}
\end{equation*}
and the convex pairing of $\t\colon P \to Q_1$ and $\s\colon P \to Q_2$ is given as
\begin{equation*}
<\t,\s>_{1,0}\defeq  \t\oplus \cobang{Q_2} \qquad <\t,\s>_{0,1} \defeq \cobang{Q_1} \oplus \s \qquad <\t,\s>_{p,1-p} \defeq \; \diagp{P}; (\t\oplus \s) \text{ for }p\in(0,1)
\end{equation*}
As we will illustrate later, the case of  $<\t,\s>_{p,q}$ with $p+q<1$ simply follows from the above one.

Before delving in proof of Proposition~\ref{prop:TChasconvexcoprduct} observe that \eqref{eq:delta} in Definition~\ref{def:convbicat} follows easily by the above definitions of $\iota_i$ and $\pi_1$, axioms~\eqref{eq:bangp nat} and \eqref{eq:cobang nat}, for the case $i=j$, and the definition of $\star$ in \eqref{eq:enrichmentTC} for the case $i\neq j$.
Instead the laws \eqref{eq: assioma aggiuntivo} follows from the characterizations of $<f,g>_{p,1-p}$ and $\codiag{}$ and the definition of $+_p$ in \eqref{eq:enrichmentTC}.

Finally, to prove that $\langle \t_1, \dots \t_n\rangle_{\vec{p}} = \diagpn{\;\;\vec{p}}{U}{n}; \bigoplus_{i=1}^n \t_i$, one can easily check that, for all $i$, $(\diagpn{\;\;\vec{p}}{U}{n}; \bigoplus_{i=1}^n \t_i ); \pi_i =p_i\cdot \t_i$, where the $i$-th projection $\pi_i\colon \bigoplus_{j=1}^nU_j \to U_i$ is defined as $(\bang{\bigoplus_{j=1}^{i-1}U_j}) \oplus \id{U_i} \oplus (\bang{\bigoplus_{j=i+1}^{n}U_j})$.
\end{proof}

The proof of Proposition~\ref{prop:TChasconvexcoprduct} substantially relies on a normal form result (Lemma~\ref{lemma:division}) and on the fact that the enrichment on $\CatTapeC$ is cancellative, namely if $p \cdot \t = p \cdot \s$, then $\t=\s$ (Lemma~\ref{cancellativity}). The remainder of this appendix is structured as follows: we first illustrate some preliminary results, then some results about normal forms, then the cancellativity property and finally convex products.

\subsection{Preliminary Results}
We collect below several simple results that we will frequently use.

\begin{lemma}\label{lemma:pcdot}
Let $\t\colon P \to Q$ be an arrow of $\CatTapeC$ . Then, for all $p\in(0,1)$, $p\cdot \t =\, \diagp{P};(\t \oplus \bang{P})$.
\end{lemma}
\begin{proof}
\begin{align*}
p\cdot \t &= p+_p \star_{P,Q} \tag{def. of $p\cdot -$}\\
&=\, \diagp{P} ; (\t \oplus \star_{P,Q}) ; \codiag{Q} \tag{\ref{eq:enrichmentTC}}\\
&=\, \diagp{P} ; (\t \oplus (\bang{P}; \cobang{Q})); \codiag{Q} \tag{\ref{eq:enrichmentTC}}\\
&=\, \diagp{P} ; (\t \oplus  \bang{P}) ;(\id{Q} \oplus \cobang{Q}) ; \codiag{Q} \tag{Symmetric Monoidal Category}\\
&=\, \diagp{P} ; (\t \oplus  \bang{P}) ;\id{Q} \tag{\ref{eq:codiag unital}} \\
&=\, \diagp{P} ; (\t \oplus  \bang{P})  \tag{Category} 
\end{align*}
\end{proof}

\begin{lemma}\label{lemma:pcopairing}
Let $\s = (\s_1 \oplus \s_2) ; \codiag{}$ be an arrow of $\CatTapeC$ . 
Then $r \cdot \s = (r\cdot \s_1 \oplus r\cdot s_2); \codiag{}$. 
\end{lemma}
\begin{proof}
Assume that $\s$ has type $P_1\oplus P_2 \to Q$ and observe that
\begin{align*}
r \cdot \s &=\, \diagpX{r}{P_1\oplus P_2} ; (\s \oplus \bang{}) \tag{Lemma \ref{lemma:pcdot}}\\
&= (\diagpX{r}{P_1} \oplus\, \diagpX{r}{P_2} ) ; (\id{P_1}\oplus \symm{P_1}{P_2} \oplus \id{P_2}) ; (\s \oplus \bang{P_1}\oplus \bang{P_2}) \tag{\ref{eq:diagp coherence}}\\
&= (\diagpX{r}{P_1} \oplus\, \diagpX{r}{ P_2} ) ; (\id{P_1}\oplus \symm{P_1}{P_2} \oplus \id{P_2}) ; (\, ((\s_1 \oplus \s_2) ; \codiag{Q}) \oplus \bang{P_1}\oplus \bang{P_2}\,) \tag{Hypothesis}\\
&= (\diagpX{r}{P_1} \oplus\, \diagpX{r}{P_2} ) ; (\id{P_1}\oplus \symm{P_1}{P_2} \oplus \id{P_2}) ; ( \s_1 \oplus \s_2  \oplus \bang{P_1}\oplus \bang{P_2}); \codiag{Q} \tag{Sym. Mon. Cat.}\\
&= (\diagpX{r}{P_1} \oplus\, \diagpX{r}{ P_2} ) ;( \s_1 \oplus \bang{P_1} \oplus \s_2   \oplus \bang{P_2}); \codiag{Q} \tag{Symmetric Monoidal Category}\\
&= (\diagpX{r}{P_1} ; ( \s_1 \oplus \bang{P_1})) \oplus (\diagpX{r}{ P_2}  ;( \s_2   \oplus \bang{P_2})); \codiag{Q} \tag{Symmetric Monoidal Category}\\
&= (r\cdot \s_1 ) \oplus (r\cdot  \s_2 ); \codiag{Q} \tag{Lemma \ref{lemma:pcdot}}
\end{align*}
\end{proof}

\begin{lemma}\label{lemma:basicTC}
Let $\t=\,\diagp{} ; (\t_1 \oplus \t_2)$ be an arrow of $\CatTapeC$. Then:
\begin{enumerate}
\item $\t; (\id{} \oplus \bang{}) = p\cdot \t_1$;
\item $\t ; (\bang{} \oplus \id{}) = (1-p) \cdot \t_2$;
\item $r \cdot \t =\, \diagp{} ; (r\cdot \t_1 \oplus r\cdot \t_2)$ for all $r\in (0,1)$.
\end{enumerate}
\end{lemma}
\begin{proof}
Points 1 and 2 follows immediately from the definition of the enrichment.

We prove below point 3 for $\t_1\colon:P\to Q_1$ and $\t_2\colon P\to Q_2$.
\begin{align*}
r \cdot \t &=\, \diagpX{r}{} ; (\t \oplus \bang{Q_1\piu Q_2}) \tag{Lemma \ref{lemma:pcdot}}\\
&=\, \diagpX{r}{} ; (\diagp{} ; (\t_1 \oplus \t_2) \oplus \bang{{Q_1\piu Q_2}}) \tag{Hypthesis}\\
&=\, \diagpX{r}{} ; ( \, (\diagp{} ; (\t_1 \oplus \t_2)) \oplus (\diagp{}; (\bang{Q_1} \oplus \bang{Q_2}))\, ) \tag{\ref{eq:bangp coherence} + \ref{eq:diagp idempotency}}\\
&=\, \diagpX{r}{} ; (\diagp{} \oplus \diagp{}) ; (\,  (\t_1 \oplus \t_2) \oplus (\bang{Q_1} \oplus \bang{Q_2}) \, ) \tag{Symmetric Monoidal Categories}\\
&=\, \diagp{} ; (\diagpX{r}{} \oplus \diagpX{r}{}) ; (\id{} \oplus \symm \oplus \id{}) ; (\,  (\t_1 \oplus \t_2) \oplus (\bang{Q_1} \oplus \bang{Q_2}) \, ) \tag{\ref{eq:diagp nat} + Sym. Mon. Cat.}\\
&=\, \diagp{} ; (\diagpX{r}{} \oplus \diagpX{r}{}) ;  (\,  (\t_1 \oplus \bang{Q_1}) \oplus (\t_2 \oplus \bang{Q_2}) \, ) \tag{Symmetric Monoidal Categories}\\
&=\, \diagp{} ; ( \,(\diagpX{r}{} ; (\t_1 \oplus \bang{Q_1})) \oplus (\diagpX{r}{} ;(\t_2 \oplus \bang{Q_2})) \, ) \tag{Symmetric Monoidal Categories}\\
&=\, \diagp{} ; (r\cdot \t_1 \oplus r\cdot \t_2)  \tag{Lemma \ref{lemma:pcdot}}
\end{align*}
 
\end{proof}

\begin{lemma}\label{lemma:fractions}
Let $\t=\,\diagp{} ; (\t_1 \oplus \t_2)$ and $\s=\, \diagq{} ; (\s_1 \oplus \s_2)$. If $\t_1=\frac{q}{p}\cdot \s_1$ and $\frac{1-p}{1-q} \cdot \t_2 = \s_2$, then $\t = \s$.
\end{lemma}
\begin{proof}
\begin{align}
\t &=\,  \diagp{};( \frac{q}{p} \cdot \s_1 \oplus \t_2  )\tag{Lemma \ref{lemma:basicTC}.1}\\
&=\, \diagp{}; (\diagpX{\frac{q}{p}\,\,}{} \oplus \id{}) ; (\s_1 \oplus \bang{} \oplus \t_2) \tag{Lemma \ref{lemma:pcdot}}\\
&=\, \diagq; (\id{} \oplus\, \diagpX{\frac{p-q}{1-q}\,\,}{}) ; (\s_1 \oplus \bang{} \oplus \t_2) \tag{\ref{eq:diagp assoc}} \\
&=\, \diagq; (\id{} \oplus ( \diagpX{\frac{1-p}{1-q}\,\,}{}; \symm)) ; (\s_1 \oplus (\bang{} \oplus \t_2)) \tag{\ref{eq:diagp symmetry}} \\
&=\, \diagq; (\id{} \oplus ( \diagpX{\frac{1-p}{1-q}\,\,}{})) ; (\s_1 \oplus  \t_2 \oplus \bang{}) \tag{Symmetric Monoidal Category}\\
&=\,  \diagq{};( \s_1 \oplus \frac{1-p}{1-q} \cdot \t_2  ) \tag{Lemma \ref{lemma:basicTC}.2}\\
&=s \tag{Hypothesis}
\end{align}
\end{proof}

\subsection{Normal forms}
We are now going to prove several normal form results, that will be crucial in the proof of Proposition~\ref{prop:TChasconvexcoprduct}. 
For these results we need to fix some notation.



For all natural numbers $n \in \mathbb{N}$ and $p_1, \dots , p_n \in (0,1)$ such that $\sum_{i=1}^n p_i \leq 1$ one can define $\diaggen{\vec{p}}{P}\colon P \to \bigoplus_{i=1}^{n+1}  P$ as follows:\\
\begin{equation}\label{eq:diag vettore in appendice}
    (n=0):\  \diaggen{\vec{p}}{P}\defeq\id{U}\qquad
(n=k+1):\ \diaggen{\vec{p}}{P}\defeq\, \diaggen{p_1}{P}; (\id{P} \oplus\, \diaggen{\frac{\vec{p'}}{1-p_1}}{P})
\end{equation}
where $\vec{p'} = p_2, \dots, p_n$.
\newline
Observe that $\diaggen{\vec{p}}{P}$ differs from the definition of $\diagpn{\vec{p}\;}{P}{n}$ given in \eqref{eq: diagpn cap equational presentation}, but we need this formulation for the rest of the section.

Similarly, for all natural numbers $n \in \mathbb{N}$ one can define $\codiag{U}^n\coloneqq \bigoplus_{i=1}^n P\to P$ as follows:
\begin{equation}\label{eq:codiagn in appendice}
    (n=0):\  \codiag{P}^0\defeq\cobang{P}\qquad (n+1):\ \codiag{P}^{n+1}\defeq(\id{P} \oplus \codiag{P}^n  );\codiag{P}.
\end{equation}

\begin{lemma}[Pre normal form]\label{lemma:prenormalform}
Every arrow $\t\colon \oplus_{i=1}^n U_i\to \oplus_{i=1}^m V_i$ in $\CatTapeC$ can be written as $\t_1;\t_2;\t_3;\t_4$ where:
\begin{itemize}
    \item $\t_1$ is obtained through the fragment \setlength{\tabcolsep}{4pt}\begin{tabular}{rc ccccccccccccccccccccc}
   $\!\!\! \mid \!\!\!$ & $\diagp{U}$  & $\!\!\! \mid \!\!\!$  & $\id{\zero}$ & $\!\!\! \mid \!\!\!$ &  $\id{U}$ & $\!\!\! \mid \!\!\!$   & $   \t ; \t   $ & $\!\!\! \mid \!\!\!$  & $  \t \piu \t  $ & $\!\!\! \mid \!\!\!$  &  $\sigma_{U,V}^{\piu}$
\end{tabular}
\item $\t_2$ is obtained through the fragment \setlength{\tabcolsep}{4pt}\begin{tabular}{rc ccccccccccccccccccccc}
           $\!\!\! \mid \!\!\!$  &  $\bangp{U}$& $\!\!\! \mid \!\!\!$&  $\id{\zero}$ & $\!\!\! \mid \!\!\!$ & $\id{U}$    & $\!\!\! \mid \!\!\!$   & $   \t ; \t   $ & $\!\!\! \mid \!\!\!$  & $  \t \piu \t  $ & $\!\!\! \mid \!\!\!$  &  $\sigma_{U,V}^{\piu}$
    \end{tabular}
    \item $\t_3$ is obtained through the fragment \begin{tabular}{rc ccccccccccccccccccccc}\setlength{\tabcolsep}{0.0pt}
        $\!\!\! \mid \!\!\!$ & $\tapeFunct{c}$  & $\!\!\! \mid \!\!\!$  &  $\id{\zero}$ & $\!\!\! \mid \!\!\!$ & $\id{U}$&  $\!\!\! \mid \!\!\!$   & $   \t ; \t   $ & $\!\!\! \mid \!\!\!$  & $  \t \piu \t  $ & $\!\!\! \mid \!\!\!$  &  $\sigma_{U,V}^{\piu}$
    \end{tabular}
    \item $\t_4$ is obtained through the fragment \begin{tabular}{rc ccccccccccccccccccccc}
        $\!\!\! \mid \!\!\!$ & $\codiag{U}$ &$\!\!\! \mid \!\!\!$ & $\cobang{U}$ & $\!\!\! \mid \!\!\!$  &  $\id{\zero}$ & $\!\!\! \mid \!\!\!$ & $\id{U}$ & $\!\!\! \mid \!\!\!$   & $   \t ; \t   $ & $\!\!\! \mid \!\!\!$  & $  \t \piu \t  $ & $\!\!\! \mid \!\!\!$  &  $\sigma_{U,V}^{\piu}$
    \end{tabular}
\end{itemize}
\end{lemma}
\begin{proof}
Using the naturality, one can first move all the occurrences of $\diagp{}$ to the left. The all the occurrences of $\cobang{}$ to the left until an occurence of $\diagp{}$ is met (note that since $\cobang{}$ is not the unit of $\diagp{}$, $\cobang{}$ remains at the right of $\diagp{}$). Similarly, one can move all occurrences of $\codiag{}$ and $\cobang{}$ on the right.
\end{proof}

\begin{lemma}\label{lemmaABform}
Let $\t\colon U \to V$ be an arrow of $\CatTapeC$. Then one of the following holds: 
\begin{itemize}
\item $\t=\star_{A,B}$;
\item there exist  $c \in \Cat{C}[U,V]$ such that $\t=\tapeFunct{c}$;
\item there exist distinct $c_1, \dots, c_n \in \Cat{C}[U,V]$ and $\vec{p}=p_1, \dots, p_{n+1}$ such that
\[\t=\, \diagvecp{A} ; ( (\bigoplus_{i=1}^n  \tapeFunct{c_i}) \oplus\, \bangp{U}  ) ; \codiag{V}^n\text{.}\]
\end{itemize}
\end{lemma}
\begin{proof}
By Lemma~\ref{lemma:prenormalform}, $\t= \t_1;\t_2;\t_3;\t_4$.
Using the naturality of $\sigma_{U,V}^{\piu}$ (see Table~\ref{fig:freestricmmoncatax}), one can move all the occurrences of $\sigma_{U,V}^{\piu}$ to the right until an occurrence of $\codiag{V}$ is met. 
Using the associativity and the symmetry of $\diagp{U}$  all multiple occurrences of every $c_i$ and $\star_{U,V}$ can be collected in the form $\diagp{U};(c_i\otimes c_i)$ and then collapsed to one occurrence of $c_i$ by using the idempotency axiom and axioms for $\codiag{U}$. Then again using the associativity and the symmetry of $\diagp{U}$ one can move all the occurrences of in the form of $\diagvecp{U}$ for some vector $\vec{p}=(p_1,\dots,p_n)$, and similarly all occurrences of $\codiag{V}$ in the form of $\codiag{V}^n$.
\end{proof}

\begin{lemma}\label{lemma: arrows A-B subdistributions}
    Arrows in $\CatTapeC$ of the form $\t:U\to V$ are in bijection with $\mathcal{D}_{\le}(\Cat{C}[U,V])$.
\end{lemma}
\begin{proof}
    By Lemma~\ref{lemmaABform} every arrow $f:U\to V$ is in one of the following forms:
\begin{itemize}
\item $\t=\star_{U,V}$;
\item there exists $c \in \Cat{C}[U,V]$ such that $\t=\tapeFunct{c}$;
\item there exist distinct $c_1, \dots, c_n \in \Cat{C}[U,V]$ and $\vec{p}=p_1, \dots, p_{n}$ such that 
\[\t=\,\diagvecp{U} ; ( (\bigoplus_{i=1}^n  \tapeFunct{c_i}) \oplus\, \bangp{U}  ) ; \codiag{V}^n\text{.}\]
\end{itemize}
The first two cases correspond respectively to the zero subdistribution and the Dirac distribution on $c$. The third case corresponds to the subdistribution $\vec{p}$ on the set of arrows $\{c_1, \dots, c_n\}$, where $p_i$ is the probability of $c_i$. 
 If $\t$ and $\t'=\, \diaggen{\vec{q}}{U} ; ( (\bigoplus_{j=1}^{n'}  \tapeFunct{c'_i}) \oplus\, \bangp{U}  ) ; \codiag{V}^{n'}$ are sent to the same subdistribution, then one has that $n=n'$, $p_i=q_i$ and $c_i=c'_i$ for all $i=1,\dots,n$. 
\end{proof}

\begin{lemma}\label{lemma:division}
Let $\t \colon U \to \bigoplus_{i=1}^n V_i$ be an arrow of $\CatTapeC$. Then, 
\begin{itemize}
\item either $\t= \bang{U}; \cobang{\bigoplus_{i=1}^n V_i}$,
\item or there exists $j\in 1\dots n$ and $\t_j \colon U \to V_j$ such that $\t = \cobang{\bigoplus_{i=1}^{j-1} V_i} \oplus \t_j \oplus \cobang{\bigoplus_{i=j+1}^{n} V_i}$,
\item or for all $j\in 1\dots n$, there exists $\s_1\colon U \to \bigoplus_{i=1}^{j} V_i$ and $\s_2 \colon U \to \bigoplus_{i=j+1}^{n} V_i$ such that $\t=\,\diagp; (\s_1\oplus \s_2)$.
\end{itemize}
\end{lemma}
\begin{proof}
We have that $\t= \t_1;\t_2;\t_3;\t_4$ where the $\t_i$ are prescribed as by Lemma \ref{lemma:prenormalform}.
We consider two different cases: either $\t_1=\id{U}$ or $\t_1 \neq \id{U}$.

Assume that $\t_1=\id{U}\colon U \to U$. Then $\t_2$ is an arrow with domain $U$ and thus either $\t_2 = \bang{U}$ or $\t_2= \id{U}$. 

\begin{itemize}
\item If $\t_2 = \bang{U} \colon U \to \zero$, then $\t_3$ is an arrow with domain $\zero$. This entails that $\t_3$ should be $\id{\zero}\colon \zero \to \zero$ and $\t_4$ and arrow of type $\zero \to \bigoplus_{i=1}^n V_i$. The latter fact entails that $\t_4 = \cobang{\bigoplus_{i=1}^n V_i}$. Thus $\t= \bang{U}; \cobang{\bigoplus_{i=1}^n V_i}$.
\item If $\t_2 = \id{U} \colon U \to U$, then $\t_3$ should have type $U\to V_j$ and thus it must be either $\id{U}$ or $\tape{c}$. We proceed with the case for $\tape{c}$, the one for $\id{U}$ is identical. Since $\tape{c}\colon U\to V_j$ then $\t_4$ should have type $V_j \to \bigoplus_{i=1}^n V_i$. This entails that $\t_4= \cobang{\bigoplus_{i=1}^{j-1} V_i} \oplus \id{V_j} \oplus \cobang{\bigoplus_{i=j+1}^{n} V_i}$. Thus $\t= \tape{c};\cobang{\bigoplus_{i=1}^{j-1} V_i} \oplus \id{V_j} \oplus \cobang{\bigoplus_{i=j+1}^{n} V_i}$ which, by the laws of symmetric monoidal categories, is equal to $\cobang{\bigoplus_{i=1}^{j-1} V_i} \oplus \tape{c} \oplus \cobang{\bigoplus_{i=j+1}^{n} V_i}$. Thus $\t$ is in the second shape.
\end{itemize}

Let us consider now the case $\t_1 \neq \id{U}$. For any $j\in 1\dots n$, $\t_4$ can be written as $\codiag{P}^x \oplus \codiag{Q}^{y}$ where 
$x,y\in \mathbb{N}$, $P=\bigoplus_{i=1}^{j} V_i$ and $Q=\bigoplus_{i=j+1}^{n} V_i$
and  $\codiag{P}^x\colon \bigoplus_{k=1}^x P \to P$ and $\codiag{Q}^y \colon  \bigoplus_{k=1}^yQ\to Q$ are the arrows defined in (\ref{eq:codiagn in appendice}).

Now $\t_3$ should be have as codomain $ \bigoplus_{k=1}^x P \oplus \bigoplus_{k=1}^yQ$. Since $\t$ has domain $U$, $\t_1; \t_2$ has type $U \to \bigoplus_{i=1}^z U$ for some $z\in \mathbb{N}$. Since the arrows in $\t_3$ preserve the size of domain and codomains we have that, the domain of $\t_3$ can be written as $ \bigoplus_{k=1}^x(\bigoplus_{i=1}^{j} U) \oplus \bigoplus_{k=1}^y(\bigoplus_{i=j+1}^{n} U)$. This fact entails that $\t_3= \t_3^1 \oplus \t_3^2$ for some $\t_3^1 \colon  \bigoplus_{k=1}^x(\bigoplus_{i=1}^{j} U) \to  \bigoplus_{k=1}^x(\bigoplus_{i=1}^{j} V_j)$ and  $\t_3^2 \colon  \bigoplus_{k=1}^y(\bigoplus_{i=1}^{j} U) \to  \bigoplus_{k=1}^y(\bigoplus_{i=j+1}^{n} V_n)$. 

Now $\t_2$ should be have as codomain $ \bigoplus_{k=1}^x(\bigoplus_{i=1}^{j} U) \oplus \bigoplus_{k=1}^y(\bigoplus_{i=j+1}^{n} U)$. Since the arrows in the shape of $\t_2$ can only decrease the size of codomain, $\t_2$ has domain $ \bigoplus_{k=1}^x(\bigoplus_{i=1}^{x'} U) \oplus \bigoplus_{k=1}^y(\bigoplus_{i=1}^{y'} U)$ for some $x'\geq j$ and $y'\geq n-j+1$.
This fact entails that $\t_2=\t_2^1 \oplus \t_2^2$ for some $\t_2^1 \colon \bigoplus_{k=1}^x(\bigoplus_{i=1}^{x'} U) \to  \bigoplus_{k=1}^x(\bigoplus_{i=1}^{j} U)$ and $\t_2^2 \colon \bigoplus_{k=1}^y(\bigoplus_{i=1}^{y'} U) \to  \bigoplus_{k=1}^y(\bigoplus_{i=j+1}^{n} U)$.

Now $\t_1$ has type $U \to \bigoplus_{k=1}^x(\bigoplus_{i=1}^{x'} U) \oplus \bigoplus_{k=1}^y(\bigoplus_{i=1}^{y'} U)$. Note that the latter object is different from $U$ as $\t_1\neq \id{U}$. By using \eqref{eq:diagp assoc} and \eqref{eq:diagp symmetry}, $\t_1$ can be written in the form $\diagp{}; (\t_1^1\oplus \t_1^2)$ for some $\t_1^1\colon U \to \bigoplus_{k=1}^x(\bigoplus_{i=1}^{x'} U)$ and $\t_1^2\colon U \to \bigoplus_{k=1}^y(\bigoplus_{i=1}^{y'} U)$.

In summary, $\t=\diagp{};(\, (\t_1^1;\t_2^1; \t_3^1; \codiag{P}^x) \oplus (\t_1^2;\t_2^2; \t_3^2; \codiag{Q}^y) \,)$ is in the third form.

\end{proof}
\begin{proof}[Proof of Lemma~\ref{decomposition}]
By iteratively applying Lemma~\ref{lemma:division}. 

%
%
\end{proof}

\subsection{Cancellativity}
In order to prove cancellativity, we first consider the case of arrows of type $U \to V$ where $U,V$ are objects of $\Cat{C}$.
\begin{lemma}\label{lemma:cancellativity11}
For all $r\in (0,1)$, for all $\s,\t \colon U \to V$, 
if $r\cdot \s = r\cdot \t$ then $\s = \t$.
\end{lemma}
\begin{proof}
By Lemma \ref{lemma: arrows A-B subdistributions}, we known $\CatTapeC[U,V]= \subdistr(\Cat{C}[U,V])$. It is well known (see e.g. \cite{sokolova2018termination}) that, for all sets $X$, $\subdistr(X)$ is a cancellative pca.
\end{proof}

Then, we consider the case of arrows of type $U \to \bigoplus_{i=1}^n V_i$ where $U,V_i$ are objects of $\Cat{C}$.
\begin{lemma}\label{lemma:cancellativityhalf}
For all $n \in \mathbb{N}$, for all $r\in (0,1)$, for all $\s,\t \colon U \to \bigoplus_{i=1}^n V_i$, 
if $r\cdot \s = r\cdot \t$ then $\s = \t$.
\end{lemma}
\begin{proof}
We proceed by induction on $n$. 

The base case $n=0$, it is trivial since $\bigoplus_{i=1}^n V_i$ is the final object $\zero$.

For the induction case $n+1$, we exploit Lemma~\ref{lemma:division} and we consider the case where
\begin{equation}\label{eq:st}
\s=\, \diagq{} ; (\s_1 \oplus \s_2)   \qquad \text{ and } \qquad \t =\, \diagp{}; (\t_1 \oplus \t_2)
\end{equation}
for $\s_1, \t_1 \colon U \to \bigoplus_{i=1}^n V_i$ and $\s_2,\t_2 \colon U \to V_{n+1}$. The other cases are trivial.

By Lemma \ref{lemma:basicTC}.3, we have that
\begin{equation}\label{eq:simple}
r\cdot \s=\, \diagq{} ; (r\cdot \s_1 \oplus r\cdot \s_2)   \qquad \text{ and } \qquad r\cdot  \t =\, \diagp{}; (r\cdot \t_1 \oplus r\cdot \t_2)
\end{equation}
Since by hypothesis $r\cdot \s = r \cdot \t$, it holds both  
\[r\cdot \s ; (\id{} \oplus \bang{}) = r \cdot \t ; (\id{} \oplus \bang{}) \qquad \text{ and } \qquad r\cdot \s ; ( \bang{} \oplus \id{}) = r \cdot \t ;( \bang{} \oplus \id{})\text{.}\]
Thus, by \eqref{eq:simple} and Lemma~\ref{lemma:basicTC}.1 and 2, one obtains that
\begin{equation}\label{eq:twoeq}q\cdot r \cdot \s_1 = p \cdot r \cdot \t_1 \qquad \text{ and } \qquad (1-q) \cdot r \cdot \s_2 = (1-p) \cdot r \cdot \t_2\text{.}\end{equation}

There are now three cases: $p=q$, $p>q$ and $p<q$.

If $p=q$, then  $\s_1=\t_1$ by induction hypothesis and  $\s_2=\t_2$ by Lemma~\ref{lemma:cancellativity11}. Then, by \eqref{eq:st}, $\s=\t$.

If $p>q$, we can rewrite \eqref{eq:twoeq} as
\begin{equation*}p\cdot r \cdot \frac{q}{p} \cdot \s_1 = p \cdot r \cdot \t_1 \qquad \text{ and } \qquad (1-q) \cdot r \cdot \s_2 = (1-q) \cdot r \cdot \frac{1-p}{1-q} \cdot \t_2\text{.}\end{equation*}
Thus, by induction hypothesis $ \frac{q}{p} \cdot \s_1 = \t_1$ and, by Lemma~\ref{lemma:cancellativity11}, $\s_2 =\frac{1-p}{1-q} \cdot \t_2$. By \eqref{eq:st} and Lemma~\ref{lemma:fractions}, we have that $\s=\t$.
The case $p<q$ is symmetrical
\end{proof}

We can finally consider the general case.

\begin{proof}[Proof of Lemma~\ref{cancellativity}]
Let $P=\bigoplus_{i=1}^n U_i$. We proceed by induction on $n$. The base case $n=0$, it is trivial since $\bigoplus_{i=1}^n U_i$ is the initial object $\zero$.

For the induction case $n+1$, since $\CatTapeC$ has finite coproducts, one can assume without loss of generality that
\begin{equation}\label{eq:s1s2coproduct}
\s = (\s_1 \oplus \s_2) ; \codiag{Q} \qquad \text{ and } \t = (\t_1 \oplus \t_2) ; \codiag{Q}
\end{equation}
for some $\s_2,\t_2 \colon U_{n+1} \to Q$ and  $\s_1,\t_1\colon P' \to Q$ where $P'=\bigoplus_{i=1}^n U_i$. Note moreover that the coproduct universal property entails that
\begin{equation}\label{eq:ifflocal}
\s=\t \text{ iff } \s_1=\t_1 \text{ and } \s_2=\t_2\text{.}
\end{equation}

By Lemma~\ref{lemma:pcopairing}, we have that
$r \cdot \s = (r\cdot \s_1 ) \oplus (r\cdot  \s_2 ); \codiag{Q}$ and $r \cdot \t = (r\cdot \t_1 ) \oplus (r\cdot  \t_2 ); \codiag{Q} $. Thus by \eqref{eq:ifflocal} we have that
 \[r\cdot \s_1= r\cdot \t_1 \qquad \text{ and }\qquad r\cdot \s_2= r\cdot \t_2 \]
 From the rightmost above and Lemma~\ref{lemma:cancellativityhalf}, we derive that $\s_2=\t_2$. From the leftmost above and induction hypothesis, we derive that $\s_1=\t_1$. By \eqref{eq:ifflocal}, it holds that $\s=\t$.
 \end{proof}


\subsection{Convex Products}
Now, proving $<\t,\s>_{p,1-p}$ is the unique mediating morphism is pretty easy in the case $p\in \{0,1\}$. It follows immediately by the following results.

\begin{lemma}\label{keyLemma001}
Let $\t\colon U \to Q \oplus R$ be an arrow in $\CatTapeC$. 
\begin{enumerate} 
\item If $\t; (\id{Q} \oplus\, \bangp{R}) =\s$
and $\t; (\,\bang{Q} \oplus \id{R}) = \bang{U}; \cobang{R}$ then $\t =\s \oplus \cobang{R}$;
\item If $\t; (\id{Q} \oplus\, \bangp{R}) =\bang{U}; \cobang{Q}$  
and $\t; (\,\bang{Q} \oplus \id{R}) = \s$ then $\t = \cobang{Q} \oplus \s$;
\end{enumerate}
\end{lemma}
\begin{proof}
We prove the first point: the second is symmetrical.

First, observe that $\t$ should be in one of the form of Lemma \ref{lemma:division}. The only challenging case is the third form (iii) $\t=\diagp{U}; (\t_1 \oplus \t_2)$ for some $\t_1\colon U \to Q$ and $\t_2:U \to R$. Since
\begin{align*}
\bang{U}; \cobang{R} &= \t; (\,\bang{Q} \oplus \id{R}) \tag{Hypothesis}\\
&=\, \diagp{U}; (\t_1 \oplus \t_2) ; (\,\bang{Q} \oplus \id{R}) \tag{iii}\\
&=\, \diagp{U}; ((\t_1 ; \bang{Q}) \oplus \t_2) \tag{Symmetric Monoidal Category}\\
&=\, \diagp{U} ; (\bang{U} \oplus \t_2) \tag{\ref{eq:bangp nat}}\\
&= (1-p)\cdot \t_2 \tag{Lemma \ref{lemma:basicTC}.2}
\end{align*}
since also $(1-p)\cdot (\bang{U}; \cobang{R})=\bang{U}; \cobang{R}$, by cancellativity (Lemma~\ref{cancellativity}), we have that $\t_2 = \bang{U}; \cobang{R}$.
Thus
\begin{align*}
\s \oplus \cobang{R} & = (\t; (\id{Q} \oplus\, \bangp{R})) \oplus \cobang{R} \tag{hypothesis}\\
&= (\diagp{U}; (\t_1 \oplus \t_2) ; (\id{Q} \oplus\, \bangp{R})) \oplus \cobang{R} \tag{iii}\\
&= (\diagp{U}; (\t_1 \oplus (\t_2 ; \bang{R}))) ; \oplus \cobang{R} \tag{Symmetric Monoidal Category}\\
&= (\diagp{U}; (\t_1 \oplus ( \bang{U}; \cobang{R} ; \bang{R}))) ; \oplus \cobang{R} \tag{$\t_2 = \bang{U}; \cobang{R}$}\\
&= (\diagp{U}; (\t_1 \oplus  \bang{U}) ) \oplus \cobang{R} \tag{\ref{eq:bangp nat}}\\
&=\, \diagp{U}; (\t_1 \oplus  (\bang{U} ; \cobang{R}))\tag{Symmetric Monoidal Category}\\
&=\, \diagp{U}; (\t_1 \oplus  \t_2)\tag{$\t_2 = \bang{U}; \cobang{R}$}\\
&= \t \tag{iii}\\
\end{align*}
\end{proof}

The lemma above is generalised to arbitrary arrows as follows. The proof substantially relies on the fact that $\CatTapeC$ has coproducts.

\begin{lemma}\label{keyLemma01}
Let  $\t\colon P \to Q \oplus R$ be an arrow in $\CatTapeC$ where $P=\bigoplus_{i=1}^nU_i$. 
\begin{enumerate} 
\item If $\t; (\id{Q} \oplus\, \bangp{R}) =\s$
and $\t; (\,\bang{Q} \oplus \id{R}) = \bang{P}; \cobang{R}$ then $\t =\s \oplus \cobang{R}$;
\item If $\t; (\id{Q} \oplus\, \bangp{R}) =\bang{P}; \cobang{Q}$  
and $\t; (\,\bang{Q} \oplus \id{R}) = \s$ then $\t = \cobang{Q} \oplus \s$;
\end{enumerate}
\end{lemma}
\begin{proof}
    The proof proceeds by induction on $n$. The base case $n=0$ is trivial since $\bigoplus_{i=1}^0 U_i$ is the initial object $\zero$ .

For the induction case $n+1$, since $\CatTapeC$ has finite coproducts, we can assume that $\t= (\t_1 \oplus \t_2) ; \codiag{Q\piu R}$ where $\t_1\colon \bigoplus_{i=1}^nU_i\to Q\piu R$ 
and $\t_2\colon U_{n+1} \to Q\piu R$ and $\s= (\s_1 \oplus \s_2) ; \codiag{Q}$ where $\s_1\colon \bigoplus_{i=1}^nU_i\to Q$ and $\s_2\colon U_{n+1} \to Q$.

Now, if $\t; (\id{Q} \oplus\, \bangp{R}) =\s$ it follows that 
\begin{align} \t_1\piu \t_2 ; \codiag{Q\piu R} ; (\id{Q} \oplus\, \bangp{R}) &= \t_1\piu \t_2 ; (\id{Q} \oplus\, \bangp{R}) \piu (\id{Q} \oplus\, \bangp{R}) ; \codiag{Q} \tag{axiom \ref{eq:codiag nat}}
\end{align}
and hence $\t_1;(\id{Q} \oplus\, \bangp{R}) = \s_1$ and $\t_2;(\id{Q} \oplus\, \bangp{R}) = \s_2$. Similarly, from $\t; (\,\bang{Q} \oplus \id{R}) = \bang{P}; \cobang{R}$, using the axiom \ref{eq:codiag nat} and \ref{eq:bangp coherence}, we have that $\t_1;(\bang{Q} \oplus \id{R}) = \bang{\bigoplus_{i=1}^nU_i}; \cobang{Q}$ and $\t_2;(\bang{Q} \oplus \id{R}) = \bang{U_{n+1}}; \cobang{R}$. Hence, by inductive hypothesis, we have that $\t_1 = \s_1 \oplus \cobang{R}$ and, from Lemma \ref{keyLemma001}, $\t_2 = \s_2 \oplus \cobang{R}$. Now, the statement is obtained as follows
\begin{align}
    \t_1\piu \t_2 ; \codiag{Q\piu R} &= (\s_1 \oplus \cobang{R}) \piu (\s_2 \oplus \cobang{R}) ; \codiag{Q\piu R} \tag*{}\\
    &= (\s_1 \oplus \cobang{R}) \piu (\s_2 \oplus \cobang{R})  ; (\id{Q} \piu \sigma_{R,Q} \piu \id{R}) ; \codiag{Q}\piu \codiag{R} \tag{axiom \ref{eq:codiag coherence}}\\
    &= (\s_1\oplus \s_2) ; \codiag{Q}\piu \cobang{R} \tag{axiom \ref{eq:codiag nat}}\\
    &=\s\piu \cobang{R} \notag
\end{align}

\end{proof}

Proving that $<\t_1,\t_2>_{p,1-p}$ is the convex pairing in the case of $p\in(0,1)$ relies on cancellativity. 
We first prove the case for $\t_1\colon U\to P$ and $\t_2\colon U \to Q$ where $U$ is an object of $\Cat{C}$

\begin{lemma}\label{lemma:comvexproductpreliminary}
Let $\t_1\colon U\to P$ and $\t_2\colon U \to Q$ be arrows of $\CatTapeC$ and $p\in (0,1)$. Then $\diagp{U}; (\t_1 \oplus \t_2) \colon U \to P\oplus Q$ is the unique arrow $\t$ such that
$p\cdot \t_1 = \t; (\id{P} \oplus \bang{Q})$ and $(1-p)\cdot \t_2 = \t; (\bang{P} \oplus \id{Q})$.
\end{lemma}
\begin{proof}
First observe that by Lemma \ref{lemma:basicTC}.1 and 2 we have that $\diagp{U}; (\t_1 \oplus \t_2) ; (\id{P} \oplus \bang{Q}) = p\cdot \t_1$ and that $\diagp{U}; (\t_1 \oplus \t_2) ; (\bang{P} \oplus \id{Q}) = (1-p) \cdot \t_2$.

We now need to prove that $\diagp{U}; (\t_1 \oplus \t_2)$ is the unique arrow with such property.  We prove that if (a) $p\cdot \t_1 = \t; (\id{P} \oplus \bang{Q})$ and (b) $(1-p)\cdot \t_2 = \t; (\bang{P} \oplus \id{Q})$ then $\t = \diagp{U}; (\t_1 \oplus \t_2)$.

Note that $\t$ can be in one of the forms of Lemma \ref{lemma:division}. The only relevant case is the last one: $\t =\, \diagpX{P'}{U} ; (\t_1' \oplus 	\t_2')$. By (a), (b) Lemma \ref{lemma:basicTC}.1 and 2, it holds that 
\begin{center}$p\cdot \t_1 = p' \cdot \t_1'$ and  $(1-p)\cdot \t_2 = (1-p') \cdot \t_2'$. \end{center}

There are now three cases:
\begin{itemize}
\item If $p=p'$, then by Lemma \ref{cancellativity}, $\t_1=\t_1'$ and $\t_2=\t_2'$. Thus, $\t = \diagp{U}; (\t_1 \oplus \t_2) $.
\item If $p>p'$, then we can rewrite the above equalities as  $p\cdot \t_1 = p \cdot \frac{p'}{p} \cdot \t_1'$ and $(1-p')\frac{(1-p)}{1-p'}\cdot \t_2 = (1-p') \cdot \t_2'$. Thus by Lemma \ref{cancellativity}, we have that $\t_1=\frac{p'}{p}\t_1'$ and $\frac{1-p}{1-p'} \cdot \t_2 =\t_2'$. Thus, by Lemma \ref{lemma:fractions}, $\t= \diagp{U}; (\t_1 \oplus \t_2)$.
\item The case $p<p'$ is symmetrical to the one above.
\end{itemize}
\end{proof}

The lemma above is generalised to arbitrary arrows as follows. The proof substantially relies on the fact that $\CatTapeC$ has coproducts.

\begin{lemma}\label{lemma:comvexproductalmost}
Let $\t_1\colon \bigoplus_{i=1}^nU_i\to P$ and $\t_2\colon \bigoplus_{i=1}^nU_i \to Q$ be arrows of $\CatTapeC$ and $p\in (0,1)$. Then there exists a unique arrow $\t \colon \bigoplus_{i=1}^n U_i \to P\oplus Q$ such that
$p\cdot \t_1 = \t; (\id{P} \oplus \bang{Q})$ and $(1-p)\cdot \t_2 = \t; (\bang{P} \oplus \id{Q})$.
\end{lemma}
\begin{proof}
The proof proceed by induction on $n$. The case $n=0$ is trivial since $\bigoplus_{i=1}^0U_i$ is the initial object $\zero$.

We consider the case $n+1$. Since $\CatTapeC$ has coproducts then 
\begin{equation} \label{eq:local2}
\t_1 = (\s_1^1 \oplus \s_1^2) ; \codiag{P} \qquad \text{ and } \qquad \t_2 = (\s_2^1 \oplus \s_2^2) ; \codiag{Q}
\end{equation}
for $\s_1^1 \colon  \bigoplus_{i=1}^nU_i \to P$, $\s_1^2\colon U_{n+1} \to P$, $\s_2^1 \colon  \bigoplus_{i=1}^nU_i \to Q$, $\s_2^2\colon U_{n+1} \to Q$.

By induction hypothesis there exists a unique arrow $\t \colon  \bigoplus_{i=1}^nU_i \to P \oplus Q$ such that
\begin{equation}\label{eq:local3}
\t; (\id{P} \oplus \bang{Q}) = p\cdot \s_1^1 \qquad \text{ and } \qquad \t; (\bang{P} \oplus \id{Q}) = (1-p)\cdot \s_2^1
\end{equation}

We now focus on  $\t \oplus ( \diagp{U_{n+1}} ; (\s_1^2 \oplus \s_2^2) ); \codiag{P\oplus Q} \colon \bigoplus_{i=1}^{n+1} U_i \to P\oplus Q$.
\begin{align*}
 & \t \oplus ( \diagp{U_{n+1}} ; (\s_1^2 \oplus \s_2^2) ); \codiag{P\oplus Q} ; (\id{P} \oplus \bang{Q}) \\& = ( \, (\t; (\id{P} \oplus \bang{Q})) \oplus (\diagp{U_{n+1}} ; (\s_1^2 \oplus \s_2^2) ; (\id{P} \oplus \bang{Q})) \,) ; \codiag{P\oplus Q} \tag{nat. of $\codiag{P\oplus Q}$}\\
&= ( \, (\t; (\id{P} \oplus \bang{Q})) \oplus (\diagp{U_{n+1}} ; (\s_1^2 \oplus \bang{U_{n+1}})) \,) ; \codiag{P\oplus Q} \tag{nat. of $\bang{}$}\\
&= ( \, (\t; (\id{P} \oplus \bang{Q})) \oplus p\cdot \s_1^2  \,) ; \codiag{P\oplus Q} \tag{Lemma \ref{lemma:pcdot}}\\
&= ( \, p\cdot \s_1^1 \oplus p\cdot \s_1^2  \,) ; \codiag{P\oplus Q} \tag{\ref{eq:local3}}\\
&=  p\cdot  ( \, ( \s_1^1 \oplus \s_1^2 ) ; \codiag{P\oplus Q} \, ) \tag{Lemma \ref{lemma:pcopairing}}\\
&=  p\cdot  \t_1 \tag{\ref{eq:local2}}\\
\end{align*}
Similar computations proves that $\t \oplus ( \diagp{U_{n+1}} ; (\s_1^2 \oplus \s_2^2) ); \codiag{P\oplus Q} ; (\bang{P} \oplus \id{Q}) = (1-p) \cdot \t_2$.

Now we need to prove that $\t \oplus ( \diagp{U_{n+1}} ; (\s_1^2 \oplus \s_2^2) ); \codiag{P\oplus Q}$ is the unique arrows with such property.
Assume that $\t'\colon \bigoplus_{i=1}^{n+1} U_i \to P\oplus Q$ is such that
\begin{equation}\label{eq:local4}
\t'; (\id{P} \oplus \bang{Q}) = p\cdot  \t_1 \qquad \text{ and }\qquad \t'; (\bang{P} \oplus \id{Q}) = (1-p)\cdot  \t_2
\end{equation}
Since $\CatTapeC$ has coproducts $\t' = (\t_1' \oplus \t_2' ); \codiag{P\oplus Q}$ for some $\t_1'\colon \bigoplus_{i=1}^{n} U_i \to P\oplus Q $ and $\t_2'\colon U_{n+1} \to P \oplus Q$.
 By using naturality of $\codiag{P\oplus Q}$ one readily has
 \begin{equation}\label{eq:local5}
 \t' ; (\id{P} \oplus \bang{Q}) = (\, (\t_1'; (\id{P} \oplus \bang{Q})) \oplus (\t_2' ; (\id{P} \oplus \bang{Q})) \,); \codiag{P\oplus Q}\text{.}
 \end{equation}
 Now, it holds that
 \begin{align*}
  (\, (\t_1'; (\id{P} \oplus \bang{Q})) \oplus (\t_2' ; (\id{P} \oplus \bang{Q})) \,); \codiag{P\oplus Q} & = \t' ; (\id{P} \oplus \bang{Q}) \tag{\ref{eq:local5}} \\
 &= p\cdot  \t_1 \tag{\ref{eq:local4}}\\
 &= ( \, p\cdot \s_1^1 \oplus p\cdot \s_1^2  \,) ; \codiag{P\oplus Q} \tag{derivation above}
 \end{align*}
 and thus, by the universal property of coproducts, we can conclude that
 \begin{equation}\label{eq:local6}
  \t_1'; (\id{P} \oplus \bang{Q}) = p\cdot \s_1^1 \qquad \text{ and } \qquad \t_2' ; (\id{P} \oplus \bang{Q}) = p\cdot \s_1^2\text{.}
 \end{equation}
 With the symmetric argument, one obtains that
 \begin{equation}\label{eq:local7}
  \t_1'; (\bang{P} \oplus \id{Q}) = (1-p)\cdot \s_2^1 \qquad \text{ and } \qquad \t_2' ; (\bang{P} \oplus \id{Q}) = (1-p)\cdot \s_2^2\text{.}
 \end{equation}
Now, by the leftmost equations in \eqref{eq:local6} and \eqref{eq:local7} and the fact that $\t$ is the unique satisfying \eqref{eq:local3}, we obtain that $\t_1' = t$.

By Lemma~\ref{lemma:comvexproductpreliminary}, $\diagp{U_{n+1}} ; (\s_1^2 \oplus \s_2^2)$ is the unique arrow $h\colon  U_{n+1} \to P \oplus Q $ such that $h; (\id{P} \oplus \bang{Q}) = p\cdot \s_1^2$ and  $h; (\bang{P} \oplus \id{Q}) = (1-p)\cdot \s_2^2$. Thus, the rightmost equations in \eqref{eq:local6} and \eqref{eq:local7} entails that $\t_2'= \diagp{U_{n+1}} ; (\s_1^2 \oplus \s_2^2)$. This proves that $\t'= \t \oplus ( \diagp{U_{n+1}} ; (\s_1^2 \oplus \s_2^2) ); \codiag{P\oplus Q}$.
\end{proof}

We can finally state our main result.
\begin{proposition}\label{prop:TChasconvexcoprduct}
$\CatTapeC$ has convex products.
\end{proposition}
\begin{proof}
We need to prove that for all arrows $\t_1\colon R\to P$ and $\t_2\colon R \to Q$ and for all $p,q\in[0,1]$ such that $p+q\leq 1$, there exists a unique $\t\colon R \to P\oplus Q$ such that $\t;\pi_1 =p\cdot \t_1$ and $\t;\pi_2 = q \cdot \t_2$.

The case when either $p=1$ or $q=1$ is proved by Lemma~\ref{keyLemma01}.
The case when $p,q\in (0,1)$ and $p+q=1$ is proved by Lemma~\ref{lemma:comvexproductalmost}.
The case when $p,q\in (0,1)$ and $p+q<1$ easily follows from the previous case: if $p\geq 1-q$ one can take $\langle f,g \rangle_{p,q} = \langle f, \frac{q}{1-p} \cdot g \rangle_{p,1-p}$.
If $p\leq 1-q$ one can take $\langle f,g \rangle_{p,q} = \langle \frac{p}{1-q} \cdot f, g \rangle_{1-q,q}$.
\end{proof}

\subsection{Proof of Proposition \ref{cor: quotient category is convex biproduct }}
We conclude this appendix with a proof of Proposition  \ref{cor: quotient category is convex biproduct } that turns out to be useful when considering diagrammatic languages (see Section \ref{sec:probboolcircuits}). First, we need the following result.

\begin{proposition}\label{prop: convex products quoziente}
    Given a category $\Cat{C}$ and a congruence relation $\sim$ on $\CatTapeC$, then the functor $Q_\sim\colon\CatTapeC\to \CatTapeC_\sim$ satisfies the following property: if $[p_1\cdot f_1]_\sim =[q_1\cdot g_1]_\sim$ and $[p_2\cdot f_2]_\sim =[q_2\cdot g_2]_\sim$  then $[\langle f_1,f_2\rangle_{(p_1,p_2)}]_\sim= [\langle g_1,g_2\rangle_{(q_1,q_2)}]_\sim$.
\end{proposition}
\begin{proof}\label{proof-prop: convex products quoziente}
    Before proceeding with the proof, we make the trivial observation that if $[f]_\sim=[g]_\sim$ then $[p\cdot f]_\sim=[p\cdot g]_\sim$ for all $p\in [0,1]$. Indeed,
\begin{align}
    p\cdot f&=\, \diagpX{p}{}; (f\piu\, \bangp{}) \tag{Thm. \ref{thm:TCconvexbiproductcategory}}\\
    &\sim \, \diagpX{p}{}; (g\piu\, \bangp{}) \tag{by $[f]_\sim=[g]_\sim$}\\
    &=\, p\cdot g \notag
\end{align}
hence $[p\cdot f]_\sim=[p\cdot g]_\sim$.
The cases with $p=0$ and $p=1$ are trivial.
  Now, define  $r_i\defeq \begin{cases}f_i &\text{if }p_i\leq q_i\\ g_i & \text{else} \end{cases}$, $u_i\defeq \begin{cases}p_i &\text{if }p_i\leq q_i\\ q_i & \text{else} \end{cases}$ and consider $\langle r_1,r_2\rangle_{(u_1,u_2)}$. We prove that $[\langle f_1,f_2\rangle_{(p_1,p_2)}]_\sim=[\langle r_1,r_2\rangle_{(u_1,u_2)}]_\sim=[\langle g_1,g_2\rangle_{(q_1,q_2)}]_\sim$. Without loss of generality, we can assume that $p_1\leq q_1$  and $q_2\leq p_2$, hence $\langle r_1,r_2\rangle_{(u_1,u_2)}=\langle f_1,g_2\rangle_{(p_1,q_2)}$.  The first equality is obtained as follows:
  \begin{align*}
        \langle f_1,f_2\rangle_{(p_1,p_2)}&=\, \diagpX{p_1}{}; (f_1\piu\, (\diagpX{\frac{p_2}{1-p_1}}{};(f_2\piu\, \bangp{})) ) \tag{Thm. \ref{thm:TCconvexbiproductcategory}}\\
        &\sim \, \diagpX{p_1}{}; (f_1\piu\, (\diagpX{\frac{q_2}{1-p_1}}{};(g_2\piu\, \bangp{})) ) \tag{by $[p_2\cdot f_2]_\sim =[q_2\cdot g_2]_\sim$}\\
        &=\, \langle f_1,g_2\rangle_{(p_1,q_2)}
  \end{align*}

hence $[\langle f_1,f_2\rangle_{(p_1,p_2)}]_\sim=[\langle f_1,g_2\rangle_{(p_1,q_2)}]_\sim$. The second equality is obtained as follows:
\begin{align*}
    \langle g_1,g_2\rangle_{(q_1,q_2)}&=\, \diagpX{q_1}{}; (g_1\piu\, (\diagpX{\frac{q_2}{1-q_1}}{};(g_2\piu\, \bangp{})) ) \tag{Thm. \ref{thm:TCconvexbiproductcategory}}\\
    &= \, \diagpX{q_2}{};(g_2\piu\, (\diagpX{\frac{q_1}{1-q_2}}{};(g_1\piu\, \bangp{})) ) \tag{\ref{eq:diagp assoc} + \ref{eq:diagp symmetry}}\\
    &\sim \, \diagpX{q_2}{};(g_2\piu\, (\diagpX{\frac{p_1}{1-q_2}}{};(f_1\piu\, \bangp{})) ) \tag{by $[p_1\cdot f_1]_\sim =[q_1\cdot g_1]_\sim$}\\
    &=  \diagpX{p_1}{}; (f_1\piu\, (\diagpX{\frac{q_2}{1-p_1}}{};(g_2\piu\, \bangp{})) ) \tag{\ref{eq:diagp assoc} + \ref{eq:diagp symmetry}}     \\
    &=\, \langle f_1,g_2\rangle_{(p_1,q_2)}
\end{align*}
hence $[\langle g_1,g_2\rangle_{(q_1,q_2)}]_\sim=[\langle f_1,g_2\rangle_{(p_1,q_2)}]_\sim$.
\end{proof}

\begin{proof}[Proof of Proposition \ref{cor: quotient category is convex biproduct }]\label{proof-cor: quotient category is convex biproduct }
    Since $\sim$ is a congruence relation, $\CatTapeC_\sim$ is a monoidal category and $Q_\sim$ is a monoidal functor. The object $0$ is initial and terminal in $\CatTapeC_\sim$ since it is so in $\CatTapeC$. The diagram  $A\overset{[\iota_1^{A\piu B}]_\sim}{\rightarrow}A\piu B\overset{[\iota_2^{A\piu B}]_\sim}{\leftarrow} B$ is the coproduct of $A$ and $B$ since for every pair of arrows $[f]_\sim\colon A\to X$ and $[g]_\sim:B\to X$ there exists the arrow $[(f\piu g);\codiag{X}]_\sim$ that gives $[\iota_1^{A\piu B}]_\sim;[(f\piu g);\codiag{X}]_\sim=[f]_\sim$ and $[\iota_B]_\sim;[(f\piu g);\codiag{X}]_\sim=[g]_\sim$. Such arrow is the unique satisfying this property, indeed if $[h]\colon A\piu B\to X$ satisfies $[\iota_1^{A\piu B}]_\sim;[h]_\sim=[f]_\sim$ and $[\iota_2^{A\piu B}]_\sim;[h]_\sim=[g]_\sim$, then coproducts in $\CatTapeC$ imply that there exist $f',g'$ such that $h\coloneqq [f',g'];\codiag{X}$ and then
    \begin{align*}
[\iota_1^{A\piu B}]_\sim;[h]_\sim&=[\iota_1^{A\piu B}]_\sim;[[f',g'];\codiag{X}]_\sim \tag{$h\coloneqq [f',g'];\codiag{X}$}\\
&=[f']_\sim \tag{$\sim$ congruence relation}\\
&=[f]_\sim \tag{by hp.}
        \end{align*}
        similarly $[\iota_2^{A\piu B}]_\sim;[h]_\sim=[g']_\sim=[g]_\sim$. Hence, we obtain
        \begin{align*}
            [h]_\sim&=[[f',g'];\codiag{X}]_\sim \tag*{}\\
            &=[f']_\sim\piu [g']_\sim;[\codiag{X}]_\sim \tag{$\sim$ congruence relation}\\
            &=[f]_\sim\piu [g]_\sim;[\codiag{X}]_\sim \tag{above equalities}\\
            &=[(f\piu g);\codiag{X}]_\sim \tag{$\sim$ congruence relation}
        \end{align*}
    $\CatTapeC_\sim$ is also enriched over convex algebras. Indeed, for all $X,Y\in \CatTapeC_\sim$, one defines $[f]_\sim +_p [g]_\sim \defeq [f +_p g]_\sim$ for all $[f]_\sim,[g]_\sim \in \CatTapeC_\sim(X,Y)$ and $p\in [0,1]$. This is well defined since if $[f]_\sim=[f']_\sim$ and $[g]_\sim=[g']_\sim$ then
    \begin{align*}
        [f +_p g]_\sim &= [\diagpX{p}{X}; (f\piu g);\codiag{Y}]_\sim \tag{Thm.~\ref{thm:TCconvexbiproductcategory}}\\
        &=[\diagpX{p}{X}]_\sim;( [f]_\sim\piu [g]_\sim);[\codiag{Y}]_\sim \tag{$\sim$ congruence relation}\\
        &=[\diagpX{p}{X}]_\sim;( [f']_\sim\piu [g']_\sim);[\codiag{Y}]_\sim \tag{by $[f]_\sim=[f']_\sim$ and $[g]_\sim=[g']_\sim$}\\
        &=[f' +_p g']_\sim
        \end{align*}
     
While $\star_{X,Y}$ is just $[\star_{X,Y}]_\sim$. Axioms of PCA are satisfied since they are so in $\CatTapeC$ and $\sim$ is a congruence relation. Now we prove that $A\overset{[\pi_1]_\sim}{\leftarrow}A\piu B\overset{[\pi_2]_\sim}{\rightarrow} B$ is a binary convex product also in $\CatTapeC$. Indeed, given $[f]_\sim \colon X\to A$ and $[g]_\sim \colon X\to B$, $p_1,p_2\in [0,1]$ such that $p_1+p_2\leq 1$, we prove that there is a unique $[h]_\sim \colon X\to A\oplus B$ such that $[h]_\sim;[\pi_1]_\sim = p_1\cdot [f]_\sim$ and $[h]_\sim;[\pi_2]_\sim = p_2\cdot [g]_\sim$. Existence is provided by $h\defeq \langle f,g\rangle_{(p_1,p_2)}$, for uniqueness, suppose that there is also $[h']_\sim$ with the same property. Then, by Theorem~\ref{thm:TCconvexbiproductcategory}, $h'=\langle f',g'\rangle_{(q_1,q_2)}$ for some $f'\colon X\to A$, $g'\colon X\to B$ and $q_1,q_2\in [0,1]$ such that $q_1+q_2\leq 1$. Now, since $[h']_\sim;[\pi_1]_\sim = [q_1\cdot f']_\sim=[p_1\cdot f]_\sim$ and $[h']_\sim;[\pi_2]_\sim = [q_2\cdot g' ]_\sim = p_2\cdot [g]_\sim$, by Proposition~\ref{prop: convex products quoziente} we have that $[h']_\sim=[\langle f,g\rangle_{(p_1,p_2)}]_\sim$. Hence, $\langle [f],[g]\rangle_{(p_1,p_2)}\coloneqq[\langle f,g\rangle_{(p_1,p_2)}]$ is the unique arrow with the required property. Observe also that, by definition, $[f]_\sim +_p [g]_\sim = [\diagp{}]; ([f]_\sim \piu [g]_\sim); [\codiag{}]_\sim$. The fact that $Q_\sim$ is a convex biproduct functor because is PCA-enriched and preserves coproducts.
 \end{proof}

%% file: appendices/appprobbooltape.tex
\section{Appendix to Section \ref{sec:probboolcircuits}}\label{app:sec:probboolcircuits}

\begin{proof}[Proof of Lemma \ref{lemma:booleans}]
By induction on $n\in \mathbb{N}$ we prove that
\begin{center}
for all $n,m\in\mathbb{N}$, for all tapes $\s,\t\colon A^n \to A^m$, if  for all $\bool\in \Cat{B}[1,A^n]$, $\bool;\s \eqsyn\bool; \t$, then $\s \eqsyn \t$.
\end{center}

For the base case $n=0$, there is only one $\bool\in \Cat{B}[1,A^0]$, which is $\id{\uno}$. Thus it is trivial.

Consider now the case $n+1$. Take $\s,\t\colon A^{n+1} \to A^m$ be two tapes such that for all $\bool\in \Cat{B}[1,A^{n+1}]$, 
\begin{equation}\label{eq:blablabla}\bool;\s \eqsyn\bool; \t\text{.}\end{equation} 
We have that for all $\bool' \in \Cat{B}[1,A^n] $
\begin{align*}
\bool'; \s_0  &=   \bool';( \Flip{0}[t]\otimes \id{A^n} ); \s  \tag{def. of $\s_0$}\\
&= (\Flip{0}[t]\otimes \bool'); \s \tag{Symmetric Monoidal Category}\\
&\eqsyn (\Flip{0}[t]\otimes \bool'); \t \tag{\ref{eq:blablabla}}\\
&= \bool';( \Flip{0}[t]\otimes \id{A^n} ); \t \tag{Symmetric Monoidal Category}\\
& =  \bool';\t_0   \tag{def. of $\t_0$}
\end{align*}
Since $\s_0,\t_0$ have type $A^n \to A^m$ we can apply our induction hypothesis to derive that $\s_0 \eqsyn \t_0$. With the symmetric argument, we have that $\s_1 \eqsyn \t_1$.
Now, it is enough to use the encoding of axiom (\ref{eq:natmultiplexer}) for tapes 
\begin{align*}
\s &\eqsyn(\, \id{A} \otimes (\ncopier[]; \s_1 \otimes \s_0)\,); \Ifgatem[] \tag{ax. (\ref{eq:natmultiplexer})} \\
& \eqsyn  (\, \id{A} \otimes (\ncopier[]; \t_1 \otimes \t_0)\,); \Ifgatem[] \tag{$\s_0 \eqsyn \t_0$, $\s_1 \eqsyn \t_1$}\\
& \eqsyn \t \tag{ax. (\ref{eq:natmultiplexer})}
\end{align*}
\end{proof}